\documentclass[aps,superscriptaddress,11pt,twoside]{revtex4}

\usepackage{amsmath,latexsym,amssymb,verbatim,enumerate,graphicx,color}

\usepackage[dvipsnames]{xcolor}

\usepackage{amsfonts}
\usepackage{mathtools}
\usepackage{hyperref}
\usepackage{bm}
\usepackage{standalone}

\usepackage{orcidlink}

\usepackage{theorem}
\newtheorem{definition}{Definition}[section]
\newtheorem{proposition}[definition]{Proposition}

\newtheorem{theorem}[definition]{Theorem}
\newtheorem{corollary}[definition]{Corollary}

\newtheorem{remark}[definition]{Remark}

\def\squareforqed{\hbox{\rlap{$\sqcap$}$\sqcup$}}
\def\qed{\ifmmode\squareforqed\else{\unskip\nobreak\hfil
\penalty50\hskip1em\null\nobreak\hfil\squareforqed
\parfillskip=0pt\finalhyphendemerits=0\endgraf}\fi}
\def\endenv{\ifmmode\;\else{\unskip\nobreak\hfil
\penalty50\hskip1em\null\nobreak\hfil\;
\parfillskip=0pt\finalhyphendemerits=0\endgraf}\fi}
\newenvironment{proof}{\noindent \textbf{{Proof~} }}{\qed}

\mathchardef\ordinarycolon\mathcode`\:
\mathcode`\:=\string"8000
\def\vcentcolon{\mathrel{\mathop\ordinarycolon}}
\begingroup \catcode`\:=\active
  \lowercase{\endgroup
  \let :\vcentcolon
  }

\newcommand{\nc}{\newcommand}
\nc{\rnc}{\renewcommand} 
\nc{\beq}{\begin{equation}}
\nc{\eeq}{{\end{equation}}} 
\nc{\bea}{\begin{eqnarray}}
\nc{\eea}{\end{eqnarray}} 
\nc{\beqa}{\begin{eqnarray}}
\nc{\eeqa}{\end{eqnarray}} 
\nc{\lbar}[1]{\overline{#1}}
\nc{\bra}[1]{\langle#1|} 
\nc{\ket}[1]{|#1\rangle}
\nc{\ketbra}[2]{|#1\rangle\!\langle#2|}
\nc{\braket}[2]{\langle#1|#2\rangle} 
\nc{\proj}[1]{|#1\rangle\!\langle #1 |} 
\nc{\avg}[1]{\langle#1\rangle}
\rnc{\max}{\operatorname{max}} 
\nc{\rank}{\operatorname{rank}}
\nc{\conv}{\operatorname{conv}}
\nc{\smfrac}[2]{\mbox{$\frac{#1}{#2}$}}
\nc{\Tr}{\operatorname{Tr}}
\nc{\ox}{\otimes}
\nc{\dg}{\dagger}
\nc{\dn}{\downarrow} \nc{\cA}{{\cal A}} \nc{\cB}{{\cal B}}
\nc{\cC}{{\cal C}} \nc{\cD}{{\cal D}} \nc{\cE}{{\cal E}}
\nc{\cF}{{\cal F}} \nc{\cG}{{\cal G}} \nc{\cH}{{\cal H}}
\nc{\cI}{{\cal I}} \nc{\cJ}{{\cal J}} \nc{\cK}{{\cal K}}
\nc{\cL}{{\cal L}} \nc{\cM}{{\cal M}} \nc{\cN}{{\cal N}}
\nc{\cO}{{\cal O}} \nc{\cP}{{\cal P}} \nc{\cR}{{\cal R}}
\nc{\cS}{{\cal S}} \nc{\cT}{{\cal T}} \nc{\cU}{{\cal U}}
\nc{\cX}{{\cal X}} \nc{\cW}{{\cal W}} \nc{\cZ}{{\cal Z}}
\nc{\csupp}{{\operatorname{csupp}}}
\nc{\qsupp}{{\operatorname{qsupp}}} \nc{\var}{\operatorname{var}}
\nc{\rar}{\rightarrow} \nc{\lrar}{\longrightarrow}
\nc{\poly}{\operatorname{poly}}
\nc{\polylog}{\operatorname{polylog}}
\nc{\Lip}{\operatorname{Lip}} \nc{\1}{\openone}
\nc{\diag}{\operatorname{diag}}

\def\>{\rangle}
\def\<{\langle}

\def\a{\alpha}
\def\b{\beta}

\def\d{\delta}

\nc{\glneq}{{\raisebox{0.6ex}{$>$}  \hspace*{-1.8ex} \raisebox{-0.6ex}{$<$}}}
\nc{\gleq}{{\raisebox{0.6ex}{$\geq$}\hspace*{-1.8ex} \raisebox{-0.6ex}{$\leq$}}}

\nc{\RR}{{{\mathbb R}}}
\nc{\CC}{{{\mathbb C}}}
\nc{\FF}{{{\mathbb F}}}
\nc{\HH}{{{\mathbb H}}}
\nc{\NN}{{{\mathbb N}}}
\nc{\ZZ}{{{\mathbb Z}}}
\nc{\PP}{{{\mathbb P}}}
\nc{\QQ}{{{\mathbb Q}}}
\nc{\UU}{{{\mathbb U}}}
\nc{\WW}{{{\mathbb W}}}
\nc{\EE}{{{\mathbb E}}}
\rnc{\SS}{{{\mathbb S}}}

\nc{\id}{{\operatorname{id}}}

\nc{\vholder}[1]{\rule{0pt}{#1}}

\nc{\ob}[1]{#1}

\def\beq{\begin {equation}}
\def\eeq{\end {equation}}

\nc{\eq}[1]{Eq.~(\ref{eq:#1})} \nc{\eqs}[2]{Eqs.~(\ref{eq:#1}) and
(\ref{eq:#2})}

\nc{\eqn}[1]{Eq.~(\ref{eqn:#1})}
\nc{\eqns}[2]{Eqs.~(\ref{eqn:#1}) and (\ref{eqn:#2})}

\nc{\region}{\cS\cW}

\newcommand{\zed}{\mathbb Z}

\newcommand{\cee}{\mathbb C}

\newcommand{\Ab}{\text{Ab}}

\begin{document}

\title{An approach to $\boldsymbol{p}$-adic qubits from irreducible representations of $\boldsymbol{SO(3)_p}$}

\author{Ilaria Svampa\,\orcidlink{0000-0002-1389-0319}}
 \email{ilaria.svampa@unicam.it}
  \affiliation{School of Science and Technology, University of Camerino,
              Via Madonna delle Carceri 9, I-62032 Camerino, Italy}
  \affiliation{INFN--Sezione Perugia, Via A. Pascoli, I-06123 Perugia, Italy}
  \affiliation{Departament de F\'isica: Grup d'Informaci\'o Qu\`antica,
              Universitat Aut\`onoma de Barcelona, ES-08193 Bellaterra (Barcelona), Spain}
 
\author{Stefano Mancini\,\orcidlink{0000-0002-3797-3987}}
 \email{stefano.mancini@unicam.it}
 \affiliation{School of Science and Technology, University of Camerino,
              Via Madonna delle Carceri 9, I-62032 Camerino, Italy}
 \affiliation{INFN--Sezione Perugia, Via A. Pascoli, I-06123 Perugia, Italy}
  
\author{Andreas Winter\,\orcidlink{0000-0001-6344-4870}}
 \email{andreas.winter@uab.cat}
 \affiliation{Departament de F\'isica: Grup d'Informaci\'o Qu\`antica,
              Universitat Aut\`onoma de Barcelona, ES-08193 Bellaterra (Barcelona), Spain}
 \affiliation{ICREA---Instituci\'o Catalana de la Recerca i Estudis Avan\c{c}ats, 
              Pg. Llu\'is Companys, 23, ES-08001 Barcelona, Spain}

\date{13 July 2022}

\begin{abstract}
 We introduce the notion of $p$-adic quantum bit ($p$-qubit) 
 in the context of the $p$-adic quantum mechanics initiated and 
 developed after the seminal paper of Volovich [Theor. Math. Phys. \textbf{71}, 574 (1987)]. In this approach, physics takes place in 
 three-dimensional $p$-adic space rather than Euclidean space. Based on 
 our prior work describing the $p$-adic special orthogonal group [Di Martino {\it et al.}, arXiv:2104.06228 [math.NT] (2021)], 
 we outline a program to classify its continuous unitary projective representations, 
 which can be interpreted as a theory of $p$-adic angular momentum. 
 The $p$-adic quantum bit arises from the irreducible representations of 
 minimal nontrivial dimension two, of which we construct examples for all 
 primes $p$.
\end{abstract}

\maketitle


\section{Introduction}
\label{sec:intro}
A qubit at its most basic is the state space of a two-level quantum system, i.e., one with two-dimensional complex Hilbert space $\cH$. This abstracts from all 
physical characteristics such a system might have, be it an electron spin, 
or that of a neutron, the polarization of a photon, a quantum dot, etc. 
Its characteristic is that the symmetry group $SU(2)$ acts transitively on 
all pure states
\[
  \cP(\cH) = \{\rho=\proj{\psi}\,:\,\rho\geq 0,\Tr\rho=1,\,\rho \text{ rank-one}\}, 
\]
which geometrically is the surface of the Bloch sphere, $\RR^3$, and the
conjugation action on the density matrices is equivalent to the action 
of $SO(3)$ on the Bloch sphere. Furthermore, for a spin-$\frac12$ particle, 
this action coincides with the one induced by the spatial rotations, 
and as an irreducible representation (irrep) gives rise to all other irreps 
by tensor powers $U^{\otimes n}$ and their Clebsch-Gordan decomposition 
into irreps. This has been long noticed and has even be taken as a possible 
foundation of quantum mechanics \cite{Weizsaecker}. 

However, this link of the most basic quantum system to the geometry 
and, indeed, symmetry of physical space \cite{Klein} opens up other 
possibilities. For example, why should we accept {\it a priori} that physical space 
is described by the Euclidean $\RR^3$? Indeed, the real field is not the only possible metric completion of the rationals: by Ostrowski's Theorem \cite{Serre}, 
there is a countable family of other completions, those of the $p$-adic 
numbers $\QQ_p$. The fractal geometry of the $p$-adic numbers made them 
interesting for attempts at ``new'' quantum physics with a space-time 
that is decidedly non-Euclidean at very short or very large distances, 
as seems possible in quantum gravity \cite{Vol1st, volovich10}. In these approaches, 
the focus is on the translation symmetry of phase space, which in ordinary 
quantum mechanics leads to the framework of the Heisenberg-Weyl group of 
position and momentum displacements, parametrized by real numbers, but 
which here are displacements by $p$-adic amounts \cite{volovich89}.
 
Here, we consider the rotation symmetry of $\mathbb{Q}^3_p$, encoded 
in the special orthogonal group $SO(3)_p$ in dimension three. 
Recently, we have characterized the geometry of this group \cite{our1st} based on the crucial observation that, up to equivalence, there is a 
unique quadratic form on $\QQ_p^3$ with a compact symmetry group.
Motivated by the analogy with the real case, our program is to find 
its unitary complex projective representations, which should constitute a 
$p$-adic theory of quantum angular momentum. In particular, the irreps of 
dimension two are crucial as $p$-adic analogs of spin-$\frac12$ particles,
which we interpret as $p$-adic qubits. In this spirit, we call any pair 
$(\cH,\psi)$, where $\cH$ is a two-dimensional Hilbert space and $\psi$ 
is a projective irrep on $\cH$, a qubit. In the present paper, we do 
not carry out this entire program, leaving the bulk of it for future work,
but contend ourselves with the construction of the simplest possible 
two-dimensional irreps of $SO(3)_p$ for every prime $p$. 
We emphasize that the angular momentum representations that we are looking for 
map the $p$-adic group $SO(3)_p$ to groups of linear transformations on complex 
vector spaces. Taking this route, moving from $p$-adic numbers to complex numbers, 
is different and alternative to the ones that can be considered more conventional, such as (i) defining $p$-adic Hilbert spaces and working on them 
\cite{Hilbert1,Hilbert2,Hilbert3,Hilbert4,Hilbert5} and (ii) defining $p$-adic Lie groups \cite{Lie}.


\medskip
The structure of the rest of the paper is as follows. 
In Sec. \ref{sec:prelim}, we recall the basic definitions and 
facts about the three-dimensional special orthogonal groups $SO(3)_p$, 
among them the crucial insight that the coefficients of special 
orthogonal matrices are $p$-adic integers. 
In Sec. \ref{sec:quot}, 
we exploit the latter property to define a hierarchy of projected groups
by reducing the matrix entries modulo $p^k$, $k\geq 1$; each of 
them is a finite group, offering an easy way to constructing 
irreps.
In Sec. \ref{pi1SO3}, we turn our attention, in particular, to the 
smallest reduction, $SO(3)_p \mod p$, for odd primes $p>2$; 
we show that each has a homomorphism to some dihedral group that is onto. 
This implies, in particular, that they all have two-dimensional irreps. 
In the subsequent sections, we explicitly compute the groups $SO(3)_3\mod 3$
(Sec. \ref{sec:3-mod-3}) and $SO(3)_5\mod 5$ (Sec. \ref{sec:5-mod-5})
and their representation theory in both cases exhibiting qubits, i.e., 
two-dimensional irreps. Finally, in Sec. \ref{sec:2-mod-2},
we do the same for $SO(3)_2\mod 2$, which turns out to be isomorphic to
$\SS_3$, the permutation group of three elements. We conclude in 
Sec. \ref{sec:conclusion} with an outlook to future work and open 
questions.


\section{Preliminaries on the compact group $\boldmath{SO(3)_p}$}
\label{sec:prelim}
We recall here some basic notions and results drawn from Ref. \cite{our1st}.

\begin{theorem}
For every prime $p$,
there is a unique quadratic form on $\QQ_p^3$, up to 
linear equivalence and scaling, which does not represent zero, 
\beq\label{eq:so3formedef}
Q(\boldsymbol{x})
 =\begin{cases}
    x_1^2-vx_2^2+px_3^2 &\text{if } p>2, \\
    x_1^2+x_2^2+x_3^2   &\text{if } p=2,
  \end{cases}
\eeq 
where
\begin{equation}\label{defv}
v\coloneqq
\begin{cases}
 -1 &\text{if } p\equiv3 \mod4,\\
 -u &\text{if } p\equiv1 \mod4,
\end{cases}
\end{equation}
where $u\in\mathbb{U}$ (the group of $p$-adic units, i.e., the group of 
invertible elements of $\mathbb{Z}_p$) is not a square.
\end{theorem}

Recalling that every quadratic form $Q$ induces a special orthogonal group as the set of determinant 1 matrices that preserve $Q$, we have the following result.
\begin{theorem} 
There is a unique (up to isomorphisms) compact special orthogonal group on $\QQ_p^3$ for every prime $p$,
\beq \label{SO3p}
SO(3)_p=\{L\in\mathcal{M}_{3\times3}(\QQ_p)\,:\, A=L^TAL,\,\det L=1\}
\eeq 
endowed with the common matrix product, where $A$ is the matrix representation 
in the canonical basis of $\QQ_p^3$ of the preserved quadratic form $Q$,
\beq \label{eq:matAdefSO3p}
A=\begin{cases}
    \diag(1,-v,p) &\text{if } p>2,\\
    \1            &\text{if } p=2.
  \end{cases}
\eeq
\end{theorem}

Since $SO(3)_p$  is a topological group, it is natural that the representations we look for are continuous [with respect to the $p$-adic metric on the domain and the Lie group metric of the complex $SU(d)$]. All of the representations we find will turn out to be continuous for trivial reasons.
\\ \par Another crucial result proved in Ref. \cite{our1st} is the following.
\begin{theorem}
\label{theo:SOinteri}
All special orthogonal matrices have $p$-adic integer entries,
\[
  SO(3)_p\subset SL(3,\zed_p).
\]
\end{theorem}
 
Next we state a fundamental fact about the geometry of $SO(3)_p$.
\begin{theorem}
The elements of $SO(3)_p$ are rotations, i.e., they always have an eigenvalue 
$1$, and the corresponding eigenspace is the rotation axis.
\end{theorem}

We recall the parameterization of matrices of $SO(3)_p$. 
\begin{theorem}
\label{thm:rot-param}
A rotation of $SO(3)_p$ around $\boldsymbol{n}\in\QQ_p^3\backslash\{\boldsymbol{0}\}$ takes the following matrix form with respect to an 
orthogonal basis $(\boldsymbol{g},\boldsymbol{h},\boldsymbol{n})$ of $\QQ_p^3$:
\beq\label{eq:rotazgeneric}
\mathcal{R}_{\boldsymbol{n}}(\sigma)=\begin{pmatrix}
\frac{1-\alpha\sigma^2}{1+\alpha\sigma^2} & -\frac{2\alpha\sigma}{1+\alpha\sigma^2}&0\\ \frac{2\sigma}{1+\alpha\sigma^2} & \frac{1-\alpha\sigma^2}{1+\alpha\sigma^2}&0\\0&0&1
\end{pmatrix}
\eeq 
with $\sigma\in\QQ_p\cup\{\infty\}$ and $\alpha=Q(\boldsymbol{h})/Q(\boldsymbol{g})$.
The set $SO(3)_{p,\boldsymbol{n}}$ of rotations around a given 
$\boldsymbol{n}\in\QQ_p^3\backslash\{\boldsymbol{0}\}$ forms an Abelian subgroup of $SO(3)_p$
\end{theorem}

Last, in parallel to the real orthogonal case, we found that only certain main angle decompositions of rotations around the reference axes hold for $SO(3)_p$.
\begin{theorem}\label{theo:Cardano}
For an odd prime $p>2$, every $M\in SO(3)_p$ can be written as any of the 
Cardano type compositions,
\beq \notag 
\mathcal{R}_x\mathcal{R}_y\mathcal{R}_z,\ \
\mathcal{R}_z\mathcal{R}_y\mathcal{R}_x,\ \ 
\mathcal{R}_z\mathcal{R}_x\mathcal{R}_y,\ \  
\mathcal{R}_y\mathcal{R}_x\mathcal{R}_z,
\eeq 
respectively, with certain rotation parameters $\sigma,\tau,\omega\in\QQ_p\cup\{\infty\}$.

For a prime $p\equiv1\mod4$, every $M\in SO(3)_p$ can also be written as any 
of the Cardano-type compositions,
\beq\notag 
\mathcal{R}_x\mathcal{R}_z\mathcal{R}_y,\ \ 
\mathcal{R}_y\mathcal{R}_z\mathcal{R}_x,
\eeq 
respectively, with certain rotation parameters $\sigma,\tau,\omega\in\QQ_p\cup\{\infty\}$.
\end{theorem}

Note that none of the other possible forms of main angle decomposition, familiar
from the real Euclidean case, hold universally for $SO(3)_p$. 
In particular, no Euler angle decomposition exists for primes $p>2$, and none of the 
Euler and Cardano decompositions exist when $p=2$.


\section{Quotienting the group $\boldmath{SO(3)_p}$}
\label{sec:quot}
A fundamental tool to understand the algebraic aspects of $SO(3)_p$ is 
quotienting modulo $p^k$ the group \cite{pigliapochi-msc}: this is possible because the matrix entries are
$p$-adic integers (Theorem \ref{theo:SOinteri}).
These quotients are finite groups and, thus, much easier to study while generating $SO(3)_p$ as a projective limit. 
For us, these projections are crucial as their representations induce 
representations of the whole $SO(3)_p$ and, thus, allow us to find 
irreps easily. 

We can write a $p$-adic integer $a\in\zed_p$ as a sequence 
$a=(a_1,a_2,a_3,\ldots)$ with $a_k\in \zed/p^k$ such that $a_{k+1}\equiv a_k \mod p^k$ 
for all $k>0$, with component-wise addition and multiplication in the respective rings \cite{Serre}. 
This means that the functions 
\begin{equation}\begin{split}
\pi_k':\zed_p &\rightarrow \zed/p^k\\
        a=(a_i)_{i=1}^{\infty} &\mapsto     a_k
\end{split}\end{equation}
are ring homomorphisms,
\begin{align}
  \pi_k'(a+b) &=\pi_k'((a_i)_i+(b_j)_j) =\pi_k'((a_i+b_i)_i) =a_k+b_k =\pi_k'(a)+\pi_k'(b),\\
  \pi_k'(ab)  &=\pi_k'((a_i)_i(b_j)_j)  =\pi_k'((a_ib_i)_i)  =a_kb_k  =\pi_k'(a)\pi_k'(b).
\end{align}
This can be extended to multiplicative groups of matrices over $\zed_p$: the 
matrix product is defined through sums and products of elements, for which $\pi_k'$
are homomorphisms; then the maps
\beq
 \pi_k \left(M\right) = \pi_k \left((m_{ij})_{ij}\right):= \left(\pi_k'\left(m_{ij}\right)\right)_{ij}
\eeq
are group homomorphisms on any group contained in $\mathcal{M}_{n\times n}(\zed_p)$ to 
some other group contained in $\mathcal{M}_{n\times n}(\zed/p^k)$.
\begin{remark}
\textup{Note that all the maps $\pi_k$ are continuous with respect to 
 the $p$-adic metric on $SO(3)_p$ and the discrete topology on 
 the range, which is a finite set.}
\end{remark}

Theorem \ref{theo:SOinteri} means that the projection maps $\pi_k$ are well defined on $SO(3)_p$,
\begin{align}
  \pi_k: SO(3)_p         &\rightarrow SO(3)_p \mod p^k \subset SL(3,\ZZ/p^k), \nonumber\\
         L=(a_{ij})_{ij} &\mapsto     (a_{ij} \mod p^k)_{ij}
\end{align}
since all matrix entries of eligible $L$ are $p$-adic integers. 
This allows us to build a projective sequence of quotient groups,
\beq
  G_{p^k}:= \pi_k(SO(3)_p) = SO(3)_p\mod p^k. 
\eeq
These are finite groups as anticipated, being subgroups of $SL(3,\ZZ/p^k)$. 
By $G_{\boldsymbol{n},p^k}$ we will denote the Abelian subgroup of $G_{p^k}$ 
of $\pi_k$-projected rotations $\mathcal{R}_{\boldsymbol{n}}(\sigma) \in SO(3)_p$ around a 
fixed $\boldsymbol{n}\in \QQ_p^3\backslash\{\boldsymbol{0}\}$,
\beq
  G_{\boldsymbol{n},p^k}:= \pi_k(SO(3)_{p,\boldsymbol{n}}). 
\eeq
The Cardano decompositions (Theorem \ref{theo:Cardano}) hold in $G_{p^k}$ too.

\medskip
It is possible to rewrite the parameterization of the subgroups $SO(3)_{p,\boldsymbol{n}}$ 
of rotations around a given $\boldsymbol{n}\in\QQ_p^3\backslash\{\boldsymbol{0}\}$ (Theorem \ref{thm:rot-param}) only in terms 
of $p$-adic integers \cite{our1st}. To wit, if $\kappa$ denotes the determinant of the 
quadratic form preserved by $p$-adic planar rotations, in the plane orthogonal to the 
rotation axis, we have
\beq\notag
SO(2)_p^\kappa=\left\{\mathcal{R}_\kappa(\sigma)\,:\,\sigma\in\ZZ_p\right\}\cup \left\{-\mathcal{R}_\kappa\left(\sigma\right)\,:\,\sigma\in p\ZZ_p\right\}
\eeq 
for $p>2$ and $\kappa=-v$ (rotations around the $z$-axis), 
or $p=2$ and $\kappa=1,\pm5$ ($\kappa=1$ corresponds to rotations around $x$, $y$ and $z$ axes), 
while
\beq\notag
SO(2)_p^\kappa=\left\{\mathcal{R}_\kappa(\sigma)\,:\,\sigma\in\ZZ_p\right\}\cup \left\{-\cR_\kappa\left(\sigma\right)\,:\,\sigma\in \ZZ_p\right\}
\eeq 
for $p>2$ and $\kappa=p,up$ (rotations around $x$ and $y$ axes),
or $p=2$ and $\kappa=\pm2,\pm10$. 
We also recall that $\mathcal{R}_{\boldsymbol{n}}\left(-\frac{1}{\alpha\sigma}\right)
 =\mathcal{R}_{\boldsymbol{n}}(\infty)\mathcal{R}_{\boldsymbol{n}}(\sigma)$, 
which allows us to write for every $p>2$,
\beq\label{eq:scrttGzpk}\begin{aligned}
G_{z,p^k}&=\left\{\mathcal{R}_{z}(\sigma) \mod p^k\,:\,\sigma\in\ZZ_p\right\}
         \cup \left\{\mathcal{R}_{z}(\infty)\mathcal{R}_{z}(\sigma) \mod p^k\,:\,\sigma\in p\ZZ_p\right\}\\
         &=\left\{\mathcal{R}_{z}(\sigma) \mod p^k\,:\,\sigma\in\ZZ/p^k\right\}
         \cup \left\{\mathcal{R}_{z}(\infty)\mathcal{R}_{z}(\sigma) \mod p^k\,:\,\sigma\in p(\ZZ/p^k)\right\},
\end{aligned}\eeq
and similarly,
\beq\label{eq:scrttGxypk}
G_{\boldsymbol{n},p^k}=\left\{\mathcal{R}_{\boldsymbol{n}}(\sigma) \mod p^k\,:\,\sigma\in\ZZ/p^k\right\}
    \cup\left\{\mathcal{R}_{\boldsymbol{n}}(\infty)\mathcal{R}_{\boldsymbol{n}}(\sigma) \mod p^k\,:\,\sigma\in \ZZ/p^k\right\}\eeq
for the $x$ and $y$ axes. 

Moreover, when $p=2$, we have the following for each of the reference axes:
\beq\label{eq:scrttGxyz2k}\begin{aligned}
G_{\boldsymbol{n},2^k}&=\left\{\mathcal{R}_{\boldsymbol{n}}(\sigma) \mod 2^k\,:\,\sigma\in\ZZ_2\right\}
     \cup \left\{\mathcal{R}_{\boldsymbol{n}}(\infty)\mathcal{R}_{\boldsymbol{n}}(\sigma) \mod 2^k\,:\,\sigma\in 2\ZZ_2\right\}\\
&=\left\{\mathcal{R}_{\boldsymbol{n}}(\sigma) \mod 2^k\,:\,\sigma\in2(\ZZ/2^k)\right\}\\
&\phantom{==}
\cup\left\{\mathcal{R}_{\boldsymbol{n}}(\infty)\mathcal{R}_{\boldsymbol{n}}(\sigma) \mod 2^k\,:\,\sigma\in 2(\ZZ/2^k)\right\}\\
&\phantom{==}
\cup\left\{\mathcal{R}_{\boldsymbol{n}}(\sigma) \mod 2^k\,:\,\sigma\in \ZZ/2^k,\,\sigma\equiv1\mod 2\right\},
\end{aligned}
\eeq 
with the elements in the last set of the union in the 
plane orthogonal to $\boldsymbol{n}$ written as
\beq \label{eq:secproblematico2}
\begin{pmatrix}
\frac{-2(\sigma+{\sigma}^2)}{1+2(\sigma+{\sigma}^2)}&-\frac{1+2\sigma}{1+2(\sigma+{\sigma}^2)}\\ \frac{1+2\sigma}{1+2(\sigma+{\sigma}^2)}&\frac{-2(\sigma+{\sigma}^2)}{1+2(\sigma+{\sigma}^2)}\end{pmatrix} \mod 2^k \ ,\ \sigma\in\ZZ/2^{k-1}.
\eeq 

Eqs. \eqref{eq:scrttGzpk}, \eqref{eq:scrttGxypk} and \eqref{eq:scrttGxyz2k} provide a parameterization of $G_{p^k}$ for every prime 
$p$ and $k\geq1$, thanks to the existing Cardano decompositions. 

We conclude this section by introducing the groups of solutions modulo $p^k$, 
$k\geq1$, of the defining equations of $SO(3)_p$,
\beq \label{solsyspk}
\widehat{G}_{p^k} 
      = \{ L\in \mathcal{M}_{3\times 3}(\ZZ/p^k) : 
            A\equiv L^TAL \mod p^k,\, \det L \equiv 1 \mod p^k \}.
\eeq 
These groups will turn out to be fundamental in the discussion about 
the representations of $SO(3)_p$ at least for $k=1$ in this paper.

\begin{remark}\label{GpkhatGpk}
By virtue of Theorem \ref{theo:SOinteri}
\beq 
G_{p^k} \subseteq \widehat{G}_{p^k},
\eeq
and a fundamental open question is whether these groups coincide or not, 
by varying the values of prime $p>2$ and $k\geq1$. This question
has the character of Hensel's Lemma \cite{Cassels}: 
do the integer solutions modulo $p$ of $A=L^TAL$ with $\det L=1$ lift to $p$-adic solutions?

This boils down to checking if $\lvert G_{p^k}\rvert=\lvert \widehat{G}_{p^k}\rvert$ for some $p$ and $k$. Focusing on primes $p>2$, this can be done by comparing the order of $\widehat{G}_{p^k}$ with the order of $G_{p^k}$ given by one of its possible Cardano decompositions. One of the main aspects left to be discovered in this matter is, then, the multiplicity of the Cardano decompositions for $G_{p^k}$ for every prime $p>2,\,k\geq1$, that is, how many distinct Cardano decompositions around the same axes exist for each matrix of $G_{p^k}$. Once known that, it will not be hard to find the order of the subgroups $G_{\boldsymbol{n},p^k}$ for the $x$, $y$ and $z$ axes from Eqs.~\eqref{eq:scrttGzpk} and \eqref{eq:scrttGxypk} and, then, deduce $\lvert G_{p^k}\rvert$.
\end{remark}


\section{The group $\boldmath{SO(3)_p\mod p}$ for $\boldmath{p>2}$}
\label{pi1SO3}
We specialize Eqs.~\eqref{eq:scrttGzpk} and \eqref{eq:scrttGxypk} for
odd prime $p>2$ in the case $k=1$: for the $x$ or $y$ axes,
\beq\label{eq:Gxyp}
G_{\boldsymbol{n},p}
 =\left\{\cR_{\boldsymbol{n}}(\sigma)\mod p\,:\,\sigma\in\ZZ/p\right\}
  \cup\left\{\cR_{\boldsymbol{n}}(\infty)\cR_{\boldsymbol{n}}\left(\sigma\right)\mod p\,:\,\sigma\in \ZZ/p\right\},
\eeq
while for the $z$ axis,
\beq \label{eq:Gzp}
G_{z,p}=\left\{\cR_z(\sigma) \mod p,\,\sigma\in\ZZ/p\right\}\cup\left\{\cR_z(\infty) \mod p\right\}.
\eeq 
By using \eqref{eq:rotazgeneric}, Eqs. \eqref{eq:Gxyp} and \eqref{eq:Gzp} can be rewritten as follows:
\begin{align}
G_{x,p}&=\left\{\begin{pmatrix}
1&0&0\\0&1&0\\0&2\xi&1
\end{pmatrix} \mod p\ :\ \xi\in\ZZ/p\right\}\cup \left\{\begin{pmatrix}
1&0&0\\0&-1&0\\0&-2\xi&-1
\end{pmatrix} \mod p\ :\ \xi\in\ZZ/p\right\},\\
G_{y,p}&=\left\{\begin{pmatrix}
1&0&0\\0&1&0\\2\eta&0&1
\end{pmatrix}\mod p\ :\ \eta\in\ZZ/p\right\}\cup\left\{\begin{pmatrix}
-1&0&0\\0&1&0\\-2\eta&0&-1
\end{pmatrix}\mod p\ :\ \eta\in\ZZ/p\right\},\\
G_{z,p}&=\left\{\begin{pmatrix}
\frac{1+v\zeta^2}{1-v\zeta^2} & 
\frac{2v\zeta}{1-v\zeta^2}&0\\
\frac{2\zeta}{1-v\zeta^2} & \frac{1+v\zeta^2}{1-v\zeta^2}&0\\0&0&1
\end{pmatrix}\mod p\ :\ \zeta\in\ZZ/p\right\}\cup\left\{\begin{pmatrix}
-1&0&0\\0&-1&0\\0&0&1
\end{pmatrix}\mod p\right\}.
\end{align}
We will refer to the first set of each of these three unions as ``first branch'' and to the second one as ``second branch.'' They provide a parameterization of 
$G_p=SO(3)_p\mod p$ thanks to the existing Cardano decompositions. 

\medskip
\begin{remark}
It is trivial to check that matrices within each of the two branches of $G_{x,p}$ and $G_{y,p}$ are all distinct by varying the parameter in $\ZZ/p$. Moreover a matrix from the first branch cannot be equal to a matrix of the second branch, for both $G_{x,p}$ and $G_{y,p}$, because $1\not\equiv-1\mod p$ when $p>2$. Furthermore, by equating two matrices from the first branch of $G_{z,p}$, we get
\beq\label{eq:zu1}
    \frac{1+v\zeta^2}{1-v\zeta^2}\equiv \frac{1+v{\zeta'}^2}{1-v{\zeta'}^2}\mod p,
\eeq
which is equivalent to $ \zeta^2\equiv{\zeta'}^2\mod p$. Plugging this into
\beq\label{eq:zu2}
    \frac{2\zeta}{1-v\zeta^2}\equiv \frac{2\zeta'}{1-v{\zeta'}^2}\mod p,
\eeq
we obtain $\zeta\equiv\zeta'\mod p$. Finally, we equate a matrix from the first branch of $G_{z,p}$ to the one of the second branch,
\beq 
 \frac{1+v\zeta^2}{1-v\zeta^2}\equiv-1\mod p,
\eeq 
which is equivalent to $1\equiv-1\mod p$, which is impossible.
Therefore, the matrices in $G_{x,p},\, G_{y,p}$ and $G_{z,p}$ are all distinct 
by varying the respective parameters. Then, it follows that
\beq \label{eq:cardxyzmodp}
\lvert G_{x,p}\rvert=\lvert G_{y,p}\rvert=2p,\quad \lvert G_{z,p}\rvert=p+1.
\eeq 
\end{remark}

\medskip
As anticipated, a fundamental point in this matter is to discover the 
multiplicity of a modulo-$p$-reduced Cardano decomposition for $G_p$, 
for example that of the kind $\cR_x\cR_y\cR_z$. A systematic ambiguity 
of order 2 in the Cardano representations is based on the relation
\beq \label{eq:Cardambiguity}
\cR_x(\infty)\cR_y(\infty)\cR_z(\infty)\equiv\1\mod p.
\eeq 
Moreover, in the calculations below, we are going to use the following commutation relation
\beq 
[\cR_z(\infty),\cR_y(\eta)]\equiv\cR_z(\infty)\cR_y(\eta)\cR_z(\infty)^{-1}\cR_y(\eta)^{-1}\equiv \cR_y(-2\eta)\mod p
\eeq 
so that
\beq \label{eq:commrelzy}
\cR_z(\infty)\cR_y(\eta)\equiv [\cR_z(\infty),\cR_y(\eta)]\cR_y(\eta)\cR_z(\infty)\equiv \cR_y(-\eta)\cR_z(\infty)\mod p.
\eeq 

Now, we start analyzing all the possibilities for the branches of the three rotations 
involved in the Cardano decomposition of the kind $\cR_x\cR_y\cR_z$. A triple 
$ijk$ with $i,j,k\in\{1,2\}$ will denote a Cardano representation of the kind 
$\cR_x\cR_y\cR_z$, where the $x$, $y$ and $z$ rotation modulo $p$ are taken 
from the $i$th, $j$th and $k$th branch, respectively. There are $8=2^3$ possible Cardano 
representations $\cR_x\cR_y\cR_z$ with respect to the branches of each of the three 
involved rotations. Thus, there are 36 possibilities of equating two triple products, 
i.e., of equating modulo $p$ two Cardano representations depending on their rotation branches.

Suppose that a given $M\in G_p$ admits a Cardano representation of the kind 111, 
i.e., there exist $\xi,\eta,\zeta\in\ZZ/p$ such that
\beq 
  M\equiv\cR_x(\xi)\cR_y(\eta)\cR_z(\zeta)\mod p.
\eeq 
Then, there exists at least another distinct Cardano representation of $M$ along the 
same axes: by using Eqs.~\eqref{eq:Cardambiguity} and \eqref{eq:commrelzy}, we get
\begin{align}
    M&\equiv\cR_x(\xi)\,\cR_x(\infty)\cR_y(\infty)\cR_z(\infty)\,\cR_y(\eta)\cR_z(\zeta)\notag\\
     &\equiv \cR_x(\infty)\cR_x(\xi)\cR_y(\infty)\cR_y(-\eta)\cR_z(\infty)\cR_z(\zeta)\mod p.
\end{align}
The latter is a Cardano representation of $M$ of the kind either 221 if $\zeta\not\equiv0\mod p$ or 222 if $\zeta\equiv0\mod p$. The reverse is also true: if $M \in G_p$ admits a 
Cardano representation 221 with non-zero parameter for the $z$-rotation (or a Cardano 
representation 222), then it admits a Cardano representation 111 too.

Similarly, if for a given $M\in G_p$ there exist $\xi,\eta,\zeta\in\ZZ/p$ such that $M$ 
has Cardano representation of the kind 121,
\beq
M\equiv \cR_x(\xi)\cR_y(\infty)\cR_y(\eta)\cR_z(\zeta)\mod p,
\eeq 
then also
\begin{align}
  M&\equiv \cR_x(\xi)\,\cR_x(\infty)\cR_y(\infty)\cR_z(\infty)\,\cR_y(\infty)\cR_y(\eta)\cR_z(\zeta)\notag\\
   &\equiv \cR_x(\infty)\cR_x(\xi)\cR_y(\infty)^2\cR_z(\infty)\cR_y(\eta)\cR_z(\zeta)\notag\\
   &\equiv \cR_x(\infty)\cR_x(\xi)\cR_y(-\eta)\cR_z(\infty)\cR_z(\zeta)\mod p,
\end{align}
which is a Cardano representation of $M$ of the kind either 211 if $\zeta\not\equiv0\mod p$ or 212 if $\zeta\equiv0\mod p$.

Moreover, if for a given $M\in G_p$ there exist $\xi,\eta\in\ZZ/p$ such that $M$ has Cardano representation 112,
\beq
  M\equiv \cR_x(\xi)\cR_y(\eta)\cR_z(\infty)\mod p,
\eeq 
then also
\begin{align}
  M&\equiv \cR_x(\xi)\,\cR_x(\infty)\cR_y(\infty)\cR_z(\infty)\,\cR_y(\eta)\cR_z(\infty)\notag\\
   &\equiv \cR_x(\infty)\cR_x(\xi)\cR_y(\infty)\cR_y(-\eta)\cR_z(0)\mod p,
\end{align}
that is, a Cardano representation 221 of $M$.

Finally, if for a given $M\in G_p$ there exist $\xi,\eta\in\ZZ/p$ such that $M$ 
has Cardano representation 122,
\beq
  M\equiv \cR_x(\xi)\cR_y(\infty)\cR_y(\eta)\cR_z(\infty)\mod p,
\eeq 
then also
\begin{align}
M&\equiv \cR_x(\xi)\,\cR_x(\infty)\cR_y(\infty)\cR_z(\infty)\,\cR_y(\infty)\cR_y(\eta)\cR_z(\infty)\notag\\
&\equiv \cR_x(\infty)\cR_x(\xi)\cR_y(\infty)^2\cR_z(\infty)\cR_y(\eta)\cR_z(\infty)\notag\\
&\equiv \cR_x(\infty)\cR_x(\xi)\cR_y(-\eta)\cR_z(0)\mod p,
\end{align}
that is, a Cardano representation 211 of $M$.

We have analyzed each of the eight initial triples $ijk$ since the reverse of each of the above reasonings is also valid. This provides six different modular congruences of Cardano representations with respect to certain triples of branches. We deduce the following.
\begin{remark}\label{rem:2cardanop}
Given a Cardano representation of $M$ of the kind $\cR_x\cR_y\cR_z$ of parameters $\xi,\eta,\zeta$, respectively, then $M$ admits at least another distinct Cardano representation of the kind 
$\cR_x\cR_y\cR_z$ with parameters $\xi',\eta',\zeta'$, respectively: that obtained 
by changing the branches 
of the $x$- and $y$-rotation and
\begin{itemize}
    \item by changing the branch of the $z$-rotation if its parameter $\zeta$ is $0\in\ZZ/p$ or $\infty$ and
    \item by fixing the branch of the $z$-rotation if $\zeta\in\ZZ/p,\, \zeta\not\equiv0\mod p$,
\end{itemize}
and with parameters
\beq 
\xi'\equiv\xi,\ \ \ \eta'\equiv-\eta,\ \ \ \zeta'\equiv-\frac{1}{\alpha\zeta}\mod p.
\eeq 
\end{remark}

Now, we show that a given $M\in G_p$ cannot admit two distinct Cardano 
representations of the kind $\cR_x\cR_y\cR_z$ with respect to the same 
three branches. Here we need to compare each of the eight branch triples $ijk$ with itself. 
The thesis is trivial to check in the four cases $ij2$ where the $z$-rotation is 
$\cR_z(\infty)$. We now focus on the case 111 where the three rotations of the 
Cardano representations all come from the first branch since the calculations 
for the remaining three cases $ij1$ are very similar.

Suppose that a given $M\in G_p$ can be represented as
\beq 
  M\equiv \cR_x(\xi)\cR_y(\eta)\cR_z(\zeta)\equiv \cR_x(\xi')\cR_y(\eta')\cR_z(\zeta')\mod p
\eeq 
for certain $\xi,\xi',\,\eta,\eta',\,\zeta,\zeta'\in\ZZ/p$. 
It means that 
\begin{align*} 
 M&\equiv \begin{pmatrix} \frac{1+v\zeta^2}{1-v\zeta^2} & \frac{2v\zeta}{1-v\zeta^2}&0\\ 
   \frac{2\zeta}{1-v\zeta^2} & \frac{1+v\zeta^2}{1-v\zeta^2}&0\\\frac{2}{1-v\zeta^2}\big(\eta(1+v\zeta^2)+2\xi\zeta\big)&\frac{2}{1-v\zeta^2}\big(\xi(1+v\zeta^2)+2v\eta\zeta\big)&1 \end{pmatrix}\\ 
  &\equiv\begin{pmatrix} \frac{1+v{\zeta'}^2}{1-v{\zeta'}^2} & \frac{2v\zeta'}{1-v{\zeta'}^2}&0\\ 
            \frac{2\zeta'}{1-v{\zeta'}^2} & \frac{1+v{\zeta'}^2}{1-v{\zeta'}^2}&0\\\frac{2}{1-v{\zeta'}^2}\big(\eta'(1+v{\zeta'}^2)+2\xi'\zeta'\big)&\frac{2}{1-v{\zeta'}^2}\big(\xi'(1+v{\zeta'}^2)+2v\eta'\zeta'\big)&1 \end{pmatrix}\mod p.
\end{align*}
By equating the entries of the $2\times2$ upper left minors, we get 
Eqs.~\eqref{eq:zu1} and \eqref{eq:zu2}, giving $\zeta\equiv\zeta'\mod p$, i.e.,
$\cR_z(\zeta)\equiv\cR_z(\zeta')\mod p$. Then the following are equivalent:
\begin{align}
    &\cR_x(\xi)\cR_y(\eta)\cR_z(\zeta)\equiv \cR_x(\xi')\cR_y(\eta')\cR_z(\zeta')\mod p,\\
    &\cR_x(\xi)\cR_y(\eta)\equiv \cR_x(\xi')\cR_y(\eta')\mod p,\\
    &\begin{pmatrix}
1&0&0\\0&1&0\\2\eta&2\xi&1
\end{pmatrix}\equiv\begin{pmatrix}
1&0&0\\0&1&0\\2\eta'&2\xi'&1
\end{pmatrix}\mod p,\\
&\xi'\equiv\xi\mod p\quad \text{and}\quad \eta'\equiv\eta\mod p,\\
&\cR_x(\xi')\equiv\cR_x(\xi)\mod p \quad \text{and}\quad \cR_y(\eta')\equiv\cR_y(\eta)\mod p.
\end{align}
Thus each $M\in G_p$ has a unique Cardano representation of the kind $\cR_x\cR_y\cR_z$ with respect to a certain triple of rotation branches.

\medskip
There are $36-6-8=22$ possibilities left of equating two triples $ijk$, representing 
Cardano decompositions $\cR_x\cR_y\cR_z$ with respect to those triples of branches. 
For all of them, equating modulo $p$ the first or last matrix entry of the two Cardano 
products leads to $1\equiv-1\mod p$, which is impossible. It means that each $M\in G_p$ 
does not have Cardano representations of the kind $\cR_x\cR_y\cR_z$ with 
respect to different triples of branches, except from the cases of Remark \ref{rem:2cardanop}.

We have, thus, proved the following.
\begin{theorem}\label{teo:2cardanomodp}
Every $M\in G_p$, for every odd prime $p$, has exactly two distinct Cardano 
representations of the kind $\cR_x\cR_y\cR_z$. They are one of the following 
six (depending on $M$) according to Remark \ref{rem:2cardanop}, for certain $\xi,\eta,\zeta\in\ZZ/p$:
\beq\notag \begin{array}{ll}
     111\leftrightarrow 221\quad &\cR_x(\xi)\cR_y(\eta)\cR_z(\zeta)\equiv \cR_x(\infty)\cR_x(\xi)\cR_y(\infty)\cR_y(-\eta)\cR_z\left(\frac{1}{v\zeta}\right)\mod p\ \text{if}\  \zeta\not\equiv0 \mod p,\\ 
     111\leftrightarrow222\quad  & \cR_x(\xi)\cR_y(\eta)\cR_z(0)\equiv \cR_x(\infty)\cR_x(\xi)\cR_y(\infty)\cR_y(-\eta)\cR_z(\infty)\mod p,\\
    121 \leftrightarrow 211\quad&\cR_x(\xi)\cR_y(\infty)\cR_y(\eta)\cR_z(\zeta)\equiv \cR_x(\infty)\cR_x(\xi)\cR_y(-\eta)\cR_z\left(\frac{1}{v\zeta}\right)\mod p\ \text{if}\  \zeta\not\equiv0 \mod p,\\
    121 \leftrightarrow 212\quad&\cR_x(\xi)\cR_y(\infty)\cR_y(\eta)\cR_z(0)\equiv \cR_x(\infty)\cR_x(\xi)\cR_y(-\eta)\cR_z(\infty)\mod p,\\
    112 \leftrightarrow 221\quad&\cR_x(\xi)\cR_y(\eta)\cR_z(\infty)\equiv \cR_x(\infty)\cR_x(\xi)\cR_y(\infty)\cR_y(-\eta)\cR_z(0)\mod p,\\
    122 \leftrightarrow 211\quad& \cR_x(\xi)\cR_y(\infty)\cR_y(\eta)\cR_z(\infty)\equiv \cR_x(\infty)\cR_x(\xi)\cR_y(-\eta)\cR_z(0)\mod p.
\end{array}\eeq
\end{theorem}

\begin{corollary}\label{cor:ordineGp}
Given a prime $p>2$, the order of the group $SO(3)_p\mod p=G_p$ is given by
\beq \label{eq:ordineGp}
\lvert G_p\rvert=\frac{1}{2}\lvert G_{x,p}\rvert\lvert G_{y,p}\rvert\lvert G_{z,p}\rvert=2p^2(p+1).
\eeq 
\end{corollary}
\begin{proof}
The rotation groups $G_{x,p}$, $G_{y,p}$, and $G_{z,p}$ are subgroups of $G_p$; hence, each product $\cR_x\cR_y\cR_z$ such that $\cR_x\in G_{x,p}$, $\cR_y\in G_{y,p}$, and $\cR_z\in G_{z,p}$ is an element of $G_p$. By the the Cardano decomposition 
$\cR_x\cR_y\cR_z$, each element of $G_p$ has a decomposition into a product of these three,
and Theorem \ref{teo:2cardanomodp}, more precisely, states that each element has 
exactly two such decompositions. Therefore,
$\lvert G_p\rvert = \frac12 \lvert G_{x,p}\rvert\lvert G_{y,p}\rvert\lvert G_{z,p}\rvert$. 
\end{proof}

\subsection{Equivalence with $\widehat{G}_p$}
Now we consider the groups of Eq. \eqref{solsyspk} for $k=1$: we want to find a parameterization of these groups $\widehat{G}_{p}$ too.
Being $\pi_1(A)=\diag(1,-v,0)$ and setting $L=\left(a_{ij}\right)_{i,j=1}^3\in SO(3)_p$, the defining conditions of $SO(3)_p$ reduced modulo $p$ translate in the following system of equations for $\widehat{G}_p$:
\begin{equation}\label{sysmodp1}
\begin{cases}
a_{11}^2-va_{21}^2\equiv 1\mod p,\\
a_{12}^2-va_{22}^2\equiv -v\mod p,\\
a_{13}^2-va_{23}^2\equiv 0\mod p,\\
a_{11}a_{12}-va_{21}a_{22}\equiv 0\mod p,\\
a_{11}a_{13}-va_{21}a_{23}\equiv 0\mod p,\\
a_{12}a_{13}-va_{22}a_{23}\equiv 0\mod p,\\
a_{11}a_{22}a_{33}+a_{12}a_{23}a_{31}+a_{13}a_{21}a_{32}-a_{31}a_{22}a_{13}-a_{32}a_{23}a_{11}-a_{33}a_{21}a_{12}\equiv 1\mod p.
\end{cases}
\end{equation}
The elements $a_{31},a_{32},a_{33}$ appear only in the last equation of \eqref{sysmodp1}, 
the one deriving from the determinant: this means that we could at first try to find 
the possible solutions for the other two rows, looking at the other equations.

If $a_{23}\not\equiv0\mod p$, the third equation of \eqref{sysmodp1} would give 
$v\equiv a_{13}^2a_{23}^{-2}\mod p$, which is in contradiction to the fact that 
$v$ is not a square modulo $p$. Hence it must be $a_{13},a_{23}\equiv0\mod p$, 
and from \eqref{sysmodp1}, we have
\begin{equation}\label{sysmodp2}
\begin{cases}
a_{11}^2-va_{21}^2\equiv 1\mod p,\\
a_{12}^2-va_{22}^2\equiv -v\mod p,\\
a_{13},a_{23}\equiv0\mod p,\\
a_{11}a_{12}-va_{21}a_{22}\equiv 0\mod p,\\
a_{33}(a_{11}a_{22}-a_{21}a_{12})\equiv 1\mod p.
\end{cases}
\end{equation}
With $a_{33}(a_{11}a_{22}-a_{21}a_{12})$ being non-zero and, thus, with its two factors, we can write $a_{33}\equiv(a_{11}a_{22}-a_{21}a_{12})^{-1}$. Actually, we can say 
more about these values: let us multiply the first two equations of \eqref{sysmodp2} 
member by member, obtaining
\begin{align*}
-v &\equiv a_{11}^2a_{12}^2-va_{11}^2a_{22}^2-va_{12}^2a_{21}^2+v^2a_{21}^2a_{22}^2\\
&\equiv a_{11}^2a_{12}^2-va_{11}^2a_{22}^2-va_{12}^2a_{21}^2+v^2a_{21}^2a_{22}^2+2v{a_{11}}{a_{12}}{a_{21}}{a_{22}}-2v{a_{11}}{a_{12}}{a_{21}}{a_{22}}\\
&\equiv \left({a_{11}}{a_{12}}-v{a_{21}}{a_{22}}\right)^2-v\left({a_{11}}{a_{22}}-{a_{12}}{a_{21}}\right)^2\\
&\equiv -v\left({a_{11}}{a_{22}}-{a_{12}}{a_{21}}\right)^2\mod p.
\end{align*}
We can, thus, deduce
\begin{equation}
a_{11}a_{22}-a_{21}a_{12}\equiv a_{33}\equiv\pm1.
\end{equation}
Now, the whole remaining study concerns the submatrix 
\beq
\begin{pmatrix}
       a_{11} & a_{12} \\
       a_{21} & a_{22}
      \end{pmatrix}
\eeq 
with
\begin{align}
&a_{11}^2-va_{21}^2\equiv1\mod p\label{eq:sysmodp21},\\
&a_{12}^2-va_{22}^2\equiv-v\mod p\label{eq:sysmodp22},\\
&{a_{11}}{a_{12}}-v{a_{21}}{a_{22}}\equiv0\mod p\label{eq:sysmodp23},\\
&{a_{11}}{a_{22}}-{a_{12}}{a_{21}}\equiv\pm_11\mod p\label{eq:sysmodp24}.
\end{align}
Here, $a_{11}$ and $a_{21}$ cannot be both divisible by $p$; otherwise, Eq. \eqref{eq:sysmodp21} would be impossible: suppose that $a_{11}\not\equiv0\mod p$. 
Equation \eqref{eq:sysmodp23} provides $a_{12}\equiv va_{11}^{-1}a_{21}a_{22}\mod p$. 
We plug this into Eq. \eqref{eq:sysmodp22} and use Eq. \eqref{eq:sysmodp21},
\begin{align*}
-v&\equiv v^2a_{11}^{-2}a_{21}^2a_{22}^2-va_{22}^2\equiv -va_{22}^2\left(1-va_{11}^{-2}a_{21}^2\right)\\ &\equiv -va_{22}^2\big(1+a_{11}^{-2}(1-a_{11}^2)\big)\equiv -va_{11}^{-2}a_{22}^2\mod p,
\end{align*}
which implies that
\begin{equation}\label{eq:pm2res}
a_{22}\equiv \pm_2 a_{11}.
\end{equation}
Now, from Eqs. \eqref{eq:sysmodp21} and \eqref{eq:sysmodp22}, we get 
$a_{12}^2\equiv v(a_{22}^2-1)\equiv v(a_{11}^2-1)\equiv v^2a_{21}^2\mod p$, and so
\beq \label{eq:pm3res}
a_{12}\equiv\pm_3va_{21}\mod p.
\eeq 
The choice is now only about $a_{11},a_{21}$ and three signs. However, we show that these three signs $\pm_1,\pm_2,\pm_3$ are actually the same. Equations \eqref{eq:pm2res} and \eqref{eq:pm3res} make Eq. \eqref{eq:sysmodp23} become $va_{11}a_{21}(\pm_31\mp_21)\equiv0\mod p\Leftrightarrow (\pm_31\mp_21)a_{21}\equiv0\mod p$. If $a_{21}\not\equiv0\mod p$, then $\pm_2=\pm_3$; if $a_{21}\equiv0\mod p$, then the choice of $\pm_3$ is meaningless. Hence, we can write $\pm_21=\pm_31$, in general. Now, Eqs. \eqref{eq:sysmodp21} and \eqref{eq:sysmodp24} give $\pm_11\equiv 
\pm_2\left(a_{11}^2-va_{21}^2\right)\equiv \pm_21$; therefore $\pm_1=\pm_2=\pm_3$.

Equation~\eqref{eq:sysmodp23} becomes an identity, while Eqs. \eqref{eq:sysmodp22} and \eqref{eq:sysmodp24} become equivalent to Eq. \eqref{eq:sysmodp21}. Very similar steps in the hypothesis that $a_{21}\not\equiv0\mod p$ provide the same results, which we collect in the following equations:
\beq 
\left\{\begin{aligned}
&a_{11}^2-va_{21}^2\equiv1\mod p,\\
&{a_{22}}\equiv\pm_1 {a_{11}}\mod p,\\
&{a_{12}}\equiv\pm_1v{a_{21}}\mod p,\\
&{a_{13}}\equiv{a_{23}}\equiv0\mod p,\\
&{a_{33}}\equiv\pm1\mod p.
\end{aligned}\right.
\eeq 
In addition, $\pm=\pm_1$: $a_{33}^{-1}\equiv a_{11}a_{22}-a_{12}a_{21}\Leftrightarrow \pm1\equiv\pm_1(a_{11}^2-va_{21}^2)\equiv\pm_11\mod p$. We conclude that matrix solutions of Eq. \eqref{sysmodp1} are of the form 
\beq 
\begin{pmatrix}
{a_{11}}&\pm v{a_{21}}&0\\{a_{21}}&\pm{a_{11}}&0\\{a_{31}}&{a_{32}}&\pm1
\end{pmatrix}\in\cM_{3\times3}(\ZZ/p)
\eeq 
with $a_{11}^2-va_{21}^2\equiv1\mod p$.

All of the above can be summarized in the following result.
\begin{theorem}\label{parampi1SO3}
The matrices of $\widehat{G}_p$ for any prime $p>2$ are uniquely defined by $a,b,c,d\in\zed/p$ (the first two solutions of $a^2-vb^2\equiv1\mod p$ and the others chosen freely) and by a sign $s=\pm1$,
\beq\label{eq:matsistmodp}
M(a,b,c,d,s)=\begin{pmatrix}
       a & svb & 0 \\
       b & s a & 0 \\
       c & d & s
      \end{pmatrix}.
\eeq
\end{theorem}

\begin{corollary}\label{ordpi1SO3}
Given a prime $p>2$, the order of the matrix group $\widehat{G}_p$ is 
\beq \label{eq:ordermodp}
\lvert \widehat{G}_p\rvert=2p^2(p+1).
\eeq
\end{corollary}
\begin{proof}
For a fixed odd prime $p$, the number of solutions of the equation $a^2-vb^2\equiv1\mod p$ is $p+1$ (Appendix \ref{sec:a2vb21modp}). Then, the total number of combinations of $a,b,c,d,s$ is $(p+1)\cdot p\cdot p\cdot 2$ such that every choice represents a distinct matrix. Then the number of distinct matrices is $2p^2(p+1)$.
\end{proof}

\begin{remark}\label{doubleapproach}
Proving that $G_{p}$ coincides with $\widehat{G}_{p}$ is equivalent to proving that $\lvert G_{p}\rvert =\lvert \widehat{G}_{p}\rvert$ (Remark \ref{GpkhatGpk}). This is true for every prime $p>2$, by comparing Eqs. \eqref{eq:ordineGp} and \eqref{eq:ordermodp}. This implies that Eq. \eqref{eq:matsistmodp} is an equivalent parameterization for $G_{p}$.

\end{remark}

\medskip
The study of $G_p:= SO(3)_p\mod p$ regarding its order and parameterization turns out 
to be fundamental to build a first list of representations of $SO(3)_p$.  
Examining the commutators and the abelianization of a group is enough to exhaust 
the list of one-dimensional representations \cite{simon}: this will be helpful to 
understand how many other representations are present and their form. It is for this 
reason that here we anticipate the structure of the commutators between matrices of 
$G_p$ for every prime $p>2$ by exploiting Remark \ref{doubleapproach}. 
We calculate $[M,N]=MNM^{-1}N^{-1}$ with $M,\,N\in G_p$ written through the parameterization \eqref{eq:matsistmodp} as
\beq\notag
M=\begin{pmatrix}
a_{11}&tva_{21}&0\\a_{21}&ta_{11}&0\\a_{31}&a_{32}&t
\end{pmatrix},\ \ N=\begin{pmatrix}
b_{11}&rvb_{21}&0\\b_{21}&rb_{11}&0\\b_{31}&b_{32}&r
\end{pmatrix}.
\eeq 
We get  
\beq\label{eq:commG35}
[M,N]=\begin{pmatrix}
c_{11}&vc_{21}&0\\c_{21}&c_{11}&0\\c_{31}&c_{32}&1
\end{pmatrix},
\eeq 
where
\beq \begin{aligned}\label{eq:commG35entries}
    &c_{11}=(b_{11}^2-tvb_{21}^2)(a_{11}^2-rva_{21}^2)+(t+r-tr-1)va_{11}a_{21}b_{21}b_{11},\\
   &c_{21}=(t-1)(a_{11}^2-rva_{21}^2)b_{21}b_{11}-(r-1)(b_{11}^2-tvb_{21}^2)a_{11}a_{21}.
\end{aligned}
\eeq

\subsection{From $p$-adic rotations to dihedral groups}
We introduce a homomorphism, which 
will turn out to be useful to our aim of finding two-dimensional irreducible representations
of $SO(3)_p$.
\begin{proposition}
The map
\beq \label{eq:homominor}
\begin{aligned}
K_p': \,& G_p\rightarrow GL((\zed/p)^2),\\ &M\mapsto M_{33},
\end{aligned}\eeq 
where $M_{33}$ is the the $2\times2$ upper-left minor of $M$, is a group homomorphism.
\end{proposition}
\begin{proof}
According to the parameterization \eqref{eq:matsistmodp} and Remark \ref{doubleapproach}, the map $K_p'$ acts as follows:
\beq\notag 
  \begin{pmatrix}
       a & svb & 0 \\
       b & s a & 0 \\
       c & d & s
      \end{pmatrix}
      \mapsto 
   \begin{pmatrix}
       a & svb\\
       b & s a
      \end{pmatrix}.
\eeq
The image $\mathfrak{I}(K_p')$ is parameterized by $(a,b)\in(\ZZ/p)^2$, which are solutions of $a^2-vb^2\equiv1\mod p$ and $s\in\{\pm1\}$
.
\begin{itemize}
    \item $\mathfrak{I}(K_p')\subset GL(\ZZ/p)$,
    \beq\notag
 \det
 \begin{pmatrix}
       a & svb\\
       b & s a
      \end{pmatrix}
      =
  s(a^2-vb^2)\equiv\pm1\neq0.   
\eeq
\item $K_p'$ is a group homomorphism,
\beq\notag  K_p'\begin{pmatrix}
a&svb&0\\b&sa&0\\c&d&s
\end{pmatrix}K_p'\begin{pmatrix}
a'&s'vb'&0\\b'&s'a'&0\\c'&d'&s'
\end{pmatrix}=\begin{pmatrix}
a&svb\\b&sa
\end{pmatrix}\begin{pmatrix}
a'&s'vb'\\b'&s'a'
\end{pmatrix}=\begin{pmatrix}
aa'+vsbb'&vs'(ab'+sa'b)\\a'b+sab'&s'(saa'+vbb')
\end{pmatrix},\eeq
which equals
\begin{align*}
    K_p'\left(\begin{pmatrix}
a&svb&0\\b&sa&0\\c&d&s
\end{pmatrix}\begin{pmatrix}
a'&s'vb'&0\\b'&s'a'&0\\c'&d'&s'
\end{pmatrix}\right)=K_p'\begin{pmatrix}
aa'+svbb'&vs'(ab'+sa'b)&0\\a'b+sab'&s'(saa'+vbb')&0\\\dots&\dots&ss'
\end{pmatrix}.
\end{align*}
\end{itemize}
\end{proof}

A point of interest lies in the behaviour of $a$ and $b$ as solutions of 
$a^2-vb^2\equiv1\mod p$ while composing matrices in $\mathfrak{I}(K_p')$, 
\cite{pigliapochi-msc}. 
In $\mathfrak{I}(K_3')$ 
there are only the matrices
\begin{equation}
 \begin{pmatrix}
  1 & 0\\
  0 & 1
 \end{pmatrix},
\quad
 \begin{pmatrix}
  0 & 1\\
  1 & 0
 \end{pmatrix},
\end{equation}
up to sign changes on the components, and it is trivial that products of matrices of this form give back matrices of this form. However, which mathematical property guarantees that a product
\beq 
\begin{pmatrix}
a&svb\\b&sa
\end{pmatrix}\begin{pmatrix}
a'&s'vb'\\b'&s'a'
\end{pmatrix}=\begin{pmatrix}
aa'+vsbb'&vs'(ab'+sa'b)\\a'b+sab'&s'(saa'+vbb')
\end{pmatrix}
\eeq 
still gives a matrix coherent with the definition in $\mathfrak{I}(K_p')$
? Can we find some deeper meaning about it?

As proved in Appendix \ref{sec:a2vb21modp}, solutions of $a^2-vb^2\equiv1$ in $\zed/p$ form a cyclic group: the key to prove it lies in looking at the pair $(a,b)$ as a single ``complex'' number $a+ib$, $i^2=v$, and introducing products between pairs as products between complex numbers. As a consequence, the ``product'' of two solutions gives another solution according to the product of matrices. As a matter of fact, another way to represent the complex product is by writing the components in a $2\times2$ matrix of the form
\beq
\begin{pmatrix}
       a & vb \\
       b & a
      \end{pmatrix}
\eeq
and applying the common matrix product.

This is of great interest for us, as these are exactly half of the matrices in $\mathfrak{I}(K_p')$
, the ones with $s=1$. The other half, the matrices with $s=-1$, are obtained simply by multiplying these by the (Pauli) matrix $Z$,
\begin{equation}
 \begin{pmatrix}
       a & -vb \\
       b & -a
      \end{pmatrix}
      =
\begin{pmatrix}
       a & vb \\
       b & a
      \end{pmatrix}Z,      
      \qquad
       Z=\begin{pmatrix}
       1 & 0 \\
       0 & -1
      \end{pmatrix}.
\end{equation}
Note that $Z\in 
\mathfrak{I}(K_p')
$ because 
$\cR_x(\infty)=\diag(1,-1,-1)_{\ZZ_p}\in SO(3)_p$ for every prime $p$.

This means that the group $\mathfrak{I}(K_p')
$ is generated by two matrices,
\beq
\begin{array}{cc}
 C=\begin{pmatrix}
       a_0 & -b_0 \\
       b_0 & a_0
      \end{pmatrix}, &
 Z=\begin{pmatrix}
       1 & 0 \\
       0 & -1
      \end{pmatrix}
\end{array},
\eeq
where $C$ is the matrix obtained by $(a_0,b_0)$, a generator of the cyclic group of the solutions $(a,b)$.

Finally we observe that
\begin{itemize}
 \item $C^{p+1}=\1$ because of the isomorphism with the solutions subgroup of $\zed/p[i]^*$,
 \item $Z^2=\1$, and
 \item $ZCZ=C^{-1}$, as can be checked readily.
\end{itemize}
\par This gives us a complete presentation for our group,
\begin{equation}
\mathfrak{I}(K_p')
=\langle C,Z \mid C^{p+1}=Z^2=\1, ZCZ=C^{-1}\rangle,
\end{equation}
which is exactly the presentation for the dihedral group $D_{p+1}$: in other words, the group $\mathfrak{I}(K_p')
$ is isomorphic to the group of symmetries of a $(p+1)$-gon. Denoting this isomorphism with $\varphi$, we have
\beq \label{eq:isodihedral}
\varphi:\mathfrak{I}(K_p')
\ \stackrel{\sim}{\rightarrow}\ D_{p+1},\ C\mapsto a,\,Z\mapsto x,
\eeq 
where $a$ and $x$ are the generators (rotation and reflection, respectively) of $D_{p+1}$.

This will allow us to exploit the known representation theory of dihedral 
groups for the $p$-adic group $SO(3)_p$, some basic facts about which are reported 
in Appendix \ref{dihrep}. While we will not discuss this in detail, Appendix \ref{dihrep1} 
shows that $\mathfrak{I}(K_p')
$ has always four one-dimensional representations. 
Furthermore, Appendix \ref{dihrep2} shows that there always exists a two-dimensional 
irrep of $\mathfrak{I}(K_p')$, and for $p>3$, there are more than one 2-irreps, 
leading to different $p$-adic qubits.

\begin{remark}
We summarize the preceding discussion by recording that 
\beq 
  J_p^{(i)}:=\rho_i\circ\varphi\circ K_p'\circ\pi_1,
\eeq 
where $\rho_i$'s are the $2$-irreps of $D_{p+1}$, are 
$2$-irreps of $SO(3)_p$, providing examples of $p$-adic qubits for every $p>2$.
\end{remark}


\section{Representations for $SO(3)_3 \mod 3$}
\label{sec:3-mod-3}
In the present section, we exploit the previous results for the case $p=3$ to study some irreducible representations of $SO(3)_3$. From these will emerge one irrep of degree 2, which we will propose as the foundation for the $3$-adic qubit.


\subsection{Abelianization}

As we anticipated, a starting point in the study of the irreducible representations of $SO(3)_3$ is through those of its projection $G_3$, which is a finite group. First, we find its abelianization to understand how many one-dimensional representations $G_3$ has, which benefits the analysis of the degrees of its remaining irreps.

We specify Eq. \eqref{eq:matsistmodp} for $p=3$: by Remark \ref{doubleapproach}, the matrices of $G_3$ are parametrized by 
\beq \label{G3param}
M(a,b,c,d,s)=\begin{pmatrix}
a & -s b & 0 \\
b & s a & 0 \\
c & d & s
\end{pmatrix}
\eeq 
with $(a,b)\in\{(\pm1,0),\,(0,\pm1)\}$. Now we consider Eqs. \eqref{eq:commG35} and \eqref{eq:commG35entries} for $p=3$: $c_{21}\equiv0\mod 3$ because $a_{11}a_{21}\equiv b_{11}b_{21}\equiv0\mod3$. Moreover, $c_{11}\equiv\pm1\mod 3$, so a commutator of $G_3$ is of the form
\beq\notag 
[M,N]=\begin{pmatrix}
\pm1&0&0\\0&\pm1&0\\c_{31}&c_{32}&1
\end{pmatrix}.
\eeq 
The set of commutators of this form is closed under multiplication, meaning that 
\beq\label{eq:[G3,G3]}
[G_3,G_3]=\left\{\begin{pmatrix}
c_{11}&0&0\\0&c_{11}&0\\c_{31}&c_{32}&1
\end{pmatrix},\ c_{11}\in\{\pm1\},\,c_{31},c_{32}\in\ZZ/3\right\}.
\eeq 
By comparing this with Eq. \eqref{G3param}, we can state the following.
\begin{proposition}
 $M(a,b,c,d,s)\in G_3$ belongs to the commutator subgroup $[G_3,G_3]$ if and only if $s=1$ and $(a,b)=(\pm1,0)$.
\end{proposition}
The order of $[G_3,G_3]$ is $18$, so the abelianization $\Ab(G_3):=G_3/[G_3,G_3]$ is of order $72/18=4$.

\begin{proposition}\label{Abchar3}
$\Ab(G_3)$ is isomorphic to $\zed/2\times\zed/2$: a class of matrices is characterized by $s$ and $t:= a^2-b^2$ using the form in \eqref{G3param}.
 The following equations are possible representatives for the classes $\Ab(s,t)$:
 \[
 \begin{array}{cc}
  \Ab(1,1)\rightsquigarrow
  \begin{pmatrix}
       1 & 0 & 0 \\
       0 & 1 & 0 \\
       0 & 0 & 1
      \end{pmatrix},
      &
  \Ab(-1,1)\rightsquigarrow    
  \begin{pmatrix}
       1 & 0 & 0 \\
       0 & -1 & 0 \\
       0 & 0 & -1
      \end{pmatrix},
     \\
  \Ab(1,-1)\rightsquigarrow    
  \begin{pmatrix}
       0 & -1 & 0 \\
       1 & 0 & 0 \\
       0 & 0 & 1
      \end{pmatrix},
      &
  \Ab(-1,-1)\rightsquigarrow    
  \begin{pmatrix}
       0 & 1 & 0 \\
       1 & 0 & 0 \\
       0 & 0 & -1
      \end{pmatrix}.     
 \end{array}
\]
\end{proposition}
\begin{proof}
According to Eqs. \eqref{G3param} and \eqref{eq:[G3,G3]} we have
\beq \notag 
M(a,b,c,d,s)[G_3,G_3]=\begin{pmatrix}
a&-s b&0\\b&s a&0\\c&d&s
\end{pmatrix}\begin{pmatrix}
c_{11}&0&0\\0&c_{11}&0\\c_{31}&c_{32}&1
\end{pmatrix}=\begin{pmatrix}
a c_{11}&-s c_{11}b &0\\c_{11}b &sa c_{11}&0\\c_{11}c +s c_{31}&c_{11}d +s c_{32}&s 
\end{pmatrix}
\eeq 
with $c_{11},s \in\{\pm1\},\,(a,b)\in\{(\pm1,0),(0,\pm1)\}$.

With fixed $a,b$ and $s$, we have $M(a,b,c,d,s)[G_3,G_3]=M(a,b,c',d',s)[G_3,G_3]$ for every $c,c',d,d'\in \ZZ/3$. Moreover, for every $c,c',\,d,d'\in\ZZ/3$,
\beq\notag 
M(1,0,c,d,1)[G_3,G_3]=M(-1,0,c',d',1)[G_3,G_3]=\left\{\begin{pmatrix}
       c_{11}&0&0\\0&c_{11}&0\\c_{31}'&c_{32}'&1
\end{pmatrix},\ c_{11}\in\{\pm1\},\,c_{31}',c_{32}'\in\ZZ/3\right\}.
\eeq 
A representative of this class is $\1
$, and the class is located by $s=1$ and $t=a^2-b^2=1$.

Similarly, for every $c,c',\,d,d'\in\ZZ/3$,
\begin{align*}
& M(1,0,c,d,-1)[G_3,G_3]=M(-1,0,c',d',-1)[G_3,G_3]=\left\{\begin{pmatrix} c_{11}&0&0\\0&-c_{11}&0\\c_{31}'&c_{32}'&-1
\end{pmatrix},\ c_{11}\in\{\pm1\},\,c_{31}',c_{32}'\in\ZZ/3\right\},\\
& M(0,1,c,d,1)[G_3,G_3]=M(0,-1,c',d',1)[G_3,G_3]=\left\{\begin{pmatrix} 0&-c_{11}&0\\c_{11}&0&0\\c_{31}'&c_{32}'&1
\end{pmatrix},\ c_{11}\in\{\pm1\},\,c_{31}',c_{32}'\in\ZZ/3\right\},\\
& M(0,1,c,d,-1)[G_3,G_3]=M(0,-1,c',d',-1)[G_3,G_3]=\left\{\begin{pmatrix} 0&c_{11}&0\\c_{11}&0&0\\c_{31}'&c_{32}'&1
\end{pmatrix},\ c_{11}\in\{\pm1\},\,c_{31}',c_{32}'\in\ZZ/3\right\}.
\end{align*}
\end{proof}

As a consequence of Proposition \ref{Abchar3}, we have the following further result.

\begin{proposition}\label{4_1rep3}
 $G_3$ cannot have more than four orthogonal one-dimensional irreducible representations $U:G_3\rightarrow\cee^\ast$: a possible choice for the four 1-irreps is as follows:
 \begin{itemize}
  \item $\det:M(a,b,c,d,s)\mapsto 1$,
  \item $s:M(a,b,c,d,s)\mapsto s$,
  \item $t:M(a,b,c,d,s)\mapsto a^2-b^2
  $, and
  \item $st:M(a,b,c,d,s)\mapsto s(a^2-b^2)
  $.
 \end{itemize}
\end{proposition}

\begin{proof}
A $1$-irrep for a group is always the trivial representation mapping any element of the group to $1$. This can be seen as the determinant on matrices of $G_3$.

The irreducible character $\chi_U$ of a 1-irrep $U$ is $U$ itself: in fact, the trace of a $1\times1$ matrix is equal to the only value it contains. Since the product in $\CC$ is commutative, the $1$-irreps of $G_3$ have to be invariant on abelianized classes. As consequences, we can actually build the character table for $\Ab(G_3)$ in order to find the 1-irreps of $G_3$, and the two characterizing quantities $s$ and $t$ in Proposition \ref{Abchar3} can be good candidates as characters. Indeed, we show that $\det,\,s,\,t,\,st$, as defined in the Proposition, are the four $1$-irreps of this group.

Now, $\det,\,s,\,t,\,st$ are homomorphisms from $G_3$ to $\CC^\ast$ because
\beq\notag 
M(a,b,c,d,s)M(a',b',c',d',s')=\begin{pmatrix}
        aa'-sbb'&\dots&0\\
        sab'+a'b&\dots&0\\\dots&\dots&ss'
        \end{pmatrix},
\eeq 
so $s(MM')=ss'=s(M)s(M')$, $t(MM')=(aa'-sbb')^2-(sab'+a'b)^2=(a^2-b^2)({a'}^2-{b'}^2)=t(M)t(M')$, and similarly for $st$.

Furthermore, in accordance with general Peter-Weyl theory \cite{simon}, 
$\det,\,s,\,t,\,st$ are orthogonal representations. This can be seen by looking at the character table of $\Ab(G_3)$.
\beq \notag
\begin{array}{ l c c c c } 
 \hline\hline
\Ab(G_3)\simeq\{(s,t)\} \quad &\quad (1,1)\quad &\quad (1,-1)\quad & \quad (-1,1)\quad & \quad (-1,-1)  \\  
\det &1 &1 &1&1 \\  
 s& 1 &1 & -1&-1 \\  
t & 1 &-1 &1 &-1\\ 
st & 1& -1&-1 &1 \\ \hline\hline
\end{array}
\eeq
The rows are all different, and they correspond to the different values the $1$-irreps associate with each element of $G_3$.

We have shown four 1-irreps of $G_3$: any other 1-irrep would be
equivalent to these because of Proposition \ref{Abchar3}.
\end{proof}


\subsection{Degree of representations}

The Peter-Weyl theorem \cite{simon} states that the sum of the squares of the degrees $d_{\alpha}$ of the irreducible representations of $G$ located by $\alpha$ is equal to the order of $G$.

The order of $G_3$ is $72$ by Corollary \ref{cor:ordineGp}, and there are $4$ 
irreducible representations of $G_3$ of degree $1$ by Proposition \ref{4_1rep3}. Then,
\begin{equation}
 \sum_{\alpha\mid d_{\alpha}>1}d_{\alpha}^2=68.
\end{equation}
Unfortunately, this is not sufficient to deduce the nature of higher-degree irreps: we need another proposition (proved in Appendix \ref{conjclg3}).

\begin{proposition}\label{conjclg3prop}
 $G_3$ is partitioned in nine conjugacy classes. Their orders and lists are in Appendix \ref{conjclg3}.
\end{proposition}

As a consequence of the fact that the number of conjugacy classes of a group equals the number of its irreducible representations \cite{simon}, there are nine non equivalent irreps of $G_3$, among which five have degree greater than 1. Suppose, then, to sort the degrees of these irreps from the smaller to the greater as $d_1,\ldots,d_9$: we can state that
\begin{align}
 d_1=d_2=d_3=d_4=1,\\
 d_5^2+d_6^2+d_7^2+d_8^2+d_9^2=68.\label{sum68}
\end{align}
It is known that $d_\alpha$ divides $\lvert G\rvert$ when $\alpha$ locates an irreducible representation of $G$. This is enough to deduce all of the $d_\alpha$:
\begin{proposition}\label{3dregrees}
The only solution for \eqref{sum68} for integers $d_\alpha$ divisors of $72$ greater than $1$ is the following up to permutations:
\beq
 d_5=2,\ d_6=d_7=d_8=d_9=4.
\eeq
\end{proposition}
The proof of this proposition is in Appendix \ref{degreeP3}. Looking at $d_5$, we can now state the following.
\begin{proposition}\label{2irrep1}
 There exists a two-dimensional irrep $V$ of $G_3$, that is,
 \[
 K_3:G_3\rightarrow GL(V), \quad V\simeq\CC^2,
 \]
 such that for any subspace $Y\subset V$,
 \[
  (K_3(G_3))Y=Y\ \text{implies}\ Y=\{0_V\}\text{ or }Y=V.
 \]
\end{proposition}


\subsection{A $3$-adic qubit}

Now it is not hard to find an explicit form of a 2-irrep for $SO(3)_3$, \cite{pigliapochi-msc}. An important consequence of the way $G_3=SO(3)_3\mod3=\pi_1(SO(3)_3)$ was constructed according to 
Sec. \ref{sec:quot} is as follows.
\begin{proposition}\label{2irrep2}
 There exists a 2-irrep $J_3$ of $SO(3)_3$,
 \[
  J_3:= K_3\circ\pi_1.
 \]
\end{proposition}

\begin{proof}
 $J_3$ is a homomorphism because it is composition of homomorphisms. $J_3$ is irreducible since $K_3$ is so: if $Y\subset V$, then the following are equivalent:
 \begin{align}
     &(J_3(SO(3)_3))Y=Y,\\
     &(K_3(\pi_1(SO(3)_3)))Y=Y,\\
     &(K_3(G_3))Y=Y.
 \end{align}
 The latter is equivalent to $Y=\{0_V\}$ or $Y=V$.
\end{proof}

The first ingredient to find an explicit form for a 2-irrep $K_3$ is the homomorphism $K_3'$ of Eq. \eqref{eq:homominor}. Furthermore, one considers the embedding of the integers modulo $3$ into the complex numbers \beq 
\imath:\ZZ/3\hookrightarrow \CC,\ 0_{\ZZ/3}\mapsto0_\CC,\, \pm1_{\ZZ/3}\mapsto\pm1_\CC.
\eeq 
This can be extended to a map between groups of matrices on these fields, respectively, which we will always call $\imath$ with abuse of notation. Note that the solutions $(a,b)\in\{(\pm1,0),(0,\pm1)\}\subset(\ZZ/3)^2$ of $a^2+b^2\equiv 1\mod3$ are such that their embedding $\{(\pm1,0),(0,\pm1)\}\subset\CC^2$ is solution of $a^2+b^2=1$ in $\CC$. 
\begin{proposition}\label{prop:immersione}
$\imath:\mathfrak{I}(K_3')
\rightarrow GL(\CC^2)$ is a group homomorphism.
\end{proposition}
\begin{proof}
\begin{itemize}
    \item $\mathfrak{I}(\imath)\subset GL(\CC^2)$,
    \beq\notag
    \det\begin{pmatrix}
           a&-sb\\b&sa
    \end{pmatrix}_{\CC}=s(a^2+b^2)=s\neq0,
    \eeq 
    where $(a,b)\in\{(\pm1,0),(0,\pm1)\}\subset\CC^2$.
\item $\{0,\pm1\}\subset\CC$ is not a subfield since it is not closed under the addition operation. However, if we look at the entries of $\begin{pmatrix}
aa'-sbb'&-s'(ab'+sa'b)\\a'b+sab'&s'(saa'-bb')
\end{pmatrix}$, when these operations are to be performed both in $\ZZ/3$ and in $\CC$, they never present sums between $\pm1$ and $\pm1$ because $(a,b),(a',b')\in\{(\pm1,0),(0,\pm1)\}$. It follows that
\begin{align*}
    &\imath\begin{pmatrix}
           a&-sb\\b&sa
    \end{pmatrix}_{\ZZ/3}\imath\begin{pmatrix}
           a'&-s'b'\\b'&s'a'
    \end{pmatrix}_{\ZZ/3}=\begin{pmatrix}
           a&-sb\\b&sa
    \end{pmatrix}_{\CC}\begin{pmatrix}
           a'&-s'b'\\b'&s'a'
    \end{pmatrix}_{\CC}=\begin{pmatrix}
aa'-sbb'&-s'(ab'+sa'b)\\a'b+sab'&s'(saa'-bb')
\end{pmatrix}_{\CC}\\
&=\imath\begin{pmatrix}
aa'-sbb'&-s'(ab'+sa'b)\\a'b+sab'&s'(saa'-bb')
\end{pmatrix}_{\ZZ/3}=\imath\left(\begin{pmatrix}
           a&-sb\\b&sa
    \end{pmatrix}_{\ZZ/3}\begin{pmatrix}
           a'&-s'b'\\b'&s'a'
    \end{pmatrix}_{\ZZ/3}\right).
\end{align*}
\end{itemize}
\end{proof}

Now we are able to state the following.

\begin{proposition}\label{irrepK}
The function 
\beq 
K_3:G_3\rightarrow GL(\CC^2),\ \ \ \ K_3:=\imath\circ K_3' 
\eeq 
is the unitary $2$-irrep of $G_3$ up to isomorphisms.
\end{proposition}
\begin{proof}
As said, parameterization \eqref{G3param} is fundamental to understand the action of $G_3$, that is,
\beq\notag 
  \begin{pmatrix}
       a & -s b & 0 \\
       b & s a & 0 \\
       c & d & s
      \end{pmatrix}_{\ZZ/3}
      \stackrel{K_3'}{\mapsto}\ 
   \begin{pmatrix}
       a & -s b\\
       b & s a
      \end{pmatrix}_{\ZZ/3}\stackrel{\imath}{\mapsto}\ 
   \begin{pmatrix}
       a & -s b\\
       b & s a
      \end{pmatrix}_\CC.
\eeq
The image $\mathfrak{I}(K_3)$ of $K_3$ is parameterized by $(a,b)\in\{(\pm1,0),(0,\pm1)\}\subset\CC^2,s\in\{\pm1\}\subset\CC$.
\begin{itemize}
    \item $\mathfrak{I}(K_3)\subset GL(\CC^2)$ as noticed in Proposition \ref{prop:immersione}.
\item $K_3$ is a group homomorphism, being the composition of two group homomorphisms $K_3'$ and $\imath$.
\item $K_3$ is irreducible:\\
if $Y$ is a subrepresentation such that $Y\neq \{0\},\CC^2$, then $\text{dim}(Y)=1$, i.e., $Y=\CC y$ for any $y\in Y$. By definition of subrepresentation, $K_3(G_3)Y\subset Y$ means that $K_3(M)y \in Y$ for every $M\in G_3$ and $y\in Y$: it must be $K_3(M)y=\lambda y$ for some $\lambda\in\CC$ depending on $y$ and $M$. In other words, $y$ is a common eigenvector for every action of the group. This must be true, in particular, for
\beq \notag K_3\big(M(1,0,c,d,1)\big)=\begin{pmatrix}
       1&0\\0&1
\end{pmatrix}_\CC\ \ \ \text{of eigenvectors}\ \ \ \begin{pmatrix}
       1\\0
\end{pmatrix},\ \begin{pmatrix}
       0\\1
\end{pmatrix}
\eeq 
and for 
\beq\notag K_3\big(M(0,1,c,d,1)\big)=\begin{pmatrix}0&-1\\1&0\end{pmatrix}_{\CC}\ \ \ \text{of eigenvectors}\ \ \ \begin{pmatrix}
       i\\1
\end{pmatrix},\ \begin{pmatrix}
       -i\\1
\end{pmatrix}.\eeq 
These two group actions have no common eigenvectors. This is a contradiction: the representation $K_3$ is irreducible.
\item $K_3$ is unitary,\\namely, $\mathfrak{I}(K_3)$ is composed by unitary linear transformations on $\CC^2$,
\begin{align*}
&\langle\begin{pmatrix}
a&-sb\\b&sa
\end{pmatrix}\begin{pmatrix}
  x_1\\
  x_2
 \end{pmatrix}
 ,
 \begin{pmatrix}
a&-sb\\b&sa
\end{pmatrix}\begin{pmatrix}
  y_1\\
  y_2
 \end{pmatrix}
 \rangle=\langle \begin{pmatrix}
 ax_1-sbx_2\\bx_1+sax_2
 \end{pmatrix}, \begin{pmatrix}
 ay_1-sby_2\\by_1+say_2
 \end{pmatrix}\rangle\\
 &=(ax_1-sbx_2)(ay_1-sby_2)+(bx_1+sax_2)(by_1+say_2)\\
 &=(a^2+b^2)(x_1y_1+x_2y_2)= x_1y_1+x_2y_2= \langle
 \begin{pmatrix}
  x_1\\
  x_2
 \end{pmatrix}
 ,
 \begin{pmatrix}
  y_1\\
  y_2
 \end{pmatrix}
 \rangle
 \end{align*}
 for every $\begin{pmatrix}
  x_1\\
  x_2
 \end{pmatrix}
 ,
 \begin{pmatrix}
  y_1\\
  y_2
 \end{pmatrix}\in\CC^2$ according to its Euclidean scalar product.
\end{itemize}
\end{proof}

We have, then, an explicit form for $J_3$ too: it is a function that makes a $2\times2$ complex matrix out of the first components of the $p$-adic sequences in the upper-left $2\times2$ part of a matrix in $SO(3)_3$, that is,
\beq \label{Jaction}
\begin{aligned}
&J_3:=K_3\circ \pi_1:SO(3)_3\rightarrow GL(\CC^2)\\
& \begin{pmatrix}
  (a_{111},a_{112},\ldots) & (a_{121},a_{122},\ldots) & (a_{131},a_{132},\ldots)\\
  (a_{211},a_{212},\ldots) & (a_{221},a_{222},\ldots) & (a_{231},a_{232},\ldots)\\
  (a_{311},a_{312},\ldots) & (a_{321},a_{322},\ldots) & (a_{331},a_{332},\ldots)
 \end{pmatrix}_{\ZZ_3}
\stackrel{J_3}{\longmapsto}
 \begin{pmatrix}
  a_{111} & a_{121}\\
  a_{211} & a_{221}
 \end{pmatrix}_\CC.
\end{aligned}\eeq

We, thus, have the following final result.
\begin{proposition}
The pair $(\CC^2,J_3)$ with $J_3$ defined as in \eqref{Jaction} is a 3-adic qubit, that is, a continuous unitary irreducible linear (and then projective) representation  of degree $2$ on $\CC$ of the $3$-adic special orthogonal group $SO(3)_3$ defined in \eqref{SO3p}.
\end{proposition}

The isomorphism $\varphi:\mathfrak{I}(K_3')
\rightarrow D_4$ of Eq. \eqref{eq:isodihedral} gives us an equivalent way to derive the $2$-irrep $J_3$ of $SO(3)_3$, applying the representation theory of dihedral groups of even degree to $\mathfrak{I}(K_3')$. However, first, Appendix \ref{dihrep1} confirms the four one-dimensional representations we found in Proposition \ref{4_1rep3} for the case $p=3$. Then, Appendix \ref{dihrep2} shows that for $p=3$ there is a unique 2-irrep of $\mathfrak{I}(K_3')$,
\beq 
C_{\ZZ/3}\mapsto\begin{pmatrix}
       0&-1\\1&0
\end{pmatrix}_\CC,\ \ \ \ Z_{\ZZ/3}\mapsto \begin{pmatrix}
       1&0\\0&-1
\end{pmatrix}_\CC.
\eeq 
This coincides with $\imath$ of Proposition \ref{prop:immersione} if we choose $(a_0,b_0)=(0,1)$, leading to the same qubit representation $J_3$ of Eq. \eqref{Jaction}, which, then, can be written as $J_3=\rho\circ\varphi\circ K_3'\circ\pi_1$ too, where $\rho$ denotes the $2$-irrep of $D_4$.

\section{Representations for $SO(3)_5\mod 5$}
\label{sec:5-mod-5}
We have $\widehat{G}_5=G_5$ according to Remark \ref{doubleapproach}. We choose $u=2$ as non-square invertible $5$-adic integer. In this way, the parameterization \eqref{eq:matsistmodp} for the matrices in $G_5$ becomes

\beq \label{G5param}
M(a,b,c,d,s)=\begin{pmatrix}
a & -2s b & 0 \\
b & s a & 0 \\
c & d & s
\end{pmatrix}, 
\eeq 
where $(a,b)\in\{(\pm1,0),\,(\pm2,1),\,(\pm2,-1)\}$ are the solutions of $a^2+2b^2\equiv1\mod5$.


\subsection{Abelianization}
We specify Eqs. \eqref{eq:commG35} and \eqref{eq:commG35entries} for $p=5$: it turns out that $(c_{11},c_{21})\in\{(1,0),\,(2,\pm1)\}$, so a commutator of $G_5$ is of one of the following forms:
\beq 
\begin{pmatrix}
1&0&0\\0&1&0\\c_{31}&c_{32}&1
\end{pmatrix},\ \ \begin{pmatrix}
2&-2&0\\1&2&0\\c_{31}&c_{32}&1
\end{pmatrix},\ \ \begin{pmatrix}
2&2&0\\-1&2&0\\c_{31}&c_{32}&1
\end{pmatrix}.
\eeq 
The set of commutators of these forms is closed under multiplication, meaning that
\beq\label{eq:[G5,G5]}
[G_5,G_5]=\left\{
\begin{pmatrix}
c_{11}&-2c_{21}&0\\c_{21}&c_{11}&0\\c_{31}&c_{32}&1
\end{pmatrix},\ (c_{11},c_{21})\in\{(1,0),(2,\pm1)\},\, c_{31},c_{32}\in\ZZ/5\right\}.
\eeq 
By comparing this with Eq. \eqref{G5param}, we can state the following.
\begin{proposition}
 $M(a,b,c,d,s)\in G_5$ belongs to the commutator subgroup $[G_5,G_5]$ if and only if $s = 1$ and $(a,b)\in\{(1,0),(2,\pm1)\}$.
\end{proposition}
The order of $[G_5,G_5]$ is $75$, so the abelianization $\Ab(G_5):=G_5/[G_5,G_5]$ is of order $300/75=4$.

\begin{proposition}\label{Abchar5}
$\Ab(G_5)$ is isomorphic to $\zed/2\times\zed/2$: a class of matrices is characterized by $s$ and $t:= \text{sign}(a)$ (taking the representatives of the elements in $\ZZ/5$ as $-2,-1,0,1,2$) using the form in \eqref{G5param}.
 The following are possible representatives for the classes $\Ab(s,t)$:
 \[
 \begin{array}{cc}
  \Ab(1,1)\rightsquigarrow
  \begin{pmatrix}
       1 & 0 & 0 \\
       0 & 1 & 0 \\
       0 & 0 & 1
      \end{pmatrix},
      &
  \Ab(-1,1)\rightsquigarrow    
  \begin{pmatrix}
       1 & 0 & 0 \\
       0 & -1 & 0 \\
       0 & 0 & -1
      \end{pmatrix},
     \\
  \Ab(1,-1)\rightsquigarrow    
  \begin{pmatrix}
       -1 & 0 & 0 \\
       0 & -1 & 0 \\
       0 & 0 & 1
      \end{pmatrix},
      &
  \Ab(-1,-1)\rightsquigarrow    
  \begin{pmatrix}
       -1 & 0 & 0 \\
       0 & 1 & 0 \\
       0 & 0 & -1
      \end{pmatrix}.     
 \end{array}
\]
\end{proposition}
\begin{proof}
According to Eqs. \eqref{G5param} and \eqref{eq:[G5,G5]}, we have
\beq \notag 
M(a,b,c,d,s)[G_5,G_5]=\begin{pmatrix}
a&-2sb&0\\b&sa&0\\c&d&s
\end{pmatrix}\begin{pmatrix}
c_{11}&-2c_{21}&0\\c_{21}&c_{11}&0\\c_{31}&c_{32}&1
\end{pmatrix}=\begin{pmatrix}
ac_{11}-2sbc_{21}&-2ac_{21}-2sc_{11}b&0\\c_{11}b+sac_{21}&-2bc_{21}+sac_{11}&0\\c_{11}c+c_{21}d+sc_{31}&c_{11}d-2c_{21}c+sc_{32}&s
\end{pmatrix}
\eeq 
with $(c_{11},c_{21})\in\{(1,0),(2,\pm1)\},s \in\{\pm1\},\,(a,b)\in\{(\pm1,0),(\pm2,1),(\pm2,-1)\}$.

With fixed $a,b$ and $s$, we have $M(a,b,c,d,s)[G_5,G_5]=M(a,b,c',d',s)[G_5,G_5]$ for every $c,c',d,d'\in \ZZ/5$. Moreover, for every $c,c',c'',\,d,d',d''\in\ZZ/5$,
\begin{align*}
M&(1,0,c,d,1)[G_5,G_5]=M(2,1,c',d',1)[G_5,G_5]=M(2,-1,c'',d'',1)[G_5,G_5]\\
&=\left\{\begin{pmatrix}
c_{11}&-2c_{21}&0\\c_{21}&c_{11}&0\\c_{31}'&c_{32}'&1
\end{pmatrix},\ (c_{11},c_{21})\in\{(1,0),(2,\pm1)\},c_{31}',c_{32}'\in\ZZ/5\right\}.
\end{align*}
A representative of this class is $\1$, 
and the class is located by $s=1,\,t=\text{sign}(a)=1$.

Similarly, for every $c,c',c'',\,d,d',d''\in\ZZ/5$,
\begin{align*}
    M&(-1,0,c,d,1)[G_5,G_5]=M(-2,1,c',d',1)[G_5,G_5]=M(-2,-1,c'',d'',1)[G_5,G_5]\\
    &=\left\{\begin{pmatrix}-c_{11}&2c_{21}&0\\ -c_{21}&-c_{11}&0\\c_{31}'&c_{32}'&1\end{pmatrix},\ (c_{11},c_{21})\in\{(1,0),(2,\pm1)\},c_{31}',c_{32}'\in\ZZ/5\right\},\\
    M&(1,0,c,d,-1)[G_5,G_5]=M(2,1,c',d',-1)[G_5,G_5]=M(2,-1,c'',d'',-1)[G_5,G_5]\\
    &=\left\{\begin{pmatrix}c_{11}&-2c_{21}&0\\ -c_{21}&-c_{11}&0\\c_{31}'&c_{32}'&-1\end{pmatrix},\ (c_{11},c_{21})\in\{(1,0),(2,\pm1)\},c_{31}',c_{32}'\in\ZZ/5\right\},\\
    M&(-1,0,c,d,-1)[G_5,G_5]=M(-2,1,c',d',-1)[G_5,G_5]=M(-2,-1,c'',d'',-1)[G_5,G_5]\\
    &=\left\{\begin{pmatrix}-c_{11}&2c_{21}&0\\c_{21}&c_{11}&0\\c_{31}'&c_{32}'&-1\end{pmatrix},\ (c_{11},c_{21})\in\{(1,0),(2,\pm1)\},c_{31}',c_{32}'\in\ZZ/5\right\}.
\end{align*}
\end{proof}
As a consequence of Proposition \ref{Abchar5}, we have the following further result.
\begin{proposition}\label{4_1rep5}
 $G_5$ has four orthogonal one-dimensional irreducible representations $U:G_5\rightarrow\cee^\ast$: a possible choice for the four 1-irreps is as follows:
 \begin{itemize}
  \item $\det:M(a,b,c,d,s)\mapsto 1$,
  \item $s:M(a,b,c,d,s)\mapsto s$,
  \item $t:M(a,b,c,d,s)\mapsto \text{sign}(a)$ (taking the representatives of the elements in $\ZZ/5$ as $-2,-1,0,1,2$), and
  \item $st:M(a,b,c,d,s)\mapsto s\cdot\text{sign}(a)$.
 \end{itemize}
\end{proposition}
\begin{proof}
As for the case $p=3$, the $1$-irreps of $G_5$ correspond to the irreducible characters of $\Ab(G_5)$. The character table of $\Ab(G_5)\simeq\ZZ/2\times \ZZ/2$ is as follows.
\beq \notag
\begin{array}{ l c c c c } 
 \hline\hline
Ab(G_5)\simeq\{Ab(s,t)\} \quad &\quad Ab(1,1)\quad & \quad Ab(1,-1)\quad & \quad Ab(-1,1)\quad &\quad Ab(-1,-1)\\  
\det &1 &1 &1&1 \\  
 s& 1 &1 & -1&-1 \\  
t & 1 &-1 &1 &-1\\ 
st & 1& -1&-1 &1 \\ \hline\hline
\end{array}
\eeq
The bijective correspondence between the $1$-irreps of $G_5$ and the irreps of $\Ab(G_5)$ is as follows: if $\widetilde{\rho}:Ab(G_5)\rightarrow \CC^\ast$ is a $1$-irrep of $\Ab(G_5)$, then $\rho(M):=\widetilde{\rho}(M[G_5,G_5])$ is a $1$-irrep of $G_5$ (and vice versa). This concludes the proof
.
\end{proof}
\subsection{Degree of representations}
The Peter-Weyl theorem \cite{simon} for $G_5$ states that $\sum_\alpha d_\alpha^2=\lvert G_5\rvert=300$, where $d_\alpha$ are the degrees of the irreducible representations of $G_5$ located by $\alpha$. Then,
\beq 
\sum_{\alpha\mid d_\alpha>1}d_\alpha^2=296
\eeq 
according to Proposition \ref{4_1rep5}.

This is not enough to deduce the nature of higher-degree irreps, but another proposition (proved in Appendix \ref{conjclg5}) comes to our help.
\begin{proposition}\label{conjclg5prop}
 $G_5$ is partitioned in $14$ conjugacy classes. Their orders and lists are in Appendix \ref{conjclg5}.
\end{proposition}
It means that there are $14$ non-equivalent irreps of $G_5$, among which four have degree 1. Sorting the degrees of these irreps from the smaller to the greater as $d_1,\ldots,d_{14}$, we have
\begin{align}
 d_1=&d_2=d_3=d_4=1,\\
 &\sum_{\alpha=5}^{14}d_\alpha^2=296.\label{sum296}
\end{align}
Solving Eq. \eqref{sum296} for $d_\alpha$ divisors of $\lvert G_5 \rvert$ greater than $1$ provides $18$ possible solutions, among which four present no degree equal to $2$, which we are interested relatively to $5$-adic qubits. However, there exists a stronger result due to It$\hat{\text{o}}$ \cite{ito}, namely, $d_\alpha$ divides the index in $G_5$ of any of its (maximal) Abelian normal subgroups. In order to exploit it, we prove the following.
\begin{proposition}
\beq A:=\big\{M(1,0,c,d,1)\,:\,c,d\in\ZZ/5\big\}\eeq 
is the maximal normal Abelian subgroup of $G_5$.
\end{proposition}
\begin{proof}
Normal subgroups are unions of conjugacy classes. The elements in each $C_i$, $i=6,\dots,14$, $C_i$ conjugacy classes of $G_5$ in Appendix \ref{conjclg5}, do not commute, in general. These classes, then, are not of interest for our purposes here. On the other hand, the matrices in $\bigcup_{i=1}^5 C_i=\big\{M(1,0,c,d,1)\,:\,c,d\in\ZZ/5\big\}$ commute. This set is closed under matrix multiplication and contains the identity matrix.\end{proof}

The index of $A$ in $G_5$ is $\lvert G_5:A\rvert =\lvert G_5\rvert/\lvert A\rvert = 300/25=12$. It follows by It$\hat{\text{o}}$'s theorem that $d_\alpha\in\{2,3,4,6,12\}$ for every $\alpha=5,\dots,14$.

\begin{proposition}\label{prop:sum296}
Equation \eqref{sum296} has three solutions for integers $d_\alpha\in\{2,3,4,6,12\}$ up to permutations.
\beq\notag  
\begin{array}{lccccccccccc} 
\hline\hline
& \quad d_5\quad &\quad d_6\quad &\quad d_7\quad &\quad d_8\quad &\quad d_9\quad &\quad d_{10}\quad &\quad d_{11}\quad &\quad d_{12}\quad &\quad d_{13}\quad  &\quad d_{14}\\ 
\textup{First}&2&2&3&3&3&3&6&6&6&\quad 12\\ 
\textup{Second}\quad &2&2&6&6&6&6&6&6&6&\quad 6\\  
\textup{Third} &2&4&4&4&4&4&4&4&6&\quad 12\\\hline\hline
\end{array}
\eeq 
\end{proposition}
The proof of this proposition is in Appendix \ref{degreeP5}.
\medskip

In order to further restrict this set of solutions, we recall Eq. \eqref{eq:isodihedral}: here $\mathfrak{I}(K_5')$ is isomorphic to the dihedral group $D_6$. Appendix \eqref{dihrep1} confirms the four one-dimensional representations we found in \eqref{4_1rep5}. Moreover, Appendix \eqref{dihrep2} shows that $\mathfrak{I}(K_5')
$ has $\frac{p-1}{2}=2$ irreps of degree $2$; hence, $G_5$ has at least these two 2-irreps. This reduces the possibilities of Proposition \ref{prop:sum296} to the following solutions.
\beq\notag  
\begin{array}{lccccccccccc} 
\hline\hline
&\quad d_5\quad &\quad d_6\quad &\quad d_7\quad &\quad d_8\quad &\quad d_9\quad &\quad d_{10}\quad &\quad d_{11}\quad &\quad d_{12}\quad &\quad d_{13}\quad &\quad d_{14}\\ 
\textup{First}&2&2&3&3&3&3&6&6&6&\quad 12\\ 
\textup{Second}\quad &2&2&6&6&6&6&6&6&6&\quad 6\\ \hline\hline
\end{array}
\eeq 
In any case, we conclude that $G_5$ leads exactly to two qubit representations: those coming from the two unitary $2$-irreps of the dihedral group $D_6$.
\subsection{$5$-Adic qubits}
Similarly to Proposition \ref{2irrep2}, we have the following.
\begin{proposition}
 There exist two unitary 2-irreps $J_5^{(1)},\,J_5^{(2)}$ of $SO(3)_5$,
 \beq 
  J_5^{(i)}:= \rho_i\circ\varphi\circ K_5'\circ\pi_1,
 \eeq
 where $\varphi$ represents isomorphism \eqref{eq:isodihedral} and $\rho_i,\ i=1,2$ are the two $2$-irreps of $D_6$ (Appendix \ref{dihrep2}),
 \begin{align*}
 & \rho_1\,:\quad a \mapsto\begin{pmatrix}
        \frac{1}{2}&-\frac{\sqrt{3}}{2}\\\frac{\sqrt{3}}{2}&\frac{1}{2}\end{pmatrix},\quad
 x\mapsto
 \begin{pmatrix}
        1 & 0\\
        0 & -1
       \end{pmatrix},\\
& \rho_2\,:\quad a \mapsto\begin{pmatrix}
        -\frac{1}{2}&-\frac{\sqrt{3}}{2}\\\frac{\sqrt{3}}{2}&-\frac{1}{2}\end{pmatrix},\quad
 x\mapsto
 \begin{pmatrix}
        1 & 0\\
        0 & -1
       \end{pmatrix}.
 \end{align*} 
\end{proposition}
\begin{proof}
$J_5^{(i)}$ for $i=1,2$ are group homomorphisms from $SO(3)_5$ to $GL(\CC^2)$ because they are compositions of group homomorphisms. They are irreducible since the $\rho_i$ are so: if $Y\subset\CC^2$, then the following are equivalent:
\begin{align}
    &\left(J_5^{(i)}(SO(3)_5)\right)Y=Y,\\
    &\rho_i\big((\varphi\circ K_5'\circ\pi_1)(SO(3)_5)\big)Y=Y,\\
    &\rho_i(D_6)Y=Y.
\end{align}
The latter is equivalent to $Y=\{0\}$ or $Y=\CC^2$.
\end{proof}

To write them explicitly, it is enough to tell the action of the generators $C,Z$ of $\mathfrak{I}(K_5')
$. Choosing $(a_0,b_0)=(-2,1)$ as the generator of the solutions $(a,b)$ to $a^2+2b^2\equiv1\mod5$ and writing $5$-adic numbers as sequences of integer numbers modulo $5^k$, we get
\beq \label{J51action}
\begin{aligned}
J_5^{(1)}:\quad &\begin{pmatrix}
  (-2,a_{112},\ldots) & (-2,a_{122},\ldots) & (0,a_{132},\ldots)\\
  (1,a_{212},\ldots) & (-2,a_{222},\ldots) & (0,a_{232},\ldots)\\
  (a_{311},a_{312},\ldots) & (a_{321},a_{322},\ldots) & (1,a_{332},\ldots)
 \end{pmatrix}_{\ZZ_5}
\stackrel{K_5'\circ\pi_1}{\longmapsto}
 C=\begin{pmatrix}
  -2 & -2\\
  1 & -2
 \end{pmatrix}_{\ZZ/5}
\stackrel{\rho_1\circ\varphi}{\longmapsto}\begin{pmatrix}
        \frac{1}{2}&-\frac{\sqrt{3}}{2}\\\frac{\sqrt{3}}{2}&\frac{1}{2}\end{pmatrix}_\CC\\
        &\begin{pmatrix}
  (1,a_{112},\ldots) & (0,a_{122},\ldots) & (0,a_{132},\ldots)\\
  (0,a_{212},\ldots) & (-1,a_{222},\ldots) & (0,a_{232},\ldots)\\
  (a_{311},a_{312},\ldots) & (a_{321},a_{322},\ldots) & (-1,a_{332},\ldots)
 \end{pmatrix}_{\ZZ_5}
\stackrel{K_5'\circ\pi_1}{\longmapsto} Z = \begin{pmatrix}
       1&0\\0&-1
\end{pmatrix}_{\ZZ/5}
\stackrel{\rho_1\circ\varphi}{\longmapsto}\begin{pmatrix}
       1&0\\0&-1
\end{pmatrix}_\CC
\end{aligned}\eeq
and 
\beq \label{J52action}
\begin{aligned}
J_5^{(2)}:\quad &\begin{pmatrix}
  (-2,a_{112},\ldots) & (-2,a_{122},\ldots) & (0,a_{132},\ldots)\\
  (1,a_{212},\ldots) & (-2,a_{222},\ldots) & (0,a_{232},\ldots)\\
  (a_{311},a_{312},\ldots) & (a_{321},a_{322},\ldots) & (1,a_{332},\ldots)
 \end{pmatrix}_{\ZZ_5}
\stackrel{K_5'\circ\pi_1}{\longmapsto}
 C=\begin{pmatrix}
  -2 & -2\\
  1 & -2
 \end{pmatrix}_{\ZZ/5}
\stackrel{\rho_2\circ\varphi}{\longmapsto}\begin{pmatrix}
        -\frac{1}{2}&-\frac{\sqrt{3}}{2}\\\frac{\sqrt{3}}{2}&-\frac{1}{2}\end{pmatrix}_\CC\\
        &\begin{pmatrix}
  (1,a_{112},\ldots) & (0,a_{122},\ldots) & (0,a_{132},\ldots)\\
  (0,a_{212},\ldots) & (-1,a_{222},\ldots) & (0,a_{232},\ldots)\\
  (a_{311},a_{312},\ldots) & (a_{321},a_{322},\ldots) & (-1,a_{332},\ldots)
 \end{pmatrix}_{\ZZ_5}
\stackrel{K_5'\circ\pi_1}{\longmapsto} Z = \begin{pmatrix}
       1&0\\0&-1
\end{pmatrix}_{\ZZ/5}
\stackrel{\rho_2\circ\varphi}{\longmapsto}\begin{pmatrix}
       1&0\\0&-1
\end{pmatrix}_\CC.
\end{aligned}\eeq
The matrix $\begin{pmatrix}       \frac{1}{2}&-\frac{\sqrt{3}}{2}\\\frac{\sqrt{3}}{2}&\frac{1}{2}\end{pmatrix}$ is of order $6$ as $C$, and $\rho_1$ is a faithful representation of $D_6\simeq \mathfrak{I}(K_5')
$. On the other hand, $\begin{pmatrix} -\frac{1}{2}&-\frac{\sqrt{3}}{2}\\\frac{\sqrt{3}}{2}&-\frac{1}{2}\end{pmatrix}$ is of order $3$, so $\rho_2$ is not injective.

We, thus, arrive to the following final result.
\begin{proposition}
 The pairs $(\CC^2,J_5^{(1)})$, $(\CC^2,J_5^{(2)})$, with $J_5^{(i)}$ defined as in \eqref{J51action} and \eqref{J52action}, are 5-adic qubits, i.e., continuous unitary irreducible linear (and then projective) representations of degree $2$ on $\CC$ of the $5$-adic special orthogonal group $SO(3)_5$ defined in \eqref{SO3p}.
\end{proposition}

Note that, in contrast to the case $p=3$, $\imath':\ZZ/5\hookrightarrow \CC,\ 0_{\ZZ/5}\mapsto 0_\CC,\,\pm1_{\ZZ/5}\mapsto \pm1_\CC,\,\pm2_{\ZZ/5}\mapsto \pm2_\CC$ does not lead to a group homomorphism $\imath'\circ K_5'$; the four solutions $(a,b)\in\{(\pm2,1),(\pm2,-1)\}\subset(\ZZ/5)^2$ of $a^2+2b^2\equiv1\mod5$ are not embedded by $\imath'$ to solutions of $a^2+2b^2=1$ in $\CC$. 

\section{Representations of $SO(3)_2\mod 2$}
\label{sec:2-mod-2}
We exploit Eqs. \eqref{eq:scrttGxyz2k} and \eqref{eq:secproblematico2} 
for $k=1$. Since $\cR_{\boldsymbol{n}}(\infty)\equiv \1
\mod2$, we just have
\beq 
G_{\boldsymbol{n},2}=\left\{\cR_{\boldsymbol{n}}\left(\sigma\right)\mod 2\,:\, \sigma\in\ZZ/2\right\}.
\eeq 
The elements of $G_{\boldsymbol{n},2}$ written on the plane orthogonal to $\boldsymbol{n}$ are 
\beq 
\begin{pmatrix}
       1&0\\0&1
\end{pmatrix}\ \ \ \text{and}\ \ \ \begin{pmatrix}
       0&1\\1&0
\end{pmatrix},
\eeq 
which are the identity and the exchange of the two reference axes orthogonal to $\boldsymbol{n}$. This holds for every $\boldsymbol{n}$ of the canonical basis of $\QQ_2^3$, for which the groups $G_{\boldsymbol{n},2}$ provide the matrices 
\beq
\begin{pmatrix}
       1&0&0\\0&1&0\\0&0&1
\end{pmatrix}, \ \ \begin{pmatrix}
1&0&0\\0&0&1\\0&1&0
\end{pmatrix}, \ \ \begin{pmatrix}
0&0&1\\0&1&0\\1&0&0
\end{pmatrix},\ \ \begin{pmatrix}
0&1&0\\1&0&0\\0&0&1
\end{pmatrix}.
\eeq 
These exchanges generate the group of permutations of the three reference axes of $\QQ_2^3$. 
We conclude that there exists an isomorphism $\phi$ between $SO(3)_2\mod 2$ and the group $\mathbb{S}_3$ of permutations of $3$ elements,
\beq \begin{aligned}
&\ \ \ \ \ \ \ \ \ \ \ \ \phi:SO(3)_2 \mod 2\ \stackrel{\sim}{\rightarrow}\ \mathbb{S}_3,\\
&\begin{pmatrix}
0&1&0\\1&0&0\\0&0&1
\end{pmatrix}\mapsto(12),\quad \begin{pmatrix}
1&0&0\\0&0&1\\0&1&0
\end{pmatrix}\mapsto (23),
\end{aligned}\eeq 
where $(12),\ (23)$ are the exchanges that generate $\mathbb{S}_3$.

This isomorphism allows us to apply the representation theory of $\mathbb{S}_3$ to 
$SO(3)_2\mod 2$. In particular, we know that $\mathbb{S}_3$ has a unique $2$-irrep, 
which we propose as the foundation for the $2$-adic qubit.

In detail \cite{repsSL}
, $\sum_\alpha d_\alpha^2=\lvert \mathbb{S}_3\rvert=6$, where $d_\alpha$ are the degrees of all the irreducible representations of $\mathbb{S}_3$ located by $\alpha$. This means that no irreducible representation of degree higher or equal than $3$ is allowed for $\mathbb{S}_3$. The commutator subgroup $[\mathbb{S}_3,\mathbb{S}_3]$ is the group of even permutations, so $\lvert\mathbb{S}_3/[\mathbb{S}_3,\mathbb{S}_3]\rvert=2$ is the number of $1$-irreps of $\mathbb{S}_3$. Hence, $\mathbb{S}_3$ has two $1$-irreps and one $2$-irrep.

The complex 1-irreps for $\mathbb{S}_3$ are  as follows:
\begin{itemize}
    \item the trivial representation
\beq \notag
\operatorname{TRIV}:\mathbb{S}_3 \rightarrow \CC^\ast, \ \ \pi \mapsto 1,
\eeq
\item the sign representation
\beq \notag
\operatorname{SIGN} : \mathbb{S}_3\rightarrow \CC^\ast, 
\ \ \ \pi \mapsto 
\text{sign}(\pi) = \left\{ \begin{aligned}&+1,\ \text{if}\ \pi \ \text{even},\\ &-1,\ \text{if}\ \pi \ \text{odd}.\end{aligned} \right.
\eeq
\end{itemize}
A complex representation of degree $3$ of $\mathbb{S}_3$ is 
\beq\notag
\tau:\mathbb{S}_3\rightarrow GL(\CC^3), \ \  \tau(\pi)(e_i)=e_{\pi(i)},
\eeq 
which permutes the basis elements $\{e_i\}_{i=1}^3$ of $\CC^3$. This representation is reducible, in fact there exists the subrepresentation $\tau'$ on 
\beq
\CC^2\simeq\text{span}\{e_1-e_2,e_2-e_3\}=\{x^1e_1+x^2e_2+x^3e_3\in\CC^3\,:\, x^1+x^2+x^3=0\}.
\eeq 
This has the following matrix form for the generators $(12)$ and $(23)$ of $\mathbb{S}_3$ with respect to the basis $\{e_1-e_2,e_2-e_3\}$:
\beq \notag (12) \longmapsto \begin{pmatrix} -1 & 1 \\ 0 & 1 \end{pmatrix}, \ \ \ \ \ (23) \longmapsto \begin{pmatrix} 1 & 0 \\ 1 & -1 \end{pmatrix}. 
\eeq
$\tau'$ can be easily shown to be the unitary irreducible representation of degree $2$ of $\mathbb{S}_3$. In fact, as in the proof of irreducibility of Proposition \ref{irrepK}, the actions of $(12)$ and $(23)$ have no common eigenvectors.

\begin{proposition}
 There exists a unitary 2-irrep $J_2$ of $SO(3)_2$,
 \beq 
  J_2:= \tau'\circ\phi\circ \pi_1.
 \eeq
 \end{proposition}
\begin{proof}
$J_2$ is a homomorphism because it is composition of homomorphisms. $J_2$ is irreducible since $\tau'$ is so: if $Y\subset\CC^2$, then the following are equivalent:
\begin{align}
&(J_2(SO(3)_2)Y=Y,\\
& \tau'\big((\phi\circ\pi_1)(SO(3)_2)\big)Y=Y,\\
&\tau'(\mathbb{S}_3)Y=Y.
\end{align}
The latter is equivalent to $Y=\{0\}\ \text{or}\ Y=\CC^2$.
\end{proof}

To write $J_2$ explicitly, it is enough to tell the action of the generators $\begin{pmatrix}
0&1&0\\1&0&0\\0&0&1
\end{pmatrix},\ \begin{pmatrix}
1&0&0\\0&0&1\\0&1&0
\end{pmatrix}$ of $SO(3)_2\mod2$,
\beq \label{J2action}
\begin{aligned}
J_2:\quad &\begin{pmatrix}
  (0,a_{112},\ldots) & (1,a_{122},\ldots) & (0,a_{132},\ldots)\\
  (1,a_{212},\ldots) & (0,a_{222},\ldots) & (0,a_{232},\ldots)\\
  (0,a_{312},\ldots) & (0,a_{322},\ldots) & (1,a_{332},\ldots)
 \end{pmatrix}_{\ZZ_2}
\stackrel{\phi\circ\pi_1}{\longmapsto}
 (12)\in\mathbb{S}_3
\stackrel{\tau'}{\longmapsto}\begin{pmatrix}
        -1&1\\0&1\end{pmatrix}_\CC\\
        &\begin{pmatrix}
  (1,a_{112},\ldots) & (0,a_{122},\ldots) & (0,a_{132},\ldots)\\
  (0,a_{212},\ldots) & (0,a_{222},\ldots) & (1,a_{232},\ldots)\\
  (0,a_{312},\ldots) & (1,a_{322},\ldots) & (0,a_{332},\ldots)
 \end{pmatrix}_{\ZZ_2}\stackrel{\phi\circ\pi_1}{\longmapsto}
 (23)\in\mathbb{S}_3
\stackrel{\tau'}{\longmapsto}\begin{pmatrix}
       1&0\\1&-1
\end{pmatrix}_\CC.
\end{aligned}\eeq

We, thus, arrive to the following final result.
\begin{proposition}
 The pair $(\CC^2,J_2)$ with $J_2$ defined as in \eqref{J2action} is a 2-adic qubit, i.e., a continuous unitary irreducible linear (and then projective) representation of degree $2$ on $\CC$ of the $2$-adic special orthogonal group $SO(3)_2$ defined in \eqref{SO3p}.
\end{proposition}


\section{Discussion} 
\label{sec:conclusion}
We have examined the $p$-adic special orthogonal group $SO(3)_p$ in three dimensions, 
outlining some of its geometric and algebraic properties with the help of modular 
arithmetics and representation theory. Our aim was to show, for each $p$, at least 
one two-dimensional continuous unitary projective irrep, thus showing the existence of
$p$-adic qubits, in the sense of the construction of a spin-$\frac{1}{2}$ 
particle through an angular momentum representation. 

As a matter of fact, for $p\geq5$, we found that there is more than one such 
$p$-adic qubit, and it remains an important open question to classify all 
two-dimensional projective irreps for every $p$, not to mention the higher-
dimensional ones. 

A special role is played by the one-dimensional irreps, which act by multiplication 
on the other irreps and which themselves form an Abelian group, isomorphic 
to the abelianization of $SO(3)_p$. It remains largely unknown, except that 
we know for odd $p>2$ that it contains the Klein group $\ZZ/2 \times\ZZ/2$. 


In any case, the present work has the potential to open new lines of investigations, as follows:
 
\begin{itemize}
 \item \textit{$p$-adic spin}: while there are works treating representations of 
       $SO(3)_p$ as a characterizing property of particles, such as Ref. \cite{varvir2010}, 
       they do not work to distinguish the behaviour of the group for different values 
       of $p$. 
 
 \item \textit{$p$-adic quantum computation}: the qubit is the simplest object in this 
       subject, and computational gates can be formulated in this framework with the 
       potential of leading to new eventually more efficient algorithms.
 
 \item \textit{$p$-adic quantum information}: stochastic processes over the field of 
       $p$-adic numbers \cite{VVZ} can be formalized for qubit, thus enlarging the 
       perspectives of quantum communication.

\end{itemize}


\section*{Acknowledgments}
The authors acknowledge fruitful discussions with Sara Di Martino and 
Michele Pigliapochi in the early stages of the present project. 
A.W. acknowledges the financial support by the Spanish
MINECO (Project Nos. FIS2016-86681-P and PID2019-107609GB-I00/AEI/10.13039/501100011033),
with the support of FEDER funds, and the Generalitat de Catalunya (Project No. 2017-SGR-1127).


\section*{Author Declarations}
\subsection*{Conflict of Interest}
The authors have no conflicts to disclose.


\section*{Data availability}
Data sharing is not applicable to this article as no new data were created or analyzed in this study.

\vfill\pagebreak
\appendix

\section{Cyclic structure of solutions of the equation $a^2-vb^2\equiv1$ in $\ZZ/p$, $p>2$ prime}
\label{sec:a2vb21modp}
Recall that $v$ is not a square modulo $p$ for every prime $p>2$.
Thus, we can construct the field extension of the Galois field $\FF_p=\ZZ/p$
\cite{galois}, with $i$ being a fixed root of $x^2\equiv v$. This is denoted as 
$\FF_p[i]$ and can be seen as $\FF_p+i\FF_p$. The Galois group 
$\text{Aut}(\FF_p[i]/\FF_p)=\{\text{automorphisms}\, \psi:\FF_p[i]\rightarrow\FF_p[i]\,:\, \psi(x)=x\ \forall x\in \FF_p\}$ 
contains only two elements: the identity and $i\mapsto -i$.

Next, $\FF_p[i]^\ast$ with the multiplicative operation is a cyclic group 
(since $\FF_p[i]$ is a finite field) of order $p^2-1$. This means there exists 
an element $\alpha+i\beta,\ \alpha,\beta\in\FF_p^\ast$ of order $p^2-1$ which 
generates the whole group:
\beq 
\FF_p[i]^\ast=\langle \alpha+i\beta\rangle =
\left\{(\alpha+i\beta)^k\,:\, k=0,\dots,p^2-2
\right\}.
\eeq
Our aim is to find all the solutions of $1\equiv a^2-vb^2\equiv(a+ib)(a-ib)\mod p$. We have
\beq 
a+ib\equiv(\alpha+i\beta)^m, \quad a-ib\equiv(\alpha-i\beta)^m
\eeq 
for some $m=0,\dots,p^2-2$ because $i\mapsto -i$ is the only nontrivial 
automorphism of the field $\FF_p[i]\supset\FF_p$. Hence, our equation becomes 
$(\alpha^2-v\beta^2)^m\equiv1\mod p$, from which we see that $m$ is a multiple 
of the order of $\alpha^2-v\beta^2$. 
\begin{proposition}
If $\alpha+i\beta$ is generator of $\FF_p[i]^\ast$, then $\alpha^2-v\beta^2$ is 
generator of $\FF_p^\ast$. 
\end{proposition}
\begin{proof}
First, we want to show that the map
\beq 
f:\mathbb{F}_p[i]\rightarrow \mathbb{F}_p[i],\ x\mapsto x^p
\eeq 
corresponds to the non trivial automorphism $i\mapsto -i$.
\\$f\in \text{Aut}(\mathbb{F}_p[i]/\mathbb{F}_p)$ because of the following conditions:
\begin{itemize}
    \item $f(x+y)\equiv(x+y)^p\equiv\sum_{j=0}^p{p\choose j}x^jy^{p-j}\equiv x^p+y^p\equiv f(x)+f(y),\ \forall x,y\in \mathbb{F}_p[i]$,\\
    since the binomial coefficients ${p\choose j}$ for $j=1,\dots,p-1$ are multiples of $p$.
    \item $f(xy)\equiv\underbrace{xy\dots xy}_{p\,\text{times}}\equiv x^py^p\equiv f(x)f(y),\ \forall x,y\in \mathbb{F}_p[i]$.
    \item $f(x)\equiv 0\Leftrightarrow x\equiv 0$.
    \item For every $x\in \mathbb{F}_p^\ast$, we have $f(x)\equiv x^p\equiv x^{p-1}x\equiv x$ since $\mathbb{F}_p^\ast$ is a multiplicative cyclic group of order $p-1$.
\end{itemize}
Furthermore, $f\neq id$; otherwise $x^p\equiv x$ for every $x\in \mathbb{F}_p[i]$ would imply $x^{p-1}\equiv 1$ for every $x\in \mathbb{F}_p[i]^\ast$, which leads to a contradiction with the fact that the order of $\mathbb{F}_p[i]^\ast$ is $p^2-1$. \\It follows that $a-ib\equiv f(a+ib)\equiv (a+ib)^p, \forall a,b\in\mathbb{F}_p$.
\par If $\alpha +i\beta $ is generator of $\mathbb{F}_p[i]^\ast$, it is an element of order $p^2-1$ and $1\equiv (\alpha +i\beta)^{p^2-1}\equiv \big[(\alpha +i\beta)^{(p+1)}\big]^{(p-1)}$. Therefore $\alpha^2-v\beta^2\equiv (\alpha+i\beta)(\alpha-i\beta)\equiv (\alpha+i\beta)^{p+1}$ is a generator of $\mathbb{F}_p^\ast$.
\end{proof}

\medskip

Since $\alpha^2-v\beta^2\in \mathbb{F}_p^\ast$ is of order $p-1$, together with $(\alpha^2-v\beta^2)^m\equiv1$, we can write $m=(p-1)n$ where $n=0,\dots,p$, and
\beq 
a+ib\equiv\left((\alpha+i\beta)^{p-1}\right)^n:\equiv(a_0+ib_0)^n.
\eeq 

We have $a_0^2-vb_0^2\equiv(\alpha+i\beta)^{p-1}(\alpha-i\beta)^{p-1}\equiv(\alpha^2-v\beta^2)^{p-1}\equiv1$. The pairs $(a,b)$ such that $a+ib\equiv (a_0+ib_0)^n$ are all distinct by varying $n$ because $(a_0+ib_0)^n\equiv(a_0+ib_0)^{n'}\Leftrightarrow1\equiv (a_0+ib_0)^{n-n'}\equiv(\alpha+i\beta)^{(p-1)(n-n')}\Leftrightarrow(n-n')(p-1)=M(p^2-1)\Leftrightarrow n-n'=M(p+1)$, which is impossible for every $M\in\ZZ_{>0}$. Furthermore, $(a_0+ib_0)^n\equiv (\alpha+i\beta)^{(p-1)n}\equiv1\Leftrightarrow (p-1)n=q(p^2-1)\Leftrightarrow n=q(p+1)$ for $q\in\ZZ$. 

This means that the solutions $(a,b)\in\FF_p^2$ of $a^2-vb^2\equiv1\mod p$ form a cyclic group of order $p+1$, endowed with the product
\beq 
(a,b)\cdot (c,d):=(ac+vbd,\,ad+bc).
\eeq 
This is because the pairs $(a,b)$ are regarded as a single number $a+ib$ and $(a+ib)(c+id)=ac+iad+ibc+i^2bd=(ac+vbd)+i(ad+bd)$. Another equivalence is 
\beq 
a+ib\leftrightarrow\begin{pmatrix}
a&vb\\b&a
\end{pmatrix}
\eeq 
that gives again
\beq 
(a,b)\cdot(c,d)=\begin{pmatrix}
a&vb\\b&a
\end{pmatrix}\begin{pmatrix}
c&vd\\d&c
\end{pmatrix} =\begin{pmatrix}
ac+vbd&v(ad+bc)\\ad+bc&ac+vbd
\end{pmatrix}.
\eeq


\section{Representations of dihedral groups of even degree}
\label{dihrep}
Consider the dihedral group $D_n$, where $n$ is even \cite{repsSL}
,
\beq
 D_n \coloneqq \langle a,x \mid a^n = x^2 = e,\, xax = a^{-1} \rangle.
\eeq

This group has $(n+6)/2$ conjugacy classes: the identity element, the element $a^{n/2}$, $(n-2)/2$ other conjugacy classes in $\langle a \rangle$, and two conjugacy classes outside $\langle a \rangle$, with representatives $x$ and $ax$. The representations we are going to show are all unitary.

\subsection{The four one-dimensional representations}
\label{dihrep1}

The commutator subgroup is $\langle a^2 \rangle$, which has index four, and the quotient group (the abelianization) is a Klein four-group. There are, thus, four one-dimensional representations:
\begin{itemize}
\item The trivial representation, sending all elements to the $1 \times 1$ matrix $(1)$.
\item The representation sending all elements in $\langle a \rangle$ to $(1)$ and all the others to $(-1)$.
\item The representation sending all elements in $\langle a^2, x \rangle$ to $(1)$ and all the others to $(-1)$.
\item The representation sending all elements in $\langle a^2, ax \rangle$ to $(1)$ and all the others to $(-1)$.
\end{itemize}

\subsection{Two-dimensional representations}
\label{dihrep2}

There are $(n-2)/2$ irreducible two-dimensional representations. The $k$th representation into real orthogonal matrices is as follows:
\begin{equation}
 a\mapsto
 \begin{pmatrix}
        \cos\frac{2\pi k}{n} & -\sin\frac{2\pi k}{n}\\
        \sin\frac{2\pi k}{n} & \cos\frac{2\pi k}{n}
       \end{pmatrix},
 \quad
 x\mapsto
 \begin{pmatrix}
        1 & 0\\
        0 & -1
       \end{pmatrix}.
\end{equation}
$k$ ranges between $1$ and $(n-2)/2$, because $k$ and $n-k$ form equivalent complex irreducible representations. Furthermore, for $k=0$ and $k=n/2$, the representation is not irreducible and can be broken down into one-dimensional representations.


\section{Conjugacy classes of $G_3$}
\label{conjclg3}
In Proposition \ref{conjclg3prop}, we need to know the number of conjugacy classes of
\[
 G_3=\pi_1(SO(3)_3)=SO(3)_3\mod3.
\]
$G_3$ is partitioned in the following nine conjugacy classes [using the form in \eqref{G3param}]:

\begin{itemize}
\item one class of order 1, composed of
\[
\begin{array}{llll}
\!\!\!\begin{pmatrix}
  1& 0& 0 \\  0& 1& 0 \\  0& 0& 1 \end{pmatrix}.
\end{array}
 \]
\item Two classes of order 4,
\[\big\{M(1,0,c,d,1)\,:\,(c,d)\in\{(0,\pm1),(\pm1,0)\}\big\}\] composed of
\[
\begin{array}{llll}
\!\!\!\begin{pmatrix}
  1& 0& 0 \\  0& 1& 0 \\  0& 1& 1 \end{pmatrix},& \!\!\!\begin{pmatrix}
  1& 0& 0 \\  0& 1& 0 \\  0&-1& 1 \end{pmatrix},& \!\!\!\begin{pmatrix}
  1& 0& 0 \\  0& 1& 0 \\  1& 0& 1 \end{pmatrix},& \!\!\!\begin{pmatrix}
  1& 0& 0 \\  0& 1& 0 \\ -1& 0& 1 \end{pmatrix},
 \end{array}
\]
\[\big\{M(1,0,c,d,1)\,:\,(c,d)\in\{(1,\pm1),(-1,\pm1)\}\big\}\] composed of
\[
\begin{array}{llll}
\!\!\!\begin{pmatrix}
  1& 0& 0 \\  0& 1& 0 \\  1& 1& 1 \end{pmatrix},& \!\!\!\begin{pmatrix}
  1& 0& 0 \\  0& 1& 0 \\  1&-1& 1 \end{pmatrix},& \!\!\!\begin{pmatrix}
  1& 0& 0 \\  0& 1& 0 \\ -1& 1& 1 \end{pmatrix},& \!\!\!\begin{pmatrix}
  1& 0& 0 \\  0& 1& 0 \\ -1&-1& 1 \end{pmatrix}.
 \end{array}
\]
\item Two classes of order 6,
\[\big\{M(a,b,c,d,-1)\,:\,(a,b)\in\{(0,\pm1)\},(c,d)\in\{(0,0),(\pm 1,\pm b)\}\big\}\] composed of
\[
\begin{array}{llll}
\!\!\!\begin{pmatrix}
  0& 1& 0 \\  1& 0& 0 \\  0& 0&-1 \end{pmatrix},& \!\!\!\begin{pmatrix}
  0& 1& 0 \\  1& 0& 0 \\  1& 1&-1 \end{pmatrix},& \!\!\!\begin{pmatrix}
  0& 1& 0 \\  1& 0& 0 \\ -1&-1&-1 \end{pmatrix},& \\ \!\!\!\begin{pmatrix}
  0&-1& 0 \\ -1& 0& 0 \\  0& 0&-1 \end{pmatrix},& \!\!\!\begin{pmatrix}
  0&-1& 0 \\ -1& 0& 0 \\  1&-1&-1 \end{pmatrix},& \!\!\!\begin{pmatrix}
  0&-1& 0 \\ -1& 0& 0 \\ -1& 1&-1 \end{pmatrix}, &
 \end{array}
\]
\[\big\{M(1,0,c,d,-1)\,:\,(c,d)\in\{(0,0),(\pm 1,0)\}\big\}\cup\big\{M(-1,0,c,d,-1)\,:\,(c,d)\in\{(0,0),(0,\pm 1)\}\big\}\] composed of
\[
\begin{array}{llll}
\!\!\!\begin{pmatrix}
  1& 0& 0 \\  0&-1& 0 \\  0& 0&-1 \end{pmatrix},& \!\!\!\begin{pmatrix}
  1& 0& 0 \\  0&-1& 0 \\  1& 0&-1 \end{pmatrix},& \!\!\!\begin{pmatrix}
  1& 0& 0 \\  0&-1& 0 \\ -1& 0&-1 \end{pmatrix},&\\ \!\!\!\begin{pmatrix}
 -1& 0& 0 \\  0& 1& 0 \\  0& 0&-1 \end{pmatrix},& \!\!\!\begin{pmatrix}
 -1& 0& 0 \\  0& 1& 0 \\  0& 1&-1 \end{pmatrix},& \!\!\!\begin{pmatrix}
 -1& 0& 0 \\  0& 1& 0 \\  0&-1&-1 \end{pmatrix}. &
 \end{array}
\]
\item One class of order 9,
\[\big\{M(-1,0,c,d,1)\,:\,c,d\in\ZZ/3\big\}\]composed of
\[
\begin{array}{llll}
\!\!\!\begin{pmatrix}
 -1& 0& 0 \\  0&-1& 0 \\  0& 0& 1 \end{pmatrix},& \!\!\!\begin{pmatrix}
 -1& 0& 0 \\  0&-1& 0 \\  0& 1& 1 \end{pmatrix},& \!\!\!\begin{pmatrix}
 -1& 0& 0 \\  0&-1& 0 \\  0&-1& 1 \end{pmatrix},&\\ \!\!\!\begin{pmatrix}
 -1& 0& 0 \\  0&-1& 0 \\  1& 0& 1 \end{pmatrix},& \!\!\!\begin{pmatrix}
 -1& 0& 0 \\  0&-1& 0 \\  1& 1& 1 \end{pmatrix},& \!\!\!\begin{pmatrix}
 -1& 0& 0 \\  0&-1& 0 \\  1&-1& 1 \end{pmatrix},&\\ \!\!\!\begin{pmatrix}
 -1& 0& 0 \\  0&-1& 0 \\ -1& 0& 1 \end{pmatrix},& \!\!\!\begin{pmatrix}
 -1& 0& 0 \\  0&-1& 0 \\ -1& 1& 1 \end{pmatrix},& \!\!\!\begin{pmatrix}
 -1& 0& 0 \\  0&-1& 0 \\ -1&-1& 1 \end{pmatrix}. &
\end{array}
\]
\item Two classes of order 12,
\[\big\{M(0,\pm1,c,d,-1)\,:\,(c,d)\in\{(0,\pm1),(\pm1,0),(\pm1,\mp b)\}\big\}\]composed of
\[
\begin{array}{llll}
\!\!\!\begin{pmatrix}
  0& 1& 0 \\  1& 0& 0 \\  0& 1&-1 \end{pmatrix},& \!\!\!\begin{pmatrix}
  0& 1& 0 \\  1& 0& 0 \\  0&-1&-1 \end{pmatrix},& \!\!\!\begin{pmatrix}
  0& 1& 0 \\  1& 0& 0 \\  1& 0&-1 \end{pmatrix},& \!\!\!\begin{pmatrix}
  0& 1& 0 \\  1& 0& 0 \\  1&-1&-1 \end{pmatrix},\\ \!\!\!\begin{pmatrix}
  0& 1& 0 \\  1& 0& 0 \\ -1& 0&-1 \end{pmatrix},& \!\!\!\begin{pmatrix}
  0& 1& 0 \\  1& 0& 0 \\ -1& 1&-1 \end{pmatrix},& \!\!\!\begin{pmatrix}
  0&-1& 0 \\ -1& 0& 0 \\  0& 1&-1 \end{pmatrix},& \!\!\!\begin{pmatrix}
  0&-1& 0 \\ -1& 0& 0 \\  0&-1&-1 \end{pmatrix},\\ \!\!\!\begin{pmatrix}
  0&-1& 0 \\ -1& 0& 0 \\  1& 0&-1 \end{pmatrix},& \!\!\!\begin{pmatrix}
  0&-1& 0 \\ -1& 0& 0 \\  1& 1&-1 \end{pmatrix},& \!\!\!\begin{pmatrix}
  0&-1& 0 \\ -1& 0& 0 \\ -1& 0&-1 \end{pmatrix},& \!\!\!\begin{pmatrix}
  0&-1& 0 \\ -1& 0& 0 \\ -1&-1&-1 \end{pmatrix},
\end{array}
\]
\[\big\{M(1,0,c,\pm1,-1)\,:\,c\in\ZZ/3\big\}\cup\big\{M(-1,0,\pm1,d,-1)\,:\,d\in\ZZ/3\big\}\]composed of
\[
\begin{array}{llll}
\!\!\!\begin{pmatrix}
  1& 0& 0 \\  0&-1& 0 \\  0& 1&-1 \end{pmatrix},& \!\!\!\begin{pmatrix}
  1& 0& 0 \\  0&-1& 0 \\  0&-1&-1 \end{pmatrix},& \!\!\!\begin{pmatrix}
  1& 0& 0 \\  0&-1& 0 \\  1& 1&-1 \end{pmatrix},& \!\!\!\begin{pmatrix}
  1& 0& 0 \\  0&-1& 0 \\  1&-1&-1 \end{pmatrix},\\ \!\!\!\begin{pmatrix}
  1& 0& 0 \\  0&-1& 0 \\ -1& 1&-1 \end{pmatrix},& \!\!\!\begin{pmatrix}
  1& 0& 0 \\  0&-1& 0 \\ -1&-1&-1 \end{pmatrix},& \!\!\!\begin{pmatrix}
 -1& 0& 0 \\  0& 1& 0 \\  1& 0&-1 \end{pmatrix},& \!\!\!\begin{pmatrix}
 -1& 0& 0 \\  0& 1& 0 \\  1& 1&-1 \end{pmatrix},\\ \!\!\!\begin{pmatrix}
 -1& 0& 0 \\  0& 1& 0 \\  1&-1&-1 \end{pmatrix},& \!\!\!\begin{pmatrix}
 -1& 0& 0 \\  0& 1& 0 \\ -1& 0&-1 \end{pmatrix},& \!\!\!\begin{pmatrix}
 -1& 0& 0 \\  0& 1& 0 \\ -1& 1&-1 \end{pmatrix},& \!\!\!\begin{pmatrix}
 -1& 0& 0 \\  0& 1& 0 \\ -1&-1&-1 \end{pmatrix}.
\end{array}
\]
\item One class of order 18,
\[\big\{M(0,\pm1,c,d,1)\,:\,c,d\in\ZZ/3\big\}\]composed of
\[
\begin{array}{llll}
\!\!\!\begin{pmatrix}
  0& 1& 0 \\ -1& 0& 0 \\  0& 0& 1 \end{pmatrix},& \!\!\!\begin{pmatrix}
  0& 1& 0 \\ -1& 0& 0 \\  0& 1& 1 \end{pmatrix},& \!\!\!\begin{pmatrix}
  0& 1& 0 \\ -1& 0& 0 \\  0&-1& 1 \end{pmatrix},& \!\!\!\begin{pmatrix}
  0& 1& 0 \\ -1& 0& 0 \\  1& 0& 1 \end{pmatrix},\\ \!\!\!\begin{pmatrix}
  0& 1& 0 \\ -1& 0& 0 \\  1& 1& 1 \end{pmatrix},& \!\!\!\begin{pmatrix}
  0& 1& 0 \\ -1& 0& 0 \\  1&-1& 1 \end{pmatrix},& \!\!\!\begin{pmatrix}
  0& 1& 0 \\ -1& 0& 0 \\ -1& 0& 1 \end{pmatrix},& \!\!\!\begin{pmatrix}
  0& 1& 0 \\ -1& 0& 0 \\ -1& 1& 1 \end{pmatrix},\\ \!\!\!\begin{pmatrix}
  0& 1& 0 \\ -1& 0& 0 \\ -1&-1& 1 \end{pmatrix},& \!\!\!\begin{pmatrix}
  0&-1& 0 \\  1& 0& 0 \\  0& 0& 1 \end{pmatrix},& \!\!\!\begin{pmatrix}
  0&-1& 0 \\  1& 0& 0 \\  0& 1& 1 \end{pmatrix},& \!\!\!\begin{pmatrix}
  0&-1& 0 \\  1& 0& 0 \\  0&-1& 1 \end{pmatrix},\\ \!\!\!\begin{pmatrix}
  0&-1& 0 \\  1& 0& 0 \\  1& 0& 1 \end{pmatrix},& \!\!\!\begin{pmatrix}
  0&-1& 0 \\  1& 0& 0 \\  1& 1& 1 \end{pmatrix},& \!\!\!\begin{pmatrix}
  0&-1& 0 \\  1& 0& 0 \\  1&-1& 1 \end{pmatrix},& \!\!\!\begin{pmatrix}
  0&-1& 0 \\  1& 0& 0 \\ -1& 0& 1 \end{pmatrix},\\ \!\!\!\begin{pmatrix}
  0&-1& 0 \\  1& 0& 0 \\ -1& 1& 1 \end{pmatrix},& \!\!\!\begin{pmatrix}
  0&-1& 0 \\  1& 0& 0 \\ -1&-1& 1 \end{pmatrix}.& &
\end{array}
\]
\end{itemize}

\section{Proof of Proposition \ref{3dregrees}}\label{degreeP3}
We need to solve $\sum_{\alpha\mid d_\alpha>1}d_\alpha^2=68$ for integers $d_\alpha$ divisors of $72$ greater than $1$ 
up to permutation, $\lvert\{\alpha \mid d_\alpha>1\}\rvert=5$. Without loss of generality, we can suppose the $d_\alpha$ in non-decreasing order. We will study case by case, starting from the largest value, $d_9$, with decreasing values. 
 
$d_9$ cannot be larger than 8, because its square would exceed 68, which is the sum of positive numbers. 

\textit{Case $d_9=8$}: this would mean that $d_5^2+d_6^2+d_7^2+d_8^2=68-64=4$, but that is not possible as $d_\alpha>1$ for any $\alpha=5,6,7,8$.

\textit{Case $d_9=6$}: assuming this, we have $d_5^2+d_6^2+d_7^2+d_8^2=32$. What can we say about $d_8$?
  \begin{itemize}
   \item $d_8=4$: then, $d_5^2+d_6^2+d_7^2=16$.
   \begin{itemize}
    \item $d_7=3$: $d_5^2+d_6^2=7$, with only the contradictory chance $d_5=d_6=2$.
    \item $d_7=2$: then, $d_5=d_6=2$, but $d_5^2+d_6^2+d_7^2$ results as 12 and not 16.
   \end{itemize}
   \item $d_8=3$: $d_5^2+d_6^2+d_7^2=23$.
   \begin{itemize}
    \item $d_7=3$: $d_5^2+d_6^2=14$.
    \begin{itemize}
     \item $d_6=3$: $d_5^2=5$, impossible.
     \item $d_6=2$: it would be $d_5^2+d_6^2=8=14$, absurd.
    \end{itemize}
    \item $d_7=2$: then, $d_5^2+d_6^2+d_7^2=12$, while it should be 23.
   \end{itemize}
   \item $d_8=2$: in this case, $d_5^2+d_6^2+d_7^2+d_8^2=16$ instead of 32.
  \end{itemize}

\textit{Case $d_9=4$}: consequently, $d_5^2+d_6^2+d_7^2+d_8^2=52$.
  \begin{itemize}
   \item $d_8=4$: $d_5^2+d_6^2+d_7^2=36$.
   \begin{itemize}
    \item $d_7=4$: $d_5^2+d_6^2=20$.
    \begin{itemize}
     \item $d_6=4$: then, $d_5^2=4\Rightarrow d_5=2$,  \underline{solution}.
     \item $d_6=3$: then, $d_5^2=11$, impossible.
     \item $d_6=2$: then, $d_5^2+d_6^2=8<16$.
    \end{itemize}
    \item $d_7=3$: then, $d_5^2+d_6^2=27$. The odd sum forces $d_6=3\Rightarrow d_5^2=18$, impossible.
    \item $d_7=2$: then, $d_5^2+d_6^2+d_7^2=12<36$.
   \end{itemize}
   \item $d_8=3$: $d_5^2+d_6^2+d_7^2$ should be 43, but $d_5^2+d_6^2+d_7^2\leq27$.
   \item $d_8=2$: then, $d_5^2+d_6^2+d_7^2+d_8^2=16<52$.
  \end{itemize}

\textit{Case $d_9=3$}: if the larger value is 3, $d_5^2+d_6^2+d_7^2+d_8^2+d_9^2$ can be at most $5\cdot9=45<68$.

\textit{Case $d_9=2$}: it is ruled out for the same reason as $d_9=3$.
\\
\\ The only solution is $d_5=2,\ d_6=d_7=d_8=d_9=4$.

\section{Conjugacy classes of $G_5$}
\label{conjclg5}
In Proposition \ref{conjclg5prop}, we need to know the number of conjugacy classes of
\[
 G_5=\pi_1(SO(3)_5)=SO(3)_5\mod5.
\]
$G_5$ is partitioned in the following $14$ conjugacy classes [referring to the form in \eqref{G5param}]:
\begin{itemize}
    \item one class of order $1$, say $C_1$, composed of
    \[ \begin{pmatrix}
           1&0&0\\0&1&0\\0&0&1
    \end{pmatrix}.
    \]
    \item Four classes of order $6$,
    \[C_2=\big\{M(1,0,c,d,1)\,:\,(c,d)\in\{(1,\pm1),(-1,\pm1),(\pm2,0)\}\big\}\]
    composed of\\
    $\begin{pmatrix}1 & 0 & 0\\0 & 1 & 0\\1 & 1 & 1\end{pmatrix}$,  $\begin{pmatrix}1 & 0 & 0\\0 & 1 & 0\\1 & -1 & 1\end{pmatrix}$,   $\begin{pmatrix}1 & 0 & 0\\0 & 1 & 0\\-1 & 1 & 1\end{pmatrix}$, $\begin{pmatrix}1 & 0 & 0\\0 & 1 & 0\\-1 & -1 & 1\end{pmatrix}$,     $\begin{pmatrix}1 & 0 & 0\\0 & 1 & 0\\2 & 0 & 1\end{pmatrix}$, $\begin{pmatrix}1 & 0 & 0\\0 & 1 & 0\\-2 & 0 & 1\end{pmatrix}$,
    \vspace{.7cm}\\
    \[C_3=\big\{M(1,0,c,d,1)\,:\,(c,d)\in\{(\pm1,0),(2,\pm2),(-2,\pm2)\}\big\}\]
    composed of\\
    $\begin{pmatrix}1 & 0 & 0\\0 & 1 & 0\\1 & 0 & 1\end{pmatrix}$, $\begin{pmatrix}1 & 0 & 0\\0 & 1 & 0\\-1 & 0 & 1\end{pmatrix}$, $\begin{pmatrix}1 & 0 & 0\\0 & 1 & 0\\2 & 2 & 1\end{pmatrix}$, $\begin{pmatrix}1 & 0 & 0\\0 & 1 & 0\\2 & -2 & 1\end{pmatrix}$, $\begin{pmatrix}1 & 0 & 0\\0 & 1 & 0\\-2 & 2 & 1\end{pmatrix}$, $\begin{pmatrix}1 & 0 & 0\\0 & 1 & 0\\-2 & -2 & 1\end{pmatrix}$,
    \vspace{.7cm}\\
    \[C_4=\big\{M(1,0,c,d,1)\,:\,(c,d)\in\{(0,\pm2),(2,\pm1),(-2,\pm1)\}\big\}\]
    composed of\\
    $\begin{pmatrix}1 & 0 & 0\\0 & 1 & 0\\0 & 2 & 1\end{pmatrix}$, $\begin{pmatrix}1 & 0 & 0\\0 & 1 & 0\\0 & -2 & 1\end{pmatrix}$, $\begin{pmatrix}1 & 0 & 0\\0 & 1 & 0\\2 & 1 & 1\end{pmatrix}$, $\begin{pmatrix}1 & 0 & 0\\0 & 1 & 0\\2 & -1 & 1\end{pmatrix}$, $\begin{pmatrix}1 & 0 & 0\\0 & 1 & 0\\-2 & 1 & 1\end{pmatrix}$, $\begin{pmatrix}1 & 0 & 0\\0 & 1 & 0\\-2 & -1 & 1\end{pmatrix}$,
    \vspace{.7cm}\\
    \[C_5=\big\{M(1,0,c,d,1)\,:\,(c,d)\in\{(0,\pm1),(1,\pm2),(-1,\pm2)\}\big\}\]
    composed of\\
    $\begin{pmatrix}1 & 0 & 0\\0 & 1 & 0\\0 & 1 & 1\end{pmatrix}$,  $\begin{pmatrix}1 & 0 & 0\\0 & 1 & 0\\0 & -1 & 1\end{pmatrix}$, $\begin{pmatrix}1 & 0 & 0\\0 & 1 & 0\\1 & 2 & 1\end{pmatrix}$, $\begin{pmatrix}1 & 0 & 0\\0 & 1 & 0\\1 & -2 & 1\end{pmatrix}$, $\begin{pmatrix}1 & 0 & 0\\0 & 1 & 0\\-1 & 2 & 1\end{pmatrix}$,  $\begin{pmatrix}1 & 0 & 0\\0 & 1 & 0\\-1 & -2 & 1\end{pmatrix}$.
    \item Two classes of order $15$,
    \[C_6=\big\{M(a,b,c,sbc,-1)\,:\,(a,b)\in\{(1,0),(2,\pm1),c\in\ZZ/5\}\big\}\]
    composed of\\
    $\begin{pmatrix}1 & 0 & 0\\0 & -1 & 0\\0 & 0 & -1\end{pmatrix}$, $\begin{pmatrix}1 & 0 & 0\\0 & -1 & 0\\2 & 0 & -1\end{pmatrix}$, $\begin{pmatrix}1 & 0 & 0\\0 & -1 & 0\\-1 & 0 & -1\end{pmatrix}$, $\begin{pmatrix}1 & 0 & 0\\0 & -1 & 0\\1 & 0 & -1\end{pmatrix}$, $\begin{pmatrix}1 & 0 & 0\\0 & -1 & 0\\-2 & 0 & -1\end{pmatrix}$,\\
    $\begin{pmatrix}2 & 2 & 0\\1 & -2 & 0\\0 & 0 & -1\end{pmatrix}$, $\begin{pmatrix}2 & 2 & 0\\1 & -2 & 0\\-2 & 2 & -1\end{pmatrix}$, $\begin{pmatrix}2 & 2 & 0\\1 & -2 & 0\\1 & -1 & -1\end{pmatrix}$,  $\begin{pmatrix}2 & 2 & 0\\1 & -2 & 0\\2 & -2 & -1\end{pmatrix}$, $\begin{pmatrix}2 & 2 & 0\\1 & -2 & 0\\-1 & 1 & -1\end{pmatrix}$,\\
    $\begin{pmatrix}2 & -2 & 0\\-1 & -2 & 0\\0 & 0 & -1\end{pmatrix}$, $\begin{pmatrix}2 & -2 & 0\\-1 & -2 & 0\\2 & 2 & -1\end{pmatrix}$, $\begin{pmatrix}2 & -2 & 0\\-1 & -2 & 0\\-1 & -1 & -1\end{pmatrix}$, $\begin{pmatrix}2 & -2 & 0\\-1 & -2 & 0\\1 & 1 & -1\end{pmatrix}$, $\begin{pmatrix}2 & -2 & 0\\-1 & -2 & 0\\-2 & -2 & -1\end{pmatrix}$,
    \vspace{.7cm}\\
    \[C_7=\big\{M(-1,0,0,d,-1)\,:\,d\in\ZZ/5\big\}\cup\big\{M(a,b,c,abc,-1)\,:\,(a,b)\in\{(-2,\pm1),c\in\ZZ/5\}\big\}\]
    composed of\\
    $\begin{pmatrix}-1 & 0 & 0\\0 & 1 & 0\\0 & 0 & -1\end{pmatrix}$, $\begin{pmatrix}-1 & 0 & 0\\0 & 1 & 0\\0 & 2 & -1\end{pmatrix}$, $\begin{pmatrix}-1 & 0 & 0\\0 & 1 & 0\\0 & -1 & -1\end{pmatrix}$, $\begin{pmatrix}-1 & 0 & 0\\0 & 1 & 0\\0 & 1 & -1\end{pmatrix}$, $\begin{pmatrix}-1 & 0 & 0\\0 & 1 & 0\\0 & -2 & -1\end{pmatrix}$, \\
    $\begin{pmatrix}-2 & 2 & 0\\1 & 2 & 0\\0 & 0 & -1\end{pmatrix}$, $\begin{pmatrix}-2 & 2 & 0\\1 & 2 & 0\\-1 & 2 & -1\end{pmatrix}$, $\begin{pmatrix}-2 & 2 & 0\\1 & 2 & 0\\-2 & -1 & -1\end{pmatrix}$, $\begin{pmatrix}-2 & 2 & 0\\1 & 2 & 0\\2 & 1 & -1\end{pmatrix}$, $\begin{pmatrix}-2 & 2 & 0\\1 & 2 & 0\\1 & -2 & -1\end{pmatrix}$\\
    $\begin{pmatrix}-2 & -2 & 0\\-1 & 2 & 0\\0 & 0 & -1\end{pmatrix}$, $\begin{pmatrix}-2 & -2 & 0\\-1 & 2 & 0\\1 & 2 & -1\end{pmatrix}$,  $\begin{pmatrix}-2 & -2 & 0\\-1 & 2 & 0\\2 & -1 & -1\end{pmatrix}$,     $\begin{pmatrix}-2 & -2 & 0\\-1 & 2 & 0\\-2 & 1 & -1\end{pmatrix}$, $\begin{pmatrix}-2 & -2 & 0\\-1 & 2 & 0\\-1 & -2 & -1\end{pmatrix}$.
    \item One class of order $25$,
    \[C_8=\big\{M(-1,0,c,d,1)\,:\,c,d\in\ZZ/5\big\}\]
    composed of\\
    $\begin{pmatrix}-1 & 0 & 0\\0 & -1 & 0\\0 & 0 & 1\end{pmatrix}$, $\begin{pmatrix}-1 & 0 & 0\\0 & -1 & 0\\-2 & 0 & 1\end{pmatrix}$, $\begin{pmatrix}-1 & 0 & 0\\0 & -1 & 0\\1 & 0 & 1\end{pmatrix}$, $\begin{pmatrix}-1 & 0 & 0\\0 & -1 & 0\\-1 & 0 & 1\end{pmatrix}$, $\begin{pmatrix}-1 & 0 & 0\\0 & -1 & 0\\2 & 0 & 1\end{pmatrix}$, $\begin{pmatrix}-1 & 0 & 0\\0 & -1 & 0\\0 & -2 & 1\end{pmatrix}$, $\begin{pmatrix}-1 & 0 & 0\\0 & -1 & 0\\-2 & -2 & 1\end{pmatrix}$, $\begin{pmatrix}-1 & 0 & 0\\0 & -1 & 0\\1 & -2 & 1\end{pmatrix}$, $\begin{pmatrix}-1 & 0 & 0\\0 & -1 & 0\\-1 & -2 & 1\end{pmatrix}$, $\begin{pmatrix}-1 & 0 & 0\\0 & -1 & 0\\2 & -2 & 1\end{pmatrix}$, $\begin{pmatrix}-1 & 0 & 0\\0 & -1 & 0\\0 & 1 & 1\end{pmatrix}$, $\begin{pmatrix}-1 & 0 & 0\\0 & -1 & 0\\-2 & 1 & 1\end{pmatrix}$, $\begin{pmatrix}-1 & 0 & 0\\0 & -1 & 0\\1 & 1 & 1\end{pmatrix}$, $\begin{pmatrix}-1 & 0 & 0\\0 & -1 & 0\\-1 & 1 & 1\end{pmatrix}$, $\begin{pmatrix}-1 & 0 & 0\\0 & -1 & 0\\2 & 1 & 1\end{pmatrix}$, $\begin{pmatrix}-1 & 0 & 0\\0 & -1 & 0\\0 & -1 & 1\end{pmatrix}$, $\begin{pmatrix}-1 & 0 & 0\\0 & -1 & 0\\-2 & -1 & 1\end{pmatrix}$, $\begin{pmatrix}-1 & 0 & 0\\0 & -1 & 0\\1 & -1 & 1\end{pmatrix}$, $\begin{pmatrix}-1 & 0 & 0\\0 & -1 & 0\\-1 & -1 & 1\end{pmatrix}$, $\begin{pmatrix}-1 & 0 & 0\\0 & -1 & 0\\2 & -1 & 1\end{pmatrix}$, $\begin{pmatrix}-1 & 0 & 0\\0 & -1 & 0\\0 & 2 & 1\end{pmatrix}$, $\begin{pmatrix}-1 & 0 & 0\\0 & -1 & 0\\-2 & 2 & 1\end{pmatrix}$, $\begin{pmatrix}-1 & 0 & 0\\0 & -1 & 0\\1 & 2 & 1\end{pmatrix}$, $\begin{pmatrix}-1 & 0 & 0\\0 & -1 & 0\\-1 & 2 & 1\end{pmatrix}$, $\begin{pmatrix}-1 & 0 & 0\\0 & -1 & 0\\2 & 2 & 1\end{pmatrix}$.
    \item Four classes of order $30$,
    \begin{align*}
    C_9=&\big\{M(-1,0,\pm2,d,-1)\,:\,d\in\ZZ/5\big\}\\&\cup\big\{M(-2,1,c,d,-1)\,:\,(c,d)\in\{(0\pm2),(\pm1,0),(\pm1,\pm1),(\pm2,\mp1),(\pm2,\mp2)\}\big\}\\&\cup\big\{M(-2,-1,c,d,-1)\,:\,(c,d)\in\{(0,\pm2),(\pm1,0),(\pm1,\mp1),(\pm2,\pm1),(\pm2,\pm2)\}\big\}\end{align*}
    composed of\\
    $\begin{pmatrix}-1 & 0 & 0\\0 & 1 & 0\\-2 & 0 & -1\end{pmatrix}$, $\begin{pmatrix}-1 & 0 & 0\\0 & 1 & 0\\2 & 0 & -1\end{pmatrix}$, $\begin{pmatrix}-1 & 0 & 0\\0 & 1 & 0\\-2 & 2 & -1\end{pmatrix}$, $\begin{pmatrix}-1 & 0 & 0\\0 & 1 & 0\\2 & 2 & -1\end{pmatrix}$, $\begin{pmatrix}-1 & 0 & 0\\0 & 1 & 0\\-2 & -1 & -1\end{pmatrix}$, $\begin{pmatrix}-1 & 0 & 0\\0 & 1 & 0\\2 & -1 & -1\end{pmatrix}$, $\begin{pmatrix}-1 & 0 & 0\\0 & 1 & 0\\-2 & 1 & -1\end{pmatrix}$, $\begin{pmatrix}-1 & 0 & 0\\0 & 1 & 0\\2 & 1 & -1\end{pmatrix}$, $\begin{pmatrix}-1 & 0 & 0\\0 & 1 & 0\\-2 & -2 & -1\end{pmatrix}$, $\begin{pmatrix}-1 & 0 & 0\\0 & 1 & 0\\2 & -2 & -1\end{pmatrix}$, $\begin{pmatrix}-2 & 2 & 0\\1 & 2 & 0\\2 & -2 & -1\end{pmatrix}$, $\begin{pmatrix}-2 & 2 & 0\\1 & 2 & 0\\-2 & 2 & -1\end{pmatrix}$, $\begin{pmatrix}-2 & 2 & 0\\1 & 2 & 0\\2 & -1 & -1\end{pmatrix}$, $\begin{pmatrix}-2 & 2 & 0\\1 & 2 & 0\\1 & 0 & -1\end{pmatrix}$, $\begin{pmatrix}-2 & 2 & 0\\1 & 2 & 0\\0 & 2 & -1\end{pmatrix}$, $\begin{pmatrix}-2 & 2 & 0\\1 & 2 & 0\\1 & 1 & -1\end{pmatrix}$, $\begin{pmatrix}-2 & 2 & 0\\1 & 2 & 0\\-1 & -1 & -1\end{pmatrix}$, $\begin{pmatrix}-2 & 2 & 0\\1 & 2 & 0\\0 & -2 & -1\end{pmatrix}$, $\begin{pmatrix}-2 & 2 & 0\\1 & 2 & 0\\-2 & 1 & -1\end{pmatrix}$, $\begin{pmatrix}-2 & 2 & 0\\1 & 2 & 0\\-1 & 0 & -1\end{pmatrix}$, $\begin{pmatrix}-2 & -2 & 0\\-1 & 2 & 0\\2 & 2 & -1\end{pmatrix}$, $\begin{pmatrix}-2 & -2 & 0\\-1 & 2 & 0\\-1 & 0 & -1\end{pmatrix}$, $\begin{pmatrix}-2 & -2 & 0\\-1 & 2 & 0\\-2 & -2 & -1\end{pmatrix}$, $\begin{pmatrix}-2 & -2 & 0\\-1 & 2 & 0\\-2 & -1 & -1\end{pmatrix}$, $\begin{pmatrix}-2 & -2 & 0\\-1 & 2 & 0\\-1 & 1 & -1\end{pmatrix}$, $\begin{pmatrix}-2 & -2 & 0\\-1 & 2 & 0\\0 & 2 & -1\end{pmatrix}$, $\begin{pmatrix}-2 & -2 & 0\\-1 & 2 & 0\\0 & -2 & -1\end{pmatrix}$, $\begin{pmatrix}-2 & -2 & 0\\-1 & 2 & 0\\1 & -1 & -1\end{pmatrix}$, $\begin{pmatrix}-2 & -2 & 0\\-1 & 2 & 0\\2 & 1 & -1\end{pmatrix}$, $\begin{pmatrix}-2 & -2 & 0\\-1 & 2 & 0\\1 & 0 & -1\end{pmatrix}$,
    \vspace{.7cm}\\
    \begin{align*}
    C_{10}=&\big\{M(1,0,c,\pm1,-1)\,:\,c\in\ZZ/5\big\}\\&\cup\big\{M(2,1,c,d,-1)\,:\,(c,d)\in\{(0\pm2),(\pm1,\pm1),(\pm1,\pm2),(\pm2,0),(\pm2,\pm1)\}\big\}&\\&\cup\big\{M(2,-1,c,d,-1)\,:\,(c,d)\in\{(0,\pm2),(\pm1,\mp1),(\pm1,\mp2),(\pm2,0),(\pm2,\mp1)\}\big\}\end{align*}composed of\\
    $\begin{pmatrix}1 & 0 & 0\\0 & -1 & 0\\0 & 1 & -1\end{pmatrix}$, $\begin{pmatrix}1 & 0 & 0\\0 & -1 & 0\\2 & 1 & -1\end{pmatrix}$, $\begin{pmatrix}1 & 0 & 0\\0 & -1 & 0\\-1 & 1 & -1\end{pmatrix}$, $\begin{pmatrix}1 & 0 & 0\\0 & -1 & 0\\1 & 1 & -1\end{pmatrix}$, $\begin{pmatrix}1 & 0 & 0\\0 & -1 & 0\\-2 & 1 & -1\end{pmatrix}$, $\begin{pmatrix}1 & 0 & 0\\0 & -1 & 0\\0 & -1 & -1\end{pmatrix}$, $\begin{pmatrix}1 & 0 & 0\\0 & -1 & 0\\2 & -1 & -1\end{pmatrix}$, $\begin{pmatrix}1 & 0 & 0\\0 & -1 & 0\\-1 & -1 & -1\end{pmatrix}$, $\begin{pmatrix}1 & 0 & 0\\0 & -1 & 0\\1 & -1 & -1\end{pmatrix}$, $\begin{pmatrix}1 & 0 & 0\\0 & -1 & 0\\-2 & -1 & -1\end{pmatrix}$, $\begin{pmatrix}2 & 2 & 0\\1 & -2 & 0\\-1 & -1 & -1\end{pmatrix}$, $\begin{pmatrix}2 & 2 & 0\\1 & -2 & 0\\1 & 1 & -1\end{pmatrix}$, $\begin{pmatrix}2 & 2 & 0\\1 & -2 & 0\\2 & 1 & -1\end{pmatrix}$, $\begin{pmatrix}2 & 2 & 0\\1 & -2 & 0\\-1 & -2 & -1\end{pmatrix}$, $\begin{pmatrix}2 & 2 & 0\\1 & -2 & 0\\2 & 0 & -1\end{pmatrix}$, $\begin{pmatrix}2 & 2 & 0\\1 & -2 & 0\\0 & -2 & -1\end{pmatrix}$, $\begin{pmatrix}2 & 2 & 0\\1 & -2 & 0\\0 & 2 & -1\end{pmatrix}$, $\begin{pmatrix}2 & 2 & 0\\1 & -2 & 0\\-2 & 0 & -1\end{pmatrix}$, $\begin{pmatrix}2 & 2 & 0\\1 & -2 & 0\\-2 & -1 & -1\end{pmatrix}$, $\begin{pmatrix}2 & 2 & 0\\1 & -2 & 0\\1 & 2 & -1\end{pmatrix}$, $\begin{pmatrix}2 & -2 & 0\\-1 & -2 & 0\\-1 & 1 & -1\end{pmatrix}$, $\begin{pmatrix}2 & -2 & 0\\-1 & -2 & 0\\1 & -1 & -1\end{pmatrix}$, $\begin{pmatrix}2 & -2 & 0\\-1 & -2 & 0\\-2 & 1 & -1\end{pmatrix}$, $\begin{pmatrix}2 & -2 & 0\\-1 & -2 & 0\\1 & -2 & -1\end{pmatrix}$, $\begin{pmatrix}2 & -2 & 0\\-1 & -2 & 0\\0 & -2 & -1\end{pmatrix}$, $\begin{pmatrix}2 & -2 & 0\\-1 & -2 & 0\\-2 & 0 & -1\end{pmatrix}$, $\begin{pmatrix}2 & -2 & 0\\-1 & -2 & 0\\0 & 2 & -1\end{pmatrix}$, $\begin{pmatrix}2 & -2 & 0\\-1 & -2 & 0\\2 & -1 & -1\end{pmatrix}$, $\begin{pmatrix}2 & -2 & 0\\-1 & -2 & 0\\-1 & 2 & -1\end{pmatrix}$, $\begin{pmatrix}2 & -2 & 0\\-1 & -2 & 0\\2 & 0 & -1\end{pmatrix}$,
    \vspace{.7cm}\\
    \begin{align*}
    C_{11}=&\big\{M(-1,0,\pm1,d,-1)\,:\, d\in\ZZ/5\big\}\\&\cup\big\{M(-2,1,c,d,-1)\,:\,(c,d)\in\{(0,\pm1),(\pm1,\mp1)(\pm1,\pm2),(\pm2,0),(\pm2,\pm2)\}\big\}\\&\cup\big\{M(-2,-1,c,d,-1)\,:\,(c,d)\in\{(0,\pm1),(\pm1,\pm1),(\pm1,\mp2),(\pm2,0),(\pm2,\mp2)\}\big\}\end{align*} composed of\\
    $\begin{pmatrix}-1 & 0 & 0\\0 & 1 & 0\\1 & 0 & -1\end{pmatrix}$, $\begin{pmatrix}-1 & 0 & 0\\0 & 1 & 0\\-1 & 0 & -1\end{pmatrix}$, $\begin{pmatrix}-1 & 0 & 0\\0 & 1 & 0\\1 & 2 & -1\end{pmatrix}$, $\begin{pmatrix}-1 & 0 & 0\\0 & 1 & 0\\-1 & 2 & -1\end{pmatrix}$, $\begin{pmatrix}-1 & 0 & 0\\0 & 1 & 0\\1 & -1 & -1\end{pmatrix}$, $\begin{pmatrix}-1 & 0 & 0\\0 & 1 & 0\\-1 & -1 & -1\end{pmatrix}$, $\begin{pmatrix}-1 & 0 & 0\\0 & 1 & 0\\1 & 1 & -1\end{pmatrix}$, $\begin{pmatrix}-1 & 0 & 0\\0 & 1 & 0\\-1 & 1 & -1\end{pmatrix}$, $\begin{pmatrix}-1 & 0 & 0\\0 & 1 & 0\\1 & -2 & -1\end{pmatrix}$, $\begin{pmatrix}-1 & 0 & 0\\0 & 1 & 0\\-1 & -2 & -1\end{pmatrix}$, $\begin{pmatrix}-2 & 2 & 0\\1 & 2 & 0\\1 & -1 & -1\end{pmatrix}$, $\begin{pmatrix}-2 & 2 & 0\\1 & 2 & 0\\-2 & -2 & -1\end{pmatrix}$, $\begin{pmatrix}-2 & 2 & 0\\1 & 2 & 0\\-1 & 1 & -1\end{pmatrix}$, $\begin{pmatrix}-2 & 2 & 0\\1 & 2 & 0\\0 & 1 & -1\end{pmatrix}$, $\begin{pmatrix}-2 & 2 & 0\\1 & 2 & 0\\-1 & -2 & -1\end{pmatrix}$, $\begin{pmatrix}-2 & 2 & 0\\1 & 2 & 0\\2 & 0 & -1\end{pmatrix}$, $\begin{pmatrix}-2 & 2 & 0\\1 & 2 & 0\\1 & 2 & -1\end{pmatrix}$, $\begin{pmatrix}-2 & 2 & 0\\1 & 2 & 0\\-2 & 0 & -1\end{pmatrix}$, $\begin{pmatrix}-2 & 2 & 0\\1 & 2 & 0\\0 & -1 & -1\end{pmatrix}$, $\begin{pmatrix}-2 & 2 & 0\\1 & 2 & 0\\2 & 2 & -1\end{pmatrix}$, $\begin{pmatrix}-2 & -2 & 0\\-1 & 2 & 0\\1 & 1 & -1\end{pmatrix}$, $\begin{pmatrix}-2 & -2 & 0\\-1 & 2 & 0\\-1 & -1 & -1\end{pmatrix}$, $\begin{pmatrix}-2 & -2 & 0\\-1 & 2 & 0\\0 & 1 & -1\end{pmatrix}$, $\begin{pmatrix}-2 & -2 & 0\\-1 & 2 & 0\\2 & -2 & -1\end{pmatrix}$, $\begin{pmatrix}-2 & -2 & 0\\-1 & 2 & 0\\1 & -2 & -1\end{pmatrix}$, $\begin{pmatrix}-2 & -2 & 0\\-1 & 2 & 0\\-2 & 0 & -1\end{pmatrix}$, $\begin{pmatrix}-2 & -2 & 0\\-1 & 2 & 0\\-1 & 2 & -1\end{pmatrix}$, $\begin{pmatrix}-2 & -2 & 0\\-1 & 2 & 0\\2 & 0 & -1\end{pmatrix}$, $\begin{pmatrix}-2 & -2 & 0\\-1 & 2 & 0\\-2 & 2 & -1\end{pmatrix}$, $\begin{pmatrix}-2 & -2 & 0\\-1 & 2 & 0\\0 & -1 & -1\end{pmatrix}$,
    \vspace{.7cm}\\
    \begin{align*}
    C_{12}=&\big\{M(1,0,c,d,-1)\,:\,(c,d)\in\{(0,\pm2),(\pm1,\pm2),(\pm1,\mp2),(\pm2,\pm2),(\pm2,\mp2)\}\big\}\\
    &\cup\big\{M(2,1,c,d,-1)\,:\,(c,d)\in\{(0,\pm1),(\pm1,0)(\pm1,\mp2),(\pm2,\mp1),(\pm2,\pm2)\}\big\}\\
    &\cup\big\{M(2,-1,c,d,-1)\,:\,(c,d)\in\{(0,\pm1),(\pm1,0),(\pm1,\pm2),(\pm2,\pm1),(\pm2,\mp2)\}\big\}
    \end{align*}composed of\\
    $\begin{pmatrix}1 & 0 & 0\\0 & -1 & 0\\0 & -2 & -1\end{pmatrix}$, $\begin{pmatrix}1 & 0 & 0\\0 & -1 & 0\\2 & -2 & -1\end{pmatrix}$, $\begin{pmatrix}1 & 0 & 0\\0 & -1 & 0\\-1 & -2 & -1\end{pmatrix}$, $\begin{pmatrix}1 & 0 & 0\\0 & -1 & 0\\1 & -2 & -1\end{pmatrix}$, $\begin{pmatrix}1 & 0 & 0\\0 & -1 & 0\\-2 & -2 & -1\end{pmatrix}$, $\begin{pmatrix}1 & 0 & 0\\0 & -1 & 0\\0 & 2 & -1\end{pmatrix}$, $\begin{pmatrix}1 & 0 & 0\\0 & -1 & 0\\2 & 2 & -1\end{pmatrix}$, $\begin{pmatrix}1 & 0 & 0\\0 & -1 & 0\\-1 & 2 & -1\end{pmatrix}$, $\begin{pmatrix}1 & 0 & 0\\0 & -1 & 0\\1 & 2 & -1\end{pmatrix}$, $\begin{pmatrix}1 & 0 & 0\\0 & -1 & 0\\-2 & 2 & -1\end{pmatrix}$, $\begin{pmatrix}2 & 2 & 0\\1 & -2 & 0\\-2 & -2 & -1\end{pmatrix}$, $\begin{pmatrix}2 & 2 & 0\\1 & -2 & 0\\2 & 2 & -1\end{pmatrix}$, $\begin{pmatrix}2 & 2 & 0\\1 & -2 & 0\\1 & 0 & -1\end{pmatrix}$, $\begin{pmatrix}2 & 2 & 0\\1 & -2 & 0\\0 & -1 & -1\end{pmatrix}$, $\begin{pmatrix}2 & 2 & 0\\1 & -2 & 0\\-1 & 2 & -1\end{pmatrix}$, $\begin{pmatrix}2 & 2 & 0\\1 & -2 & 0\\-2 & 1 & -1\end{pmatrix}$, $\begin{pmatrix}2 & 2 & 0\\1 & -2 & 0\\1 & -2 & -1\end{pmatrix}$, $\begin{pmatrix}2 & 2 & 0\\1 & -2 & 0\\2 & -1 & -1\end{pmatrix}$, $\begin{pmatrix}2 & 2 & 0\\1 & -2 & 0\\0 & 1 & -1\end{pmatrix}$, $\begin{pmatrix}2 & 2 & 0\\1 & -2 & 0\\-1 & 0 & -1\end{pmatrix}$, $\begin{pmatrix}2 & -2 & 0\\-1 & -2 & 0\\-2 & 2 & -1\end{pmatrix}$, $\begin{pmatrix}2 & -2 & 0\\-1 & -2 & 0\\2 & -2 & -1\end{pmatrix}$, $\begin{pmatrix}2 & -2 & 0\\-1 & -2 & 0\\0 & -1 & -1\end{pmatrix}$, $\begin{pmatrix}2 & -2 & 0\\-1 & -2 & 0\\-1 & 0 & -1\end{pmatrix}$, $\begin{pmatrix}2 & -2 & 0\\-1 & -2 & 0\\1 & 2 & -1\end{pmatrix}$, $\begin{pmatrix}2 & -2 & 0\\-1 & -2 & 0\\2 & 1 & -1\end{pmatrix}$, $\begin{pmatrix}2 & -2 & 0\\-1 & -2 & 0\\-1 & -2 & -1\end{pmatrix}$, $\begin{pmatrix}2 & -2 & 0\\-1 & -2 & 0\\-2 & -1 & -1\end{pmatrix}$, $\begin{pmatrix}2 & -2 & 0\\-1 & -2 & 0\\1 & 0 & -1\end{pmatrix}$, $\begin{pmatrix}2 & -2 & 0\\-1 & -2 & 0\\0 & 1 & -1\end{pmatrix}$.
    \item Two classes of order $50$,
    \[C_{13}=\big\{M(-2,\pm1,c,d,1)\,:\,c,d\in\ZZ/5\big\}\]composed of\\
    $\begin{pmatrix}-2 & 2 & 0\\-1 & -2 & 0\\0 & 0 & 1\end{pmatrix}$, $\begin{pmatrix}-2 & -2 & 0\\1 & -2 & 0\\0 & 0 & 1\end{pmatrix}$, $\begin{pmatrix}-2 & 2 & 0\\-1 & -2 & 0\\1 & -1 & 1\end{pmatrix}$, $\begin{pmatrix}-2 & -2 & 0\\1 & -2 & 0\\1 & 1 & 1\end{pmatrix}$, $\begin{pmatrix}-2 & 2 & 0\\-1 & -2 & 0\\2 & -2 & 1\end{pmatrix}$, $\begin{pmatrix}-2 & -2 & 0\\1 & -2 & 0\\2 & 2 & 1\end{pmatrix}$, $\begin{pmatrix}-2 & 2 & 0\\-1 & -2 & 0\\-2 & 2 & 1\end{pmatrix}$, $\begin{pmatrix}-2 & -2 & 0\\1 & -2 & 0\\-2 & -2 & 1\end{pmatrix}$, $\begin{pmatrix}-2 & 2 & 0\\-1 & -2 & 0\\-1 & 1 & 1\end{pmatrix}$, $\begin{pmatrix}-2 & -2 & 0\\1 & -2 & 0\\-1 & -1 & 1\end{pmatrix}$, $\begin{pmatrix}-2 & 2 & 0\\-1 & -2 & 0\\-2 & 1 & 1\end{pmatrix}$, $\begin{pmatrix}-2 & -2 & 0\\1 & -2 & 0\\2 & 1 & 1\end{pmatrix}$, $\begin{pmatrix}-2 & 2 & 0\\-1 & -2 & 0\\-1 & 0 & 1\end{pmatrix}$, $\begin{pmatrix}-2 & -2 & 0\\1 & -2 & 0\\-2 & 2 & 1\end{pmatrix}$, $\begin{pmatrix}-2 & 2 & 0\\-1 & -2 & 0\\0 & -1 & 1\end{pmatrix}$, $\begin{pmatrix}-2 & -2 & 0\\1 & -2 & 0\\-1 & -2 & 1\end{pmatrix}$, $\begin{pmatrix}-2 & 2 & 0\\-1 & -2 & 0\\1 & -2 & 1\end{pmatrix}$, $\begin{pmatrix}-2 & -2 & 0\\1 & -2 & 0\\0 & -1 & 1\end{pmatrix}$, $\begin{pmatrix}-2 & 2 & 0\\-1 & -2 & 0\\2 & 2 & 1\end{pmatrix}$, $\begin{pmatrix}-2 & -2 & 0\\1 & -2 & 0\\1 & 0 & 1\end{pmatrix}$, $\begin{pmatrix}-2 & 2 & 0\\-1 & -2 & 0\\1 & 2 & 1\end{pmatrix}$, $\begin{pmatrix}-2 & -2 & 0\\1 & -2 & 0\\-1 & 2 & 1\end{pmatrix}$, $\begin{pmatrix}-2 & 2 & 0\\-1 & -2 & 0\\2 & 1 & 1\end{pmatrix}$, $\begin{pmatrix}-2 & -2 & 0\\1 & -2 & 0\\0 & -2 & 1\end{pmatrix}$, $\begin{pmatrix}-2 & 2 & 0\\-1 & -2 & 0\\-2 & 0 & 1\end{pmatrix}$, $\begin{pmatrix}-2 & -2 & 0\\1 & -2 & 0\\1 & -1 & 1\end{pmatrix}$, $\begin{pmatrix}-2 & 2 & 0\\-1 & -2 & 0\\-1 & -1 & 1\end{pmatrix}$, $\begin{pmatrix}-2 & -2 & 0\\1 & -2 & 0\\2 & 0 & 1\end{pmatrix}$, $\begin{pmatrix}-2 & 2 & 0\\-1 & -2 & 0\\0 & -2 & 1\end{pmatrix}$, $\begin{pmatrix}-2 & -2 & 0\\1 & -2 & 0\\-2 & 1 & 1\end{pmatrix}$, $\begin{pmatrix}-2 & 2 & 0\\-1 & -2 & 0\\-1 & -2 & 1\end{pmatrix}$, $\begin{pmatrix}-2 & -2 & 0\\1 & -2 & 0\\1 & -2 & 1\end{pmatrix}$, $\begin{pmatrix}-2 & 2 & 0\\-1 & -2 & 0\\0 & 2 & 1\end{pmatrix}$, $\begin{pmatrix}-2 & -2 & 0\\1 & -2 & 0\\2 & -1 & 1\end{pmatrix}$, $\begin{pmatrix}-2 & 2 & 0\\-1 & -2 & 0\\1 & 1 & 1\end{pmatrix}$, $\begin{pmatrix}-2 & -2 & 0\\1 & -2 & 0\\-2 & 0 & 1\end{pmatrix}$, $\begin{pmatrix}-2 & 2 & 0\\-1 & -2 & 0\\2 & 0 & 1\end{pmatrix}$, $\begin{pmatrix}-2 & -2 & 0\\1 & -2 & 0\\-1 & 1 & 1\end{pmatrix}$, $\begin{pmatrix}-2 & 2 & 0\\-1 & -2 & 0\\-2 & -1 & 1\end{pmatrix}$, $\begin{pmatrix}-2 & -2 & 0\\1 & -2 & 0\\0 & 2 & 1\end{pmatrix}$, $\begin{pmatrix}-2 & 2 & 0\\-1 & -2 & 0\\2 & -1 & 1\end{pmatrix}$, $\begin{pmatrix}-2 & -2 & 0\\1 & -2 & 0\\-2 & -1 & 1\end{pmatrix}$, $\begin{pmatrix}-2 & 2 & 0\\-1 & -2 & 0\\-2 & -2 & 1\end{pmatrix}$, $\begin{pmatrix}-2 & -2 & 0\\1 & -2 & 0\\-1 & 0 & 1\end{pmatrix}$, $\begin{pmatrix}-2 & 2 & 0\\-1 & -2 & 0\\-1 & 2 & 1\end{pmatrix}$, $\begin{pmatrix}-2 & -2 & 0\\1 & -2 & 0\\0 & 1 & 1\end{pmatrix}$, $\begin{pmatrix}-2 & 2 & 0\\-1 & -2 & 0\\0 & 1 & 1\end{pmatrix}$, $\begin{pmatrix}-2 & -2 & 0\\1 & -2 & 0\\1 & 2 & 1\end{pmatrix}$, $\begin{pmatrix}-2 & 2 & 0\\-1 & -2 & 0\\1 & 0 & 1\end{pmatrix}$, $\begin{pmatrix}-2 & -2 & 0\\1 & -2 & 0\\2 & -2 & 1\end{pmatrix}$,
    \vspace{.7cm}\\
    \[C_{14}=\big\{M(2,\pm1,c,d,1)\,:\,c,d\in\ZZ/5\big\}\]composed of\\
    $\begin{pmatrix}2 & -2 & 0\\1 & 2 & 0\\0 & 0 & 1\end{pmatrix}$, $\begin{pmatrix}2 & 2 & 0\\-1 & 2 & 0\\0 & 0 & 1\end{pmatrix}$, $\begin{pmatrix}2 & -2 & 0\\1 & 2 & 0\\-1 & 1 & 1\end{pmatrix}$, $\begin{pmatrix}2 & 2 & 0\\-1 & 2 & 0\\-1 & -1 & 1\end{pmatrix}$, $\begin{pmatrix}2 & -2 & 0\\1 & 2 & 0\\-2 & 2 & 1\end{pmatrix}$, $\begin{pmatrix}2 & 2 & 0\\-1 & 2 & 0\\-2 & -2 & 1\end{pmatrix}$, $\begin{pmatrix}2 & -2 & 0\\1 & 2 & 0\\2 & -2 & 1\end{pmatrix}$, $\begin{pmatrix}2 & 2 & 0\\-1 & 2 & 0\\2 & 2 & 1\end{pmatrix}$, $\begin{pmatrix}2 & -2 & 0\\1 & 2 & 0\\1 & -1 & 1\end{pmatrix}$, $\begin{pmatrix}2 & 2 & 0\\-1 & 2 & 0\\1 & 1 & 1\end{pmatrix}$, $\begin{pmatrix}2 & -2 & 0\\1 & 2 & 0\\2 & -1 & 1\end{pmatrix}$, $\begin{pmatrix}2 & 2 & 0\\-1 & 2 & 0\\-2 & -1 & 1\end{pmatrix}$, $\begin{pmatrix}2 & -2 & 0\\1 & 2 & 0\\1 & 0 & 1\end{pmatrix}$, $\begin{pmatrix}2 & 2 & 0\\-1 & 2 & 0\\2 & -2 & 1\end{pmatrix}$, $\begin{pmatrix}2 & -2 & 0\\1 & 2 & 0\\0 & 1 & 1\end{pmatrix}$, $\begin{pmatrix}2 & 2 & 0\\-1 & 2 & 0\\1 & 2 & 1\end{pmatrix}$, $\begin{pmatrix}2 & -2 & 0\\1 & 2 & 0\\-1 & 2 & 1\end{pmatrix}$, $\begin{pmatrix}2 & 2 & 0\\-1 & 2 & 0\\0 & 1 & 1\end{pmatrix}$, $\begin{pmatrix}2 & -2 & 0\\1 & 2 & 0\\-2 & -2 & 1\end{pmatrix}$, $\begin{pmatrix}2 & 2 & 0\\-1 & 2 & 0\\-1 & 0 & 1\end{pmatrix}$, $\begin{pmatrix}2 & -2 & 0\\1 & 2 & 0\\-1 & -2 & 1\end{pmatrix}$, $\begin{pmatrix}2 & 2 & 0\\-1 & 2 & 0\\1 & -2 & 1\end{pmatrix}$, $\begin{pmatrix}2 & -2 & 0\\1 & 2 & 0\\-2 & -1 & 1\end{pmatrix}$, $\begin{pmatrix}2 & 2 & 0\\-1 & 2 & 0\\0 & 2 & 1\end{pmatrix}$, $\begin{pmatrix}2 & -2 & 0\\1 & 2 & 0\\2 & 0 & 1\end{pmatrix}$, $\begin{pmatrix}2 & 2 & 0\\-1 & 2 & 0\\-1 & 1 & 1\end{pmatrix}$, $\begin{pmatrix}2 & -2 & 0\\1 & 2 & 0\\1 & 1 & 1\end{pmatrix}$, $\begin{pmatrix}2 & 2 & 0\\-1 & 2 & 0\\-2 & 0 & 1\end{pmatrix}$, $\begin{pmatrix}2 & -2 & 0\\1 & 2 & 0\\0 & 2 & 1\end{pmatrix}$, $\begin{pmatrix}2 & 2 & 0\\-1 & 2 & 0\\2 & -1 & 1\end{pmatrix}$, $\begin{pmatrix}2 & -2 & 0\\1 & 2 & 0\\1 & 2 & 1\end{pmatrix}$, $\begin{pmatrix}2 & 2 & 0\\-1 & 2 & 0\\-1 & 2 & 1\end{pmatrix}$, $\begin{pmatrix}2 & -2 & 0\\1 & 2 & 0\\0 & -2 & 1\end{pmatrix}$, $\begin{pmatrix}2 & 2 & 0\\-1 & 2 & 0\\-2 & 1 & 1\end{pmatrix}$, $\begin{pmatrix}2 & -2 & 0\\1 & 2 & 0\\-1 & -1 & 1\end{pmatrix}$, $\begin{pmatrix}2 & 2 & 0\\-1 & 2 & 0\\2 & 0 & 1\end{pmatrix}$, $\begin{pmatrix}2 & -2 & 0\\1 & 2 & 0\\-2 & 0 & 1\end{pmatrix}$, $\begin{pmatrix}2 & 2 & 0\\-1 & 2 & 0\\1 & -1 & 1\end{pmatrix}$, $\begin{pmatrix}2 & -2 & 0\\1 & 2 & 0\\2 & 1 & 1\end{pmatrix}$, $\begin{pmatrix}2 & 2 & 0\\-1 & 2 & 0\\0 & -2 & 1\end{pmatrix}$, $\begin{pmatrix}2 & -2 & 0\\1 & 2 & 0\\-2 & 1 & 1\end{pmatrix}$, $\begin{pmatrix}2 & 2 & 0\\-1 & 2 & 0\\2 & 1 & 1\end{pmatrix}$, $\begin{pmatrix}2 & -2 & 0\\1 & 2 & 0\\2 & 2 & 1\end{pmatrix}$, $\begin{pmatrix}2 & 2 & 0\\-1 & 2 & 0\\1 & 0 & 1\end{pmatrix}$, $\begin{pmatrix}2 & -2 & 0\\1 & 2 & 0\\1 & -2 & 1\end{pmatrix}$, $\begin{pmatrix}2 & 2 & 0\\-1 & 2 & 0\\0 & -1 & 1\end{pmatrix}$, $\begin{pmatrix}2 & -2 & 0\\1 & 2 & 0\\0 & -1 & 1\end{pmatrix}$, $\begin{pmatrix}2 & 2 & 0\\-1 & 2 & 0\\-1 & -2 & 1\end{pmatrix}$, $\begin{pmatrix}2 & -2 & 0\\1 & 2 & 0\\-1 & 0 & 1\end{pmatrix}$, $\begin{pmatrix}2 & 2 & 0\\-1 & 2 & 0\\-2 & 2 & 1\end{pmatrix}$.
\end{itemize}

\section{Proof of Proposition \ref{prop:sum296}}\label{degreeP5}
To solve $\sum_{\alpha=5}^{14}d_\alpha^2=296$ for $d_\alpha\in\{2,3,4,6,12\}$ up to permutations, we proceed analogously to the Proof of Proposition \ref{3dregrees} (Appendix \ref{degreeP3}).

\par \textit{Case $d_{14}=12$}: this gives $\sum_{\alpha=5}^{13}d_\alpha^2=296-12^2=152$.
\begin{itemize}
    \item $d_{13}=12:$ $\sum_{\alpha=5}^{12}d_\alpha^2=152-12^2=8<8\cdot2^2=\min\left(\sum_{\alpha=5}^{12}d_\alpha^2\right)$, impossible.
    \item $d_{13}=6$: $\sum_{\alpha=5}^{12}d_\alpha^2=152-6^2=116$.
    \begin{itemize}
        \item $d_{12}=6$: $\sum_{\alpha=5}^{11}d_\alpha^2=116-6^2=80$.
        \begin{itemize}
            \item $d_{11}=6$: $\sum_{\alpha=5}^{10}d_\alpha^2=80-6^2=44$.
            \begin{itemize}
                \item $d_{10}=6$: $\sum_{\alpha=5}^9d_\alpha^2=44-6^2=8<5\cdot2^2$ impossible.
                \item $d_{10}=4$: $\sum_{\alpha=5}^9d_\alpha^2=44-4^2=28$.
                \begin{enumerate}
                    \item $d_9=4$: $\sum_{\alpha=5}^8d_\alpha^2=28-4^2=12<4\cdot2^2$ impossible.
                    \item $d_9=3$: $\sum_{\alpha=5}^8d_\alpha^2=28-3^2=19$.
                    \begin{enumerate}
                        \item $d_8=3$: $\sum_{\alpha=5}^7d_\alpha^2=19-3^2=10<3\cdot2^2$, impossible.
                        \item $d_8=2$: this forces $\sum_{\alpha=5}^8d_\alpha^2=4\cdot2^2=16=19$, absurd.
                    \end{enumerate}
                    \item $d_9=2$: this forces $\sum_{\alpha=5}^9d_\alpha^2=5\cdot2^2=20=28$, absurd.
                \end{enumerate}
                \item $d_{10}=3$: $\sum_{\alpha=5}^9d_\alpha^2=44-3^2=35$.
                \begin{enumerate}
                    \item $d_9=3$: $\sum_{\alpha=5}^8d_\alpha^2=35-3^2=26$.
                    \begin{enumerate}
                        \item $d_8=3$:  $\sum_{\alpha=5}^7d_\alpha^2=26-3^2=17$.
                        \\$d_7=3$ gives $d_5^2+d_6^2=17-3^2=8\Rightarrow d_5=d_6=2$, \underline{solution}.\\
                        $d_7=2$ forces $\sum_{\alpha=5}^7d_\alpha^2=3\cdot2^2=12=17$, absurd.
                        \item $d_8=2$: this forces $\sum_{\alpha=5}^8d_\alpha^2=4\cdot2^2=16=26$, absurd.
                    \end{enumerate}
                    \item $d_9=2$: this forces $\sum_{\alpha=5}^9d_\alpha^2=5\cdot2^2=20=35$.
                \end{enumerate}
                \item $d_{10}=2$: this forces $\sum_{\alpha=5}^{10}d_\alpha^2=6\cdot2^2=24=44$, absurd.
            \end{itemize}
            \item $d_{11}=4$: $\sum_{\alpha=5}^{10}d_\alpha^2=80-4^2=64$.
            \begin{itemize}
                \item $d_{10}=4$: $\sum_{\alpha=5}^9d_\alpha^2=64-4^2=48$.
                \begin{enumerate}
                    \item $d_9=4$: $\sum_{\alpha=5}^8d_\alpha^2=48-4^2=32$.
                    \begin{enumerate}
                        \item $d_8=4$: $\sum_{\alpha=5}^7d_\alpha^2=32-4^2=16$.\\
                        If $d_7=3,4$, then $d_5^2+d_6^2<8$ is impossible.\\
                        $d_7=2$ forces $\sum_{\alpha=5}^7d_\alpha^2=3\cdot2^2=12=16$, absurd.
                        \item $d_8=3$: $\sum_{\alpha=5}^7d_\alpha^2=32-3^2=23$.\\
                        If $d_7=3$, then $d_5^2+d_6^2=23-3^2=14$: this is impossible both for $d_6=3$ (for which $d_5^2=14-3^2=5\Rightarrow d_5\not\in\ZZ$) and for $d_6=2$ (which forces $d_5^2+d_6^2=2\cdot 2^2=8=14$).\\
                        $d_7=2$ forces $\sum_{\alpha=5}^7d_\alpha^2=3\cdot2^2=12=23$, absurd.
                        \item $d_8=2$: this forces $\sum_{\alpha=5}^8d_\alpha^2=4\cdot2^2=16=32$, absurd.
                    \end{enumerate}
                    \item $d_9=3$: $\sum_{\alpha=5}^8d_\alpha^2=48-3^2=39$.
                    \begin{enumerate}
                        \item $d_8=3$: $\sum_{\alpha=5}^7d_\alpha^2=39-3^2=30>3\cdot3^2=\max\left(\sum_{\alpha=5}^7d_\alpha^2\right)$, impossible.
                        \item $d_8=3$: this forces $\sum_{\alpha=5}^8d_\alpha^2=4\cdot2^2=16=39$, absurd.
                    \end{enumerate}
                    \item $d_9=2$: this forces $\sum_{\alpha=5}^9d_\alpha^2=5\cdot2^2=20=48$, absurd.
                \end{enumerate}
                \item $d_{10}=3$: $\sum_{\alpha=5}^9d_\alpha^2=64-9=55>5\cdot3^2=\max\left(\sum_{\alpha=5}^9d_\alpha^2\right)$, impossible.
                \item $d_{10}=2$: this forces $\sum_{\alpha=5}^{10}d_\alpha^2=6\cdot2^2=24=64$, absurd.
            \end{itemize}
            \item $d_{11}=3$: $\sum_{\alpha=5}^{10}d_\alpha^2=80-3^2=71>6\cdot3^2=\max\left(\sum_{\alpha=5}^{10}d_\alpha^2\right)$, impossible.
            \item $d_{11}=2$: this forces $\sum_{\alpha=5}^{11}d_\alpha^2=7\cdot2^2=28=80$, absurd.
        \end{itemize}
        \item $d_{12}=4$: $\sum_{\alpha=5}^{11}d_\alpha^2=116-4^2=100$.
        \begin{itemize}
            \item $d_{11}=4$: $\sum_{\alpha=5}^{10}d_\alpha^2=100-4^2=84$.
            \begin{itemize}
                \item $d_{10}=4$: $\sum_{\alpha=5}^9d_\alpha^2=84-4^2=68$.
                \begin{enumerate}
                    \item $d_9=4$: $\sum_{\alpha=5}^8d_\alpha^2=68-4^2=52$.
                    \begin{enumerate}
                        \item $d_8=4$: $\sum_{\alpha=5}^7d_\alpha^2=52-4^2=36$. It must be $d_7=4$ otherwise at most $\sum_{\alpha=5}^7d_\alpha^2=3\cdot3^2=27<36$. Then, $d_5^2+d_6^2=36-4^2=20$, which implies $d_6=4,d_5=2$, \underline{solution}.
                        \item $d_8=3$: $\sum_{\alpha=5}^7d_\alpha^2=52-3^2=43>3\cdot3^2=\max\left(\sum_{\alpha=5}^7d_\alpha^2\right)$, impossible.
                        \item $d_8=2$: this forces $\sum_{\alpha=5}^8d_\alpha^2=4\cdot2^2=16=52$, absurd.
                    \end{enumerate}
                    \item $d_9=3$: $\sum_{\alpha=5}^8d_\alpha^2=68-3^2=59>4\cdot3^2=\max\left(\sum_{\alpha=5}^8d_\alpha^2\right)$, impossible.
                    \item $d_9=2$: this forces $\sum_{\alpha=5}^9d_\alpha^2=5\cdot2^2=20=52$, absurd.
                \end{enumerate}
                \item $d_{10}=3$: $\sum_{\alpha=5}^9d_\alpha^2=84-3^2=75>5\cdot3^2=\max\left(\sum_{\alpha=5}^9d_\alpha^2\right)$, impossible.
                \item $d_{10}=2$: this forces $\sum_{\alpha=5}^{10}d_\alpha^2=24=84$, absurd.
            \end{itemize}
            \item $d_{11}=3$: $\sum_{\alpha=5}^{10}d_\alpha^2=100-3^2=91>6\cdot3^2$, impossible.
            \item $d_{11}=2$: this forces $\sum_{\alpha=5}^{11}d_\alpha^2=28=100$, absurd.
        \end{itemize}
        \item $d_{12}=3$: $\sum_{\alpha=5}^{12}d_\alpha^2=116-3^2=107>8\cdot3^2$, impossible.
        \item $d_{12}=2$: this forces $\sum_{\alpha=5}^{12}d_\alpha^2=32=112$, absurd.
    \end{itemize}
    \item $d_{13}=4$: $\sum_{\alpha=5}^{12}d_\alpha^2=152-4^2=136>8\cdot4^2$, impossible.
    \item $d_{13}=3$: $\sum_{\alpha=5}^{12}d_\alpha^2=152-3^2=143>8\cdot3^2$, impossible.
    \item $d_{13}=2$: this forces $\sum_{\alpha=5}^{13}d_\alpha^2=36=152$, absurd.
\end{itemize}

{\it Case $d_{14}=6$}: this gives $\sum_{\alpha=5}^{13}d_\alpha^2=296-6^2=260$.
\begin{itemize}
    \item $d_{13}=6$: $\sum_{\alpha=5}^{12}d_\alpha^2=260-6^2=224$.
\begin{itemize}
    \item $d_{12}=6$: $\sum_{\alpha=5}^{11}d_\alpha^2=224-6^2=188$.
    \begin{itemize}
        \item $d_{11}=6$: $\sum_{\alpha=5}^{10}d_\alpha^2=188-6^2=152$.
        \begin{itemize}
            \item $d_{10}=6$: $\sum_{\alpha=5}^9d_\alpha^2=152-6^2=116$.
            \begin{enumerate}
                \item $d_9=6$: $\sum_{\alpha=5}^8d_\alpha^2=116-6^2=80$.
                \begin{enumerate}
                    \item $d_8=6$: $\sum_{\alpha=5}^7d_\alpha^2=80-6^2=44$. \\ If $d_7=6$ then, $d_5^2+d_6^2=44-6^2=8$, which implies $d_5=d_6=2$, \underline{solution}.\\
                    If $d_7=4$ then, $d_5^2+d_6^2=44-4^2=28$ is impossible: $28<2\cdot4^2$ and $28>3^2+4^2$.\\
                    If $d_7=3$ then, $d_5^2+d_6^2=44-3^2=35>2\cdot3^2$, impossible.\\
                    $d_7=2$ forces $\sum_{\alpha=5}^7d_\alpha^2=12=44$, absurd.
                    \item $d_8=3,4$:
                    $\sum_{\alpha=5}^7d_\alpha^2 
                    >3\cdot4^2 
                    $, impossible.
                    \item $d_8=2$: this forces $\sum_{\alpha=5}^8d_\alpha^2=16=80$, absurd.
                \end{enumerate}
                \item $d_9=3,4$: $\sum_{\alpha=5}^8d_\alpha^2 
                >4\cdot4^2
                $, impossible.
                \item $d_9=2$: this forces $\sum_{\alpha=5}^9d_\alpha^2=20=116$, absurd.
            \end{enumerate}
            \item $d_{10}=3,4$: $\sum_{\alpha=5}^9d_\alpha^2>5\cdot4^2$, impossible.
            \item $d_{10}=2$: this forces $\sum_{\alpha=5}^{10}d_\alpha^2=24=152$, absurd.
        \end{itemize}
        \item $d_{11}=3,4$: $\sum_{\alpha=5}^{10}d_\alpha^2>6\cdot4^2$, impossible.
        \item $d_{11}=2$: this forces $\sum_{\alpha=5}^{11}d_\alpha^2=28=188$, absurd.
    \end{itemize}
    \item $d_{12}=3,4$: $\sum_{\alpha=5}^{11}d_\alpha^2>7\cdot4^2$, impossible.
    \item $d_{12}=2$: this forces $\sum_{\alpha=5}^{12}d_\alpha^2=32=224$, absurd.
\end{itemize}
\item $d_{13}=3,4$: $\sum_{\alpha=5}^{12}d_\alpha^2>8\cdot4^2$, impossible.
\item $d_{13}=2$: this forces $\sum_{\alpha=5}^{13}d_\alpha^2=36=260$, absurd.
\end{itemize}

{\it Cases $d_{14}=3,4$:} $\sum_{\alpha=5}^{13}d_\alpha^2>9\cdot4^2$, impossible.

{\it Case $d_{12}=2$:} this forces $\sum_{\alpha=5}^{14}d_\alpha^2=40=296$, absurd.
\\
\\ We found three distinct solutions:\\
$d_5=d_6=2,\,d_7=d_8=d_9=d_{10}=3,\,d_{11}=d_{12}=d_{13}=6,\,d_{14}=12$,\\
$d_5=d_6=2,\,d_7=d_8=d_9=d_{10}=d_{11}=d_{12}=d_{13}=d_{14}=6$, and\\
$d_5=2,\,d_6=d_7=d_8=d_9=d_{10}=d_{11}=d_{12}=4,\,d_{13}=6,\,d_{14}=12$.


\bibliographystyle{unsrt}

\begin{thebibliography}{10}

\bibitem{Weizsaecker} C. F. von Weizs\"acker, \emph{Aufbau der Physik} (Hanser Verlag, 
Munich, 1985); {\it The Structure of Physics}, Fundamental Theories of Physics Vol. 155 (Springer Verlag, Dordrecht, 2006).

\bibitem{Klein} F. Klein, 
\emph{Vergleichende Betrachtungen \"uber neuere geometrische Forschungen} (Verlag Andreas Deichert, Erlangen, 1872); ``A comparative review of recent researches in geometry,'' (transalted by M. W. Haskell), Bull. New York Math. Soc. \textbf{2}, 215-249 (1892-1893); see arXiv[math.HO]:0807.3161 [math.HO]. 

\bibitem{Serre}
J.-P. Serre, \emph{A Course in Arithmetic} (Springer Verlag, Berlin, Heidelberg, New York, 1973).

\bibitem{Vol1st} I. V. Volovich, ``$p$-adic space-time and string theory,'' Theor. Math. Phys. \textbf{71}(3), 574-576 (1987).

\bibitem{volovich10}
I. V. Volovich, ``Number theory as the ultimate physical theory,'' P-Adic Numbers, Ultrametric Anal. Appl. \textbf{2}(1), 77-87 (2010).

\bibitem{volovich89}
V. S. Vladimirov and I. V. Volovich, ``$p$-adic quantum mechanics,'' Commun. Math. Phys. \textbf{123}(4), 659-676 (1989).

\bibitem{our1st} S. Di Martino, S. Mancini, M. Pigliapochi, I. Svampa, and A. Winter, ``Geometry of the $p$-adic special orthogonal group $SO(3)_p$,'' arXiv:2104.06228 [math.NT] (2021).

\bibitem{Hilbert1}
S. Albeverio and A. Khrennikov, ``$p$-adic Hilbert space representation of quantum systems
with an infinite number of degrees of freedom,'' Int. J. Mod. Phys. B \textbf{10}(13n14), 1665-1673 (1996).

\bibitem{Hilbert2}
A. Y. Khrennikov, ``Mathematical methods of non-Archimedean physics,'' Russ. Math. Surv. \textbf{45}(4), 87-125 (1990).

\bibitem{Hilbert3}
G. K. Kalisch, ``On $p$-adic Hilbert spaces,'' Ann. Math. \textbf{48}(1), 180-192 (1947).

\bibitem{Hilbert4}
S. Albeverio, J. M. Bayod, C. Perez-Garcia, R. Cianci, and A. Khrennikov, ``Non-Archimedean analogues of orthogonal and symmetric operators and $p$-adic quantization,'' Acta Appl. Math. \textbf{57}, 205-237 (1999). 

\bibitem{Hilbert5} 
S. Albeverio, J. M. Bayod, C. Perez-Garsia, A. Y. Khrennikov, and R. Cianci, ``Non-Archimedean analogues of orthogonal and symmetric operators,'' Izv.: Math. \textbf{63}(6), 1063-1087 (1999). 

\bibitem{Lie}
P. Schneider, \emph{$p$-Adic Lie groups}, Grundlehren der mathematischen Wissenschaften Vol. 344 (Springer Verlag, Berlin, Heidelberg; New York, 2011).

\bibitem{pigliapochi-msc} M. Pigliapochi, ``A realization of quantum bit from $p$-adic rotations,'' M.S. thesis, University of Camerino, 2017.

\bibitem{Cassels} 
J. W. S. Cassels, \emph{Rational Quadratic Forms} (Courier Dover Publications, 2008).

\bibitem{simon}
B. Simon, \emph{Representations of Finite and Compact Groups}, Graduate Studies in Mathematics Vol. 10 (American Mathematical Society, RI, 1995).

\bibitem{ito}
N. It$\hat{\text{o}}$, ``On the degrees of irreducible representations of a finite group,'' Nagoya Math. J. \textbf{3}, 5-6 (1951).

\bibitem{repsSL}
J.-P. Serre and L. L. Scott, \emph{Linear representations of Finite Groups}, 
Graduate Texts in Mathematics Vol. 42 (Springer Verlag, Berlin, Heidelberg; New York, 1996).

\bibitem{varvir2010}
V. S. Varadarajan and J. T. Virtanen, 
``Structure, classification, and conformal symmetry of elementary particles over 
non-archimedean space-time,'' P-Adic Numbers, Ultrametric Anal. Appl. \textbf{2}(2), 157-174 (2010).

\bibitem{VVZ}
V. S. Vladimirov, I. V. Volovich, and E. I. Zelenov,
\emph{$P$-Adic Analysis and Mathematical Physics} (World Scientific, Singapore, 1994).

\bibitem{galois}
E. Artin, \emph{Galois Theory} (University of Notre Dame Press, London, 1971).

\end{thebibliography}


\end{document}